\documentclass[11pt]{article}

\usepackage[margin=1in]{geometry}

\usepackage{amsmath,amssymb,amsthm,mathtools,bm}
\usepackage{booktabs,array,enumitem,microtype}
\usepackage{float}
\usepackage{xcolor}
\usepackage{tikz}
\usetikzlibrary{arrows.meta,calc,matrix,positioning}
\usepackage{pgfplots}
\usepgfplotslibrary{fillbetween}
\pgfplotsset{compat=1.18}
\usepackage{authblk}

\usepackage[colorlinks=true,linkcolor=blue!55!black,citecolor=blue!55!black,urlcolor=blue!55!black]{hyperref}
\usepackage[nameinlink,capitalise,noabbrev]{cleveref}

\usepackage{braket}
\usepackage{mathrsfs}

\newtheorem{theorem}{Theorem}[section]
\newtheorem{lemma}[theorem]{Lemma}
\newtheorem{corollary}[theorem]{Corollary}
\newtheorem{proposition}[theorem]{Proposition}
\theoremstyle{definition}
\newtheorem{definition}[theorem]{Definition}
\newtheorem{remark}[theorem]{Remark}

\def\H{\mathcal{H}}
\def\D{\mathcal{D}}

\newcommand{\proj}[1]{\ket{#1} \bra{#1}}
\newcommand{\tc}{\mathrm{Tr}}
\newcommand{\U}{\mathcal{U}}
\newcommand{\cC}{\mathcal C}
\newcommand{\cM}{\mathcal M}
\newcommand{\cT}{\mathcal T}
\newcommand{\E}{\mathbb E}
\newcommand{\Prb}{\mathbb P}
\newcommand{\Tr}{\operatorname{Tr}}
\newcommand{\supp}{\operatorname{supp}}
\newcommand{\one}{\mathbf 1}
\newcommand{\ketbra}[1]{\lvert #1\rangle\!\langle #1\rvert}
\newcommand{\dtr}[2]{d\!\left(#1,#2\right)}
\newcommand{\tv}[2]{d_{\mathrm{TV}}\!\left(#1,#2\right)}

\DeclareMathOperator{\rank}{rank}
\DeclareMathOperator{\Var}{Var}
\DeclareMathOperator{\Cov}{Cov}

\title{Quantum Broadcast Channels with Mutually Confidential Messages}

\author[1,3]{Paula Belzig}
\author[1,2,4]{Sukanya Ghosal}
\author[1,4,5]{Farzin Salek\thanks{The authors are listed alphabetically by surname; this order does not indicate their relative contributions. SG and FS made the principal contributions to the research and writing of this paper.

Correspondence to FS: \href{mailto:farzin.salek@gmail.com}{farzin.salek@gmail.com}.}}
\author[1,2]{Graeme Smith}

\affil[1]{Institute for Quantum Computing, University of Waterloo,
Waterloo, ON N2L 3G1, Canada}

\affil[2]{Department of Applied Mathematics, University of Waterloo,
Waterloo, ON N2L 3G1, Canada}

\affil[3]{Department of Combinatorics and Optimization,
University of Waterloo, Waterloo, ON N2L 3G1, Canada}

\affil[4]{Perimeter Institute for Theoretical Physics,
Waterloo, ON N2L 2Y5, Canada}

\affil[5]{Dahlem Center for Complex Quantum Systems,
Freie Universit\"at Berlin, Berlin, Germany}

\begin{document}
\maketitle

\begin{abstract}
We study the transmission of two independent confidential classical messages over a quantum broadcast channel, one for each receiver. Each message must remain secret from the other receiver, including when that receiver knows its own message. For classical inputs and quantum outputs, we establish the classical Marton-type inner bound under average reliability and conditional strong secrecy. The encoder selects pairs of codewords from independently generated codebooks using normalized likelihood weights. We prove reliability through a change-of-distribution argument. Our main technical result is a bipartite classical–quantum resolvability theorem that accounts for the dependence created by pair selection and establishes secrecy for the same encoder. We also obtain a multi-letter capacity characterization and extend it to arbitrary quantum inputs under secrecy against the other receiver, together with the Stinespring environment. We recover confidential capacity regions for deterministic classical and degraded classical–quantum channels, with an explicit evaluation for the classical Blackwell channel. For coherent isometric extensions of injective deterministic
classical broadcast channels, we show that the confidential
classical capacity region equals that of the corresponding
classical channel. We then compare confidential classical communication with quantum transmission. For the coherent isometric extension of the Blackwell channel, we determine the unassisted quantum-capacity region and show that some achievable confidential classical rate pairs lie outside it. The Platypus channel provides another example of this separation.

\end{abstract}

\section{Introduction}

In a broadcast channel with confidential messages, a receiver may be
entitled to decode one message while being kept ignorant of another.
The classical broadcast channel with confidential messages was
introduced in \cite{CsiszarKorner1978}.
For two mutually confidential messages, a Marton-type inner bound was derived in \cite{LiuMaricSpasojevicYates2008}, and strong secrecy and stealth were subsequently established in \cite{BjelakovicMohammadiStanczak2016}.

Here, we consider the corresponding problem with quantum outputs. The messages are independent and classical, and each must remain secret from the unintended receiver even when that receiver knows its own message. Our main result gives a classical-quantum (cq) realization of the classical inner bound under this conditional strong-secrecy requirement.

Classical communication over quantum broadcast channels has been
studied through random coding and polar coding in both asymptotic and
one-shot settings
\cite{SavovWilde2015,RadhakrishnanSenWarsi2016,HircheMorgan2015}.
Different message configurations lead to different coding problems.
The combination of common, individualized, and confidential messages
was considered in \cite{SalekHsiehFonollosa2020}, while hierarchical
message sets and nonunique decoding for three receivers were studied
in \cite{SalekHaydenHayashi2026}.  Compound broadcast channels with
common and confidential messages were treated in
\cite{BocheJanssenSaeedinaeeni2020}.  Here, both receivers have their own
confidential message, and there is no common message. Entanglement-assisted private communication over quantum broadcast
channels was studied in \cite{QiSharmaWilde2018}.

The main difficulty in the decoding error and secrecy analyses lies in the dependence created by pair selection. Marton coding starts with independently generated marginal codebooks and selects a pair intended to reproduce a correlated joint distribution \cite{Marton1979}.  The selected indices depend on both codebooks. In particular, varying a randomization index in one book can change the selected codeword in the other.  A secrecy argument
must therefore account for the selected pair and its dependence on the full codebook.

Channel resolvability provides a general method for establishing strong
secrecy \cite{BlochLaneman2013}.  Likelihood encoding connects the
distribution induced by a random code with a target distribution
\cite{Cuff2013,SongCuffPoor2016}, and has been used for broadcast
channels with privacy-leakage constraints and with receiver cooperation
\cite{GoldfeldKramerPermuter2017,GoldfeldKramerPermuterCuff2017}.
Related random-binning methods with quantum side information also
appear in common-randomness and conference-key distillation from
multipartite quantum states \cite{SalekWinter2022,SalekWinter2025}.

We use a normalized likelihood selector for both reliability and
secrecy.  For reliability, we follow the change-of-distribution
approach of \cite{GouiaaPadakandla2026}.  Under a planted distribution,
the transmitted pair has the target joint law and competing codewords
retain the independence needed for packing.  Each receiver decodes
the complete fine index and retains only its intended message.
For secrecy, we compare the same selector with two oriented posterior
distributions.  Each orientation makes one of the correlation indices uniform and leaves a one-sided likelihood encoder for the other.  Ordinary cq soft
covering then gives a bipartite resolvability bound for the actual
selected output.  This bound is the main technical contribution of
the paper.  Together with the packing argument, it establishes
reliability and both secrecy requirements for one codebook.  The proof
uses neither pinching nor the ``Marton overcounting" technique from \cite{RadhakrishnanSenWarsi2016}.

We also bound the classical auxiliary alphabets.  The role of
cardinality bounds in evaluating classical Marton regions has been
studied using perturbation and concave-envelope methods
\cite{GohariAnantharam2012,AnantharamGohariNair2019}.
For the present region, entropy concavity gives a support reduction
that controls both averaged rate bounds.  Applying the coding theorem
to block channels then yields a regularized capacity characterization.
We derive a separate outer bound involving a quantum auxiliary.
Conditioning on such a system is generally richer than averaging over
a classical variable \cite{BeigiGohari2014}; this distinction also
appears in channel comparisons with quantum side information
\cite{HircheRouzeFranca2022}. For Hadamard quantum broadcast channels, the structure of the
degrading map supplies a classical auxiliary that enables
single-letter capacity characterizations for several communication
tasks \cite{WangDasWilde2017}.

Finally, we allow arbitrary quantum input preparations while continuing
to transmit classical messages.  In this model, secrecy is required
against the unintended receiver together with the channel environment.
We distinguish this task from quantum-information transmission, whose
relation to private classical communication has also been studied in
the one-shot setting \cite{SalekAnshuHsiehJainFonollosa2020}.
Converse bounds for entanglement generation over quantum broadcast
channels provide another approach to the quantum task
\cite{HaydenLeungRallSalek2026}.  Deterministic and degraded channels
give applications of our confidential-message theorem.  The classical
Blackwell channel, its coherent isometric extension, and the Platypus
channel illustrate the distinction between confidential classical
and quantum communication.

\section{Notation}
\label{sec:model}

All alphabets are finite, all Hilbert spaces are finite-dimensional, and all
logarithms are base two.  For states \(\rho,\sigma\), we define the \textit{trace distance}
\[
 \dtr{\rho}{\sigma}:=\frac12\lVert\rho-\sigma\rVert_1,
\]
and for probability distributions \(P,Q\), we define the \textit{total variation distance} as follows.
\[
 \tv{P}{Q}:=\frac12\sum_z|P(z)-Q(z)|.
\]
Both quantities contract under channels.  In particular, two input laws at
total variation distance \(\epsilon\) induce average cq output states at
trace distance at most \(\epsilon\).  The state of a uniform classical
register \(M\) is denoted by
\(\pi^M=|\mathcal M|^{-1}\sum_{m \in \mathcal{M}}\ketbra m\), and \([L]=\{1,\ldots,L\}\).
The binary entropy is
\(h_2(t)=-t\log t-(1-t)\log(1-t)\), with \(0\log0=0\).

For a bipartite quantum state $\rho^{AB}$, the quantum mutual information is defined as $I(A;B)_\rho = H(A)_\rho + H(B)_\rho - H(AB)_\rho$ where $H(C)_\sigma = - \tc \{\sigma \log \sigma\}$ is the von Neumann entropy of the quantum state $\sigma^C$. 

\section{Main result: Achievable rates for confidential communication over cq broadcast channels}

\subsection{Operational model}

Let
\[
 W:x\longmapsto\rho_x^{B_1B_2}
\]
be a memoryless cq broadcast channel and denote
\(\rho_{x^n}^{B_1^nB_2^n}=\bigotimes_{t=1}^n
\rho_{x_t}^{B_{1,t}B_{2,t}}\).
The communication model contains two independent confidential classical
messages and no common message, and it is illustrated in \cref{cq-fig-operational-model}.

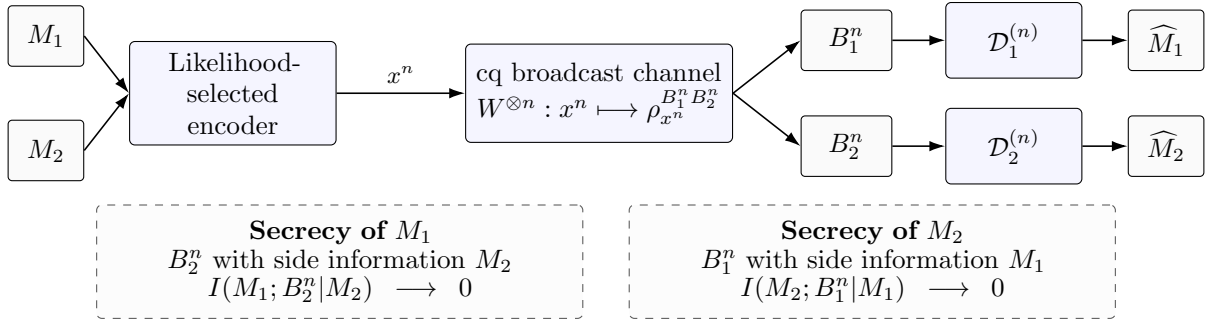
\begin{figure}[t]
\centering
\begin{tikzpicture}[
  font=\small,
  message/.style={
    draw, rounded corners=1.5pt, minimum width=10mm,
    minimum height=8mm, inner sep=2pt, fill=black!2
  },
  block/.style={
    draw, rounded corners=2pt, align=center, minimum height=13mm,
    inner sep=4pt, fill=blue!4
  },
  system/.style={
    draw, rounded corners=1.5pt, minimum width=12mm,
    minimum height=8mm, inner sep=2pt, fill=black!2
  },
  decoder/.style={
    draw, rounded corners=2pt, align=center, minimum width=18mm,
    minimum height=10mm, inner sep=3pt, fill=blue!4
  },
  wire/.style={
    -{Latex[length=2.1mm,width=1.4mm]}, semithick
  },
  secrecy/.style={
    draw=black!65, dashed, rounded corners=2pt, align=center,
    text width=6.15cm, minimum height=15mm, inner sep=4pt,
    fill=black!1
  }
]
\node[message] (m1) {$M_1$};
\node[message, below=7mm of m1] (m2) {$M_2$};
\coordinate (mcenter) at ($(m1)!0.5!(m2)$);

\node[block, right=11mm of mcenter, text width=2.45cm] (encoder)
  {Likelihood-selected encoder};
\node[block, right=17mm of encoder, text width=3.25cm] (channel)
  {cq broadcast channel\\
  $W^{\otimes n} :x^n \longmapsto\rho_{x^n}^{B^n_1B^n_2}$ };

\node[system] (b1)
  at ([xshift=15mm,yshift=7mm]channel.east) {$B_1^n$};
\node[system] (b2)
  at ([xshift=15mm,yshift=-7mm]channel.east) {$B_2^n$};

\node[decoder, right=7mm of b1] (dec1) {$\mathcal D_1^{(n)}$};
\node[decoder, right=7mm of b2] (dec2) {$\mathcal D_2^{(n)}$};
\node[message, right=6mm of dec1] (hat1) {$\widehat M_1$};
\node[message, right=6mm of dec2] (hat2) {$\widehat M_2$};

\draw[wire] (m1.east) -- (encoder.west);
\draw[wire] (m2.east) -- (encoder.west);
\draw[wire] (encoder.east) --
  node[above,font=\footnotesize] {$x^n$}
  (channel.west);
\draw[wire] (channel.east) -- (b1.west);
\draw[wire] (channel.east) -- (b2.west);
\draw[wire] (b1.east) -- (dec1.west);
\draw[wire] (b2.east) -- (dec2.west);
\draw[wire] (dec1.east) -- (hat1.west);
\draw[wire] (dec2.east) -- (hat2.west);

\coordinate (figurecenter) at ($(m1.west)!0.5!(hat2.east)$);
\node[secrecy, anchor=north east] (sec1)
  at ([xshift=-3mm,yshift=-15mm]figurecenter)
  {\textbf{Secrecy of $M_1$}\\[-0.5mm]
   $B_2^n$ with side information $M_2$\\[-0.5mm]
   $I(M_1;B_2^n|M_2)\longrightarrow0$};
\node[secrecy, anchor=north west] (sec2)
  at ([xshift=3mm,yshift=-15mm]figurecenter)
  {\textbf{Secrecy of $M_2$}\\[-0.5mm]
   $B_1^n$ with side information $M_1$\\[-0.5mm]
   $I(M_2;B_1^n|M_1)\longrightarrow0$};
\end{tikzpicture}
\caption{Confidential-message model for cq channels. The two secrecy
conditions protect $M_1$ from $B_2^n$ given $M_2$, and
$M_2$ from $B_1^n$ given $M_1$.}
\label{cq-fig-operational-model}
\end{figure}

\begin{definition}[Confidential-message code]
An \(n\)-block code consists of message sets \(\cM_1,\cM_2\), a stochastic
encoder
\[
 E_n:\cM_1\times\cM_2\longrightarrow\mathcal P(\mathcal X^n),
\]
and POVMs \(\{\Lambda^{(i)}_{m_i}:m_i\in\cM_i\}\) on \(B_i^n\),
\(i=1,2\).  For independent uniform messages, the induced state is
\begin{align}
 \Omega^{M_1M_2B_1^nB_2^n}
 :=\frac1{|\cM_1||\cM_2|}\sum_{m_1,m_2}
 \ketbra{m_1,m_2}\otimes
 \sum_{x^n}E_n(x^n|m_1,m_2)\, \rho_{x^n}^{B_1^nB_2^n}.
 \label{eq:operational-state}
\end{align}
Writing \(\rho_{m_1,m_2}^{B_i^n}\) for the corresponding conditional
marginal, receiver \(i\)'s average error is
\begin{equation}
 P_{e,i}^{(n)}
 :=1-\frac1{|\cM_1||\cM_2|}\sum_{m_1,m_2}
 \Tr\!\left[\Lambda^{(i)}_{m_i}\rho_{m_1,m_2}^{B_i^n}\right].
 \label{eq:average-error}
\end{equation}
The conditional trace-secrecy parameters are
\begin{align}
 \Delta_{1,n}
 &:=\dtr{\Omega^{M_1M_2B_2^n}}
 {\pi^{M_1}\otimes\Omega^{M_2B_2^n}},
 \label{eq:delta-one}\\
 \Delta_{2,n}
 &:=\dtr{\Omega^{M_1M_2B_1^n}}
 {\pi^{M_2}\otimes\Omega^{M_1B_1^n}}.
 \label{eq:delta-two}
\end{align}
\end{definition}

\begin{definition}[Achievability]
A pair \((R_1,R_2)\) is achievable with conditional strong secrecy if
there is a sequence of confidential-message codes such that
\[
 \liminf_{n\to\infty}\frac1n\log|\cM_i|\ge R_i,
 \qquad P_{e,i}^{(n)}\longrightarrow0,
\]
for \(i=1,2\), and
\begin{equation}
 I(M_1;B_2^n|M_2)_\Omega\longrightarrow0,
 \qquad
 I(M_2;B_1^n|M_1)_\Omega\longrightarrow0,
 \label{eq:strong-secrecy-definition}
\end{equation}
where $\Omega$ is the state from \eqref{eq:operational-state}.
\end{definition}

The \textit{confidential classical capacity region} is the closure of the set of all rate pairs achievable in the sense of the preceding definition.

The trace-distance and conditional strong-secrecy criteria give
the same capacity region. Indeed, trace secrecy implies vanishing
normalized mutual-information leakage, which suffices for the
converse in Theorem~\ref{reg-thm-exact-capacity}; its achievability
proof establishes conditional strong secrecy. Conversely,
conditional strong secrecy implies trace secrecy by Pinsker's
inequality.

Conditioning on $M_i$ in \eqref{eq:strong-secrecy-definition} ensures that the other message remains secret even if the receiver $i$ knows its own message perfectly.  Moreover, any fine index produced by that receiver's measurement is obtained from \(B_i^n\) by a qc channel. Then, the data processing inequality ensures that such measurement outcomes cannot reveal more information about the other message than \(B_i^n\) itself, and the other message remains secret even when the receiver $i$ has access to this index.

\subsection{Main result}
\label{sec:main-result}

The following theorem gives a cq realization of the classical confidential-message inner bound from \cite{LiuMaricSpasojevicYates2008}.  

Fix finite \(U,V_1,V_2\) and a distribution
\begin{equation}
 p(u,v_1,v_2,x)=p(u)p(v_1,v_2|u)p(x|u,v_1,v_2).
 \label{eq:factorization}
\end{equation}
Let
\begin{align}
 \omega^{UV_1V_2XB_1B_2}
 :=\sum_{u,v_1,v_2,x} p(u)p(v_1,v_2|u)p(x|u,v_1,v_2)
 \ketbra{u,v_1,v_2,x} \otimes\rho_x^{B_1B_2}.
 \label{eq:single-letter-state}
\end{align}
The auxiliaries \(V_1,V_2\)
describe the correlated codeword pair, while \(U\) specifies time sharing. 

\begin{theorem}[Strong-secrecy inner bound]
\label{thm:main}
For every state \(\omega\) in \eqref{eq:single-letter-state}, every
nonnegative pair satisfying
\begin{align}
 R_1&<I(V_1;B_1|U)_\omega-I(V_1;V_2|U)_\omega
       -I(V_1;B_2|V_2,U)_\omega,
 \label{eq:main-rate-one}\\
 R_2&<I(V_2;B_2|U)_\omega-I(V_1;V_2|U)_\omega
       -I(V_2;B_1|V_1,U)_\omega
 \label{eq:main-rate-two}
\end{align}
is achievable. If one rate is zero, the corresponding message set may be taken to be a singleton, and only the other inequality is required. The closure of the resulting union, over all finite auxiliaries and distributions \eqref{eq:factorization}, is an inner bound on the confidential classical capacity region. The same inner region is obtained by restricting the union to
\[
 |\mathcal U|\le 3,\qquad
 |\mathcal V_1|\le|\mathcal X|,\qquad
 |\mathcal V_2|\le|\mathcal X|.
\]
\end{theorem}

The stochastic prefix may be absorbed into the pair channel by defining
\begin{equation}
 \rho_{u,v_1,v_2}^{B_1B_2}
 :=\sum_xp(x|u,v_1,v_2)\,\rho_x^{B_1B_2}.
 \label{eq:prefixed-pair-channel}
\end{equation}
All subsequent state calculations are performed for this induced channel.
The choice \(p(x|v_1,v_2)\) is included as a special case, while any explicit
\(u\)-dependence can be absorbed into the auxiliaries \((U,V_i)\).

For simplicity of notation, we will subsequently write \(p_i=p_{V_i}\).

\begin{remark}[Classical-output specialization]
If the states \(\rho_x^{B_1B_2}\) commute in fixed product bases and are
identified with a classical broadcast law \(p(y_1,y_2|x)\),
\cref{thm:main} has the single-letter form of the confidential-message
broadcast inner bound in Theorem~4 of
\cite{LiuMaricSpasojevicYates2008}.  Here, we show an inner bound of this form for
quantum outputs with conditional strong secrecy.
\end{remark}

\section{Proof of the main result}
\label{sec:achievability}

We first prove every coding lemma for constant \(U\). The same likelihood-selected code
is used in the reliability and secrecy arguments.  We then choose the
auxiliary rates and a codebook satisfying both requirements in Section~\ref{sec:rate-selection}, restore
the time-sharing variable in Section~\ref{sec:time-sharing}, and establish the auxiliary-alphabet bounds in Section~\ref{sec:cardinality}.

Thus, fix a joint law
\(p_{V_1V_2}\), denote \(p_i=p_{V_i}\), and write
\begin{equation}
 \omega^{V_1V_2B_1B_2}
 =\sum_{v_1,v_2}p(v_1,v_2)\ketbra{v_1,v_2}
 \otimes\rho_{v_1,v_2}^{B_1B_2}.
 \label{eq:no-u-state}
\end{equation}
For channel copies, we use the memoryless extension
\begin{equation}
 \rho_{v_1^n,v_2^n}^{B_1^nB_2^n}
 :=\bigotimes_{t=1}^n\rho_{v_{1,t},v_{2,t}}^{B_{1,t}B_{2,t}},
 \label{eq:pair-channel-extension}
\end{equation}
and a superscript \(B_j^n\) denotes its corresponding marginal.  Whenever a
conditional distribution such as \(p(v_{\bar\jmath}|v_j)\) is evaluated at a
zero-marginal letter, choose it arbitrarily.  Such values never affect an
average under the stated joint law.
All mutual informations in \cref{sec:code,sec:reliability,sec:resolvability}
refer to this state.

\subsection{Likelihood-selected double-binned code}
\label{sec:code}

Besides the nonnegative message rates \(R_i\), introduce nonnegative
secrecy-randomization rates \(R_i'\) and correlation-index rates
\(\widetilde R_i\).  Let
\[
 M_i=\lfloor2^{nR_i}\rfloor,\qquad
 J_i=\lceil2^{nR_i'}\rceil,\qquad
 K_i=\lceil2^{n\widetilde R_i}\rceil.
\]
The floors and ceilings have no asymptotic rate effect.
Independently generate the marginal codebooks
\begin{equation}
 V_i^n(m_i,s_i,k_i)\overset{\mathrm{ind}}{\sim} p_i^{\otimes n},\qquad
 (m_i,s_i,k_i)\in[M_i]\times[J_i]\times[K_i],
 \label{eq:codebook-generation}
\end{equation}
for \(i=1,2\).  The realized pair of books is denoted by \(\cC\). We will refer to \((m_i,s_i,k_i)\) as each receiver's \emph{fine index}.

Define the likelihood ratio
\begin{equation}
 L_n(v_1^n,v_2^n)
 :=\frac{p_{V_1V_2}^{\otimes n}(v_1^n,v_2^n)}
 {p_1^{\otimes n}(v_1^n)p_2^{\otimes n}(v_2^n)},
 \label{eq:likelihood-ratio}
\end{equation}
with value zero off the support of the denominator, and the corresponding information-density
\begin{equation}
 \imath_n(v_1^n;v_2^n):=\log L_n(v_1^n,v_2^n),
 \qquad \log0:=-\infty.
 \label{eq:information-density}
\end{equation}
For a message pair \((m_1,m_2)\), the encoder first chooses \(S_1,S_2\)
independently and uniformly.  Conditioned on these four indices, it selects
\((K_1^\star,K_2^\star)\) according to
\begin{align}
 Q_\cC(k_1,k_2|m_1,m_2,s_1,s_2)
 :=\frac{L_n(V_1^n(m_1,s_1,k_1),V_2^n(m_2,s_2,k_2))}
 {\displaystyle\sum_{a=1}^{K_1}\sum_{b=1}^{K_2}
 L_n(V_1^n(m_1,s_1,a),V_2^n(m_2,s_2,b))}.
 \label{eq:global-selector}
\end{align}
If the denominator vanishes, \(Q_\cC\) is defined to be uniform on
\([K_1]\times[K_2]\) (\emph{universal fallback rule}).

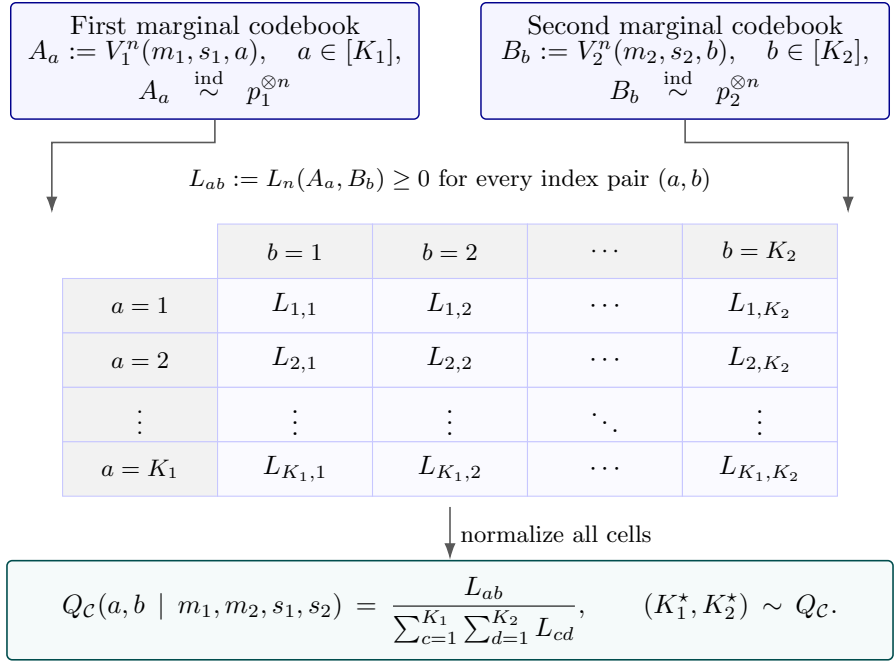
\begin{figure}[t]
\centering
\begin{tikzpicture}[
  font=\small,
  book/.style={
    draw=blue!55!black,
    fill=blue!4,
    rounded corners=2pt,
    line width=.55pt,
    align=center,
    inner sep=5pt,
    text width=5.05cm
  },
  selector/.style={
    draw=teal!60!black,
    fill=teal!4,
    rounded corners=2pt,
    line width=.6pt,
    align=center,
    inner sep=6pt,
    text width=11.3cm
  },
  flow/.style={
    -{Latex[length=2.2mm,width=1.5mm]},
    line width=.6pt,
    draw=black!65
  },
  every node/.style={text=black}
]

\node[book] (book1) {First marginal codebook\\[-1mm]
  $A_a:=V_1^n(m_1,s_1,a),\quad a\in[K_1],$\\
  $A_a\overset{\mathrm{ind}}{\sim}p_1^{\otimes n}$};

\node[book,right=8mm of book1] (book2) {Second marginal codebook\\[-1mm]
  $B_b:=V_2^n(m_2,s_2,b),\quad b\in[K_2],$\\
  $B_b\overset{\mathrm{ind}}{\sim}p_2^{\otimes n}$};

\node[
  font=\footnotesize,
  below=5mm of $(book1.south)!0.5!(book2.south)$
] (weight-title)
  {$L_{ab}:=L_n(A_a,B_b)\geq0$ for every index pair $(a,b)$};

\matrix (weights) [
  matrix of math nodes,
  below=1.5mm of weight-title,
  nodes={
    draw=blue!22,
    fill=blue!2,
    minimum width=2.05cm,
    minimum height=7.2mm,
    inner sep=1.5pt,
    anchor=center
  },
  row 1/.style={nodes={fill=black!5,font=\footnotesize}},
  column 1/.style={nodes={fill=black!5,font=\footnotesize}},
  column sep=-\pgflinewidth,
  row sep=-\pgflinewidth
] {
        & b=1         & b=2         & \cdots & b=K_2         \\
 a=1   & L_{1,1}     & L_{1,2}     & \cdots & L_{1,K_2}     \\
 a=2   & L_{2,1}     & L_{2,2}     & \cdots & L_{2,K_2}     \\
 \vdots& \vdots      & \vdots      & \ddots & \vdots        \\
 a=K_1 & L_{K_1,1}   & L_{K_1,2}   & \cdots & L_{K_1,K_2}   \\
};

\draw[flow]
  (book1.south) -- ++(0,-3mm) -| (weights.north west);
\draw[flow]
  (book2.south) -- ++(0,-3mm) -| (weights.north east);

\node[selector,below=7mm of weights] (selector) {$\displaystyle
  Q_{\cC}(a,b\mid m_1,m_2,s_1,s_2)
  =\frac{L_{ab}}
  {\sum_{c=1}^{K_1}\sum_{d=1}^{K_2}L_{cd}},
  \qquad
  (K_1^\star,K_2^\star)\sim Q_{\cC}.$};

\draw[flow]
  (weights.south)
  -- node[right,font=\footnotesize] {normalize all cells}
  (selector.north);

\end{tikzpicture}
\caption{Likelihood selection for fixed
$(m_1,m_2,s_1,s_2)$.  The two independently generated marginal
codebooks form a table of weights $L_{ab}$, and the encoder samples
$(K_1^\star,K_2^\star)$ proportionally to these weights.}
\label{fig:global-likelihood-selector}
\end{figure}

The cloud-mixing argument below compares this selected distribution with a planted distribution that is constructed to have the desired joint law \(p_{V_1V_2}^{\otimes n}\).  The encoder then draws each \(X_t\) from \(p(\cdot|V_{1,t},V_{2,t})\) (or from \(p(\cdot|u,V_{1,t},V_{2,t})\) in the time-shared construction), and sends \(X^n\).

For a fixed message pair, the average state at receiver \(j\) is
\begin{align}
 \rho_{m_1,m_2,\cC}^{B_j^n}
 =\frac1{J_1J_2}\sum_{s_1,s_2}\sum_{k_1,k_2}
 Q_\cC(k_1,k_2|m_1,m_2,s_1,s_2) \, \rho_{V_1^n(m_1,s_1,k_1),V_2^n(m_2,s_2,k_2)}^{B_j^n}.
 \label{eq:exact-induced-state}
\end{align}
The factor \(\frac1{J_1J_2}\) averages over the secrecy-randomizing indices \((s_1,s_2)\), while the correlation indices $(k_1,k_2)$ are averaged according to \(Q_\cC\).

\subsection{Reliability}
\label{sec:reliability}
Using the encoder obtained from the actual selector in \eqref{eq:global-selector}, we will prove reliability for the decoder in the following way: First, we will compare the selected distribution with a planted distribution via a cloud-mixing estimate in Lemma~\ref{lem:cloud-mixing}. Under the latter distribution, the selected pair has the joint law \(p_{V_1V_2}^{\otimes n}\). Then, each receiver can therefore use an ordinary square-root measurement for its complete fine index $ (m_i,s_i,k_i)$ to decode in Theorem~\ref{thm:reliability}.

\subsubsection{Cloud mixing}

The following lemma is a bipartite, likelihood-weighted analogue of the cloud-mixing argument from \cite{GouiaaPadakandla2026,cuff2009communication} in order to relate the law induced by likelihood-selected encoder with a planted law\footnote{The cloud mixing lemma also appears under the name ``soft-covering lemma'', see also \cite{SongCuffPoor2016} and \cite{Cuff2013}.}.

\begin{lemma}[Bipartite cloud mixing]
\label[lemma]{lem:cloud-mixing}
Independently generate two collections,
\[
 A_1^n,\ldots,A_{K_1}^n\sim p_1^{\otimes n},\qquad
 B_1^n,\ldots,B_{K_2}^n\sim p_2^{\otimes n}.
\]
Denote
\[
 L_{ab}:=L_n(A_a^n,B_b^n),\qquad
 Z_\cC:=\frac1{K_1K_2}\sum_{a,b}L_{ab}.
\]
If
\begin{equation}
 \widetilde R_1>0,\qquad \widetilde R_2>0,\qquad
 \widetilde R_1+\widetilde R_2>I(V_1;V_2),
 \label{eq:cloud-mixing-rates}
\end{equation}
then \(\E|Z_\cC-1|=o(1)\).  Consequently, for the normalized selector
\(Q_\cC(a,b)=L_{ab}/\sum_{c,d}L_{cd}\), with the stated uniform fallback,
\begin{equation}
 \E\sum_{a,b}\left|Q_\cC(a,b)-\frac{L_{ab}}{K_1K_2}\right|=o(1).
 \label{eq:selector-change-law}
\end{equation}
\end{lemma}

\begin{proof}
Let \(C=I(V_1;V_2)\), choose \(\delta>0\) such that
\(\widetilde R_1+\widetilde R_2>C+2\delta\), and truncate
\[
 \widetilde L_{ab}:=L_{ab}\one\{\log L_{ab}\le n(C+\delta)\}.
\]
Put \(\widetilde Z=(K_1K_2)^{-1}\sum_{a,b}\widetilde L_{ab}\) and
\(\mu_n=\E\widetilde L_{11}\). By definition of $L_n$, weighting the product distribution \(p_{V_1}^{\otimes n}\otimes p_{V_2}^{\otimes n}\) by $L_n$ yields \(p_{V_1V_2}^{\otimes n}\), and this change of measure gives
\[
 \mu_n=\Prb_{p_{V_1V_2}^{\otimes n}}
 \{\imath_n(V_1^n;V_2^n)\le n(C+\delta)\}\longrightarrow1.
\]
Under the target joint law \(p_{V_1V_2}^{\otimes n}\), the single-letter information density is finite
and bounded on its support. The $n-$letter information density is therefore a sum of bounded finite-alphabet random variables, and the convergence to one follows from the weak law of large numbers.

For this truncation, we have
\[
 \E \widetilde L_{11}^{2}\le 2^{n(C+\delta)}\mu_n.
\]

To prove concentration of $\widetilde Z$, note that the variables \(\widetilde L_{ab}\) form a \(K_1\times K_2\) array, where only the terms sharing a row or a column can be dependent.
For a common row, define
\(g(a^n)=\E_{B^n}\widetilde L(a^n,B^n)\).
For every \(a^n\in\supp p_1^{\otimes n}\),
\[
 0\le g(a^n)\le\E_{B^n}L(a^n,B^n)=1.
\]
Since \(A_1^n\) lies in this support almost surely,
\[
 \Cov(\widetilde L_{11},\widetilde L_{12})
 =\E g(A_1^n)^2-\mu_n^2\le1.
\]
The same bound holds for a common column.  Counting coincident index
patterns yields
\begin{equation}
 \Var(\widetilde Z)
 \le\frac{2^{n(C+\delta)}}{K_1K_2}+\frac1{K_1}+\frac1{K_2}=o(1).
 \label{eq:cloud-variance}
\end{equation}
Since \(Z_\cC\ge\widetilde Z\), \(\E Z_\cC=1\), and
\(\E\widetilde Z=\mu_n\),
\[
 \E|Z_\cC-1|
 \le 2(1-\mu_n)+\sqrt{\Var(\widetilde Z)}=o(1).
\]
When \(\sum_{a,b}L_{ab}>0\), nonnegativity gives
\[
 \sum_{a,b}\left|\frac{L_{ab}}{\sum_{c,d}L_{cd}}
 -\frac{L_{ab}}{K_1K_2}\right|=|1-Z_\cC|.
\]
If the denominator is zero, both sides equal one under the fallback rule. Taking expectations proves \eqref{eq:selector-change-law}.
\end{proof}

While the weights in \eqref{eq:selector-change-law} need not be normalized, we can define a normalized joint distribution by retaining the random-codebook ensemble. To see this,
let \(P_{\rm cb}(d\cC)\) be the product codebook-ensemble law, put
\(a_i=(m_i,s_i)\in[M_i]\times[J_i]\), and define the tilted measure
\begin{align}
 d\widetilde P(\cC,a_1,a_2,k_1,k_2)
 :=\frac{L_n(V_1^n(a_1,k_1),V_2^n(a_2,k_2))}
 {M_1J_1M_2J_2K_1K_2}\,P_{\rm cb}(d\cC).
 \label{eq:global-planted-measure}
\end{align}
For every fixed selected fine-index pair, we have that
\begin{equation}
 p_1^{\otimes n}(v_1^n)p_2^{\otimes n}(v_2^n)L_n(v_1^n,v_2^n)
 =p_{V_1V_2}^{\otimes n}(v_1^n,v_2^n).
 \label{eq:planted-identity}
\end{equation}
Using this identity, we prove that $\widetilde P$ is indeed normalized and has the desired joint distribution:
\begin{lemma}[Factorization under the planted law]
\label{lem:planted-factorization}
Let \(\mathsf L_i=[M_i]\times[J_i]\times[K_i]\) be the fine-index set and
write \(I^\star=(\ell_1^\star,\ell_2^\star)\) for the selected fine-index
pair in \eqref{eq:global-planted-measure}.  Under \(\widetilde P\),
\(I^\star\) is uniform on \(\mathsf L_1\times\mathsf L_2\).  For every fixed
\(I^\star=(\ell_1^\star,\ell_2^\star)\), write \(\boldsymbol v\) for the
codebook realization \(\{v_i^n(\ell_i):i=1,2,\ \ell_i\in\mathsf L_i\}\).
Its conditional probability mass is
\begin{align}
 &\widetilde P(\cC=\boldsymbol v
       \mid I^\star=(\ell_1^\star,\ell_2^\star))\nonumber\\
 &\quad=p_{V_1V_2}^{\otimes n}
   (v_1^n(\ell_1^\star),v_2^n(\ell_2^\star))
   \prod_{\ell_1\ne\ell_1^\star}p_1^{\otimes n}(v_1^n(\ell_1))
   \prod_{\ell_2\ne\ell_2^\star}p_2^{\otimes n}(v_2^n(\ell_2)).
 \label{eq:planted-factorization}
\end{align}
Thus the selected pair has the target joint distribution, while every
competing codeword is independent of that pair and of the other competitors,
with its original marginal distribution.
\end{lemma}

\begin{proof}
For each fixed index pair, its probability under \(\widetilde P\) is
\(1/(|\mathsf L_1||\mathsf L_2|)\), since the expectation of its likelihood
ratio under \(P_{\rm cb}\) is one.  Conditioned on that index pair, the
Radon--Nikodym derivative of \(\widetilde{P}\) with respect to \(P_{\rm cb}\) is the likelihood
ratio of the two selected codewords.  Substituting the product codebook law
and using \eqref{eq:planted-identity} gives
\eqref{eq:planted-factorization}.  This also proves that \(\widetilde P\) is
normalized.
\end{proof}

\subsubsection{A cq packing lemma}

We use the following Holevo-Schumacher-Westmoreland (HSW) packing lemma
\cite{Holevo1998,SchumacherWestmoreland1997}, stated with the decoder
normalization used below.

\begin{lemma}[Fine-index cq packing]
\label{lem:cq-packing}
Let \(v\mapsto\sigma_v^B\) be a finite-dimensional cq channel and let
\(L_n\) be a sequence of positive integers.  Generate independent
\(V^n(\ell)\sim p_V^{\otimes n}\), \(\ell\in[L_n]\).  If
\[
 \limsup_{n\to\infty}\frac1n\log L_n<I(V;B)_\sigma,
\]
there are codebook-dependent POVM elements
\(\{\Lambda_\ell\}_{\ell=1}^{L_n}\) whose ensemble-average decoding error
for a uniformly chosen \(\ell\) tends to zero.
\end{lemma}

\begin{proof}
Write \(L=L_n\) at blocklength \(n\).
Choose a sequence of typicality tolerances $\delta_n$ such that \(\delta_n\downarrow0\) slowly
enough that \(n\delta_n^2\to\infty\), and let \(\Pi\) be the corresponding
typical projector of
\(\bar\sigma^{\otimes n}\), where \(\bar\sigma=\sum_vp(v)\sigma_v\).
Let \(\Pi_{v^n}\) be the conditional typical projector of
\(\sigma_{v^n}=\bigotimes_t\sigma_{v_t}\), and denote
\[
 \Gamma_{v^n}=\Pi\Pi_{v^n}\Pi
\]
for typical \(v^n\), with \(\Gamma_{v^n}=0\) otherwise.  Standard typical
subspace estimates~\cite{Wilde2017} give a sequence \(\eta_n\to0\) such
that
\begin{align}
& 0 \le\Gamma_{v^n}\le I,\label{eq:gamma-effect}\\
& \E_{V^n}\Tr[(I-\Gamma_{V^n})\sigma_{V^n}]=o(1),
 \label{eq:packing-miss}\\
& \E_{\widetilde V^n}\Tr[\Gamma_{\widetilde V^n}
                  \bar\sigma^{\otimes n}]
 \le2^{-n(I(V;B)-\eta_n)}.
 \label{eq:packing-cross}
\end{align}
Equation~\eqref{eq:packing-cross} follows from
\(\rank\Pi_{v^n}\le2^{n(H(B|V)+o(1))}\) and
\(\Pi\bar\sigma^{\otimes n}\Pi
\le2^{-n(H(B)-o(1))}\Pi\); while \eqref{eq:packing-miss} follows from typicality and
the gentle measurement lemma.

For each codeword, set \(\Gamma_\ell=\Gamma_{V^n(\ell)}\),
\(G=\sum_{r=1}^L\Gamma_r\), and first define the sub-POVM
\[
 \widetilde\Lambda_\ell=G^{-1/2}\Gamma_\ell G^{-1/2},
\]
where the inverse is taken on \(\supp G\).  Choose an arbitrary
\(\ell_0\in[L]\) and complete it to an \(L\)-outcome POVM by setting
\begin{equation}
 \Lambda_{\ell_0}=\widetilde\Lambda_{\ell_0}
   +I-\sum_{r=1}^L\widetilde\Lambda_r,
 \qquad
 \Lambda_\ell=\widetilde\Lambda_\ell\quad(\ell\ne\ell_0).
 \label{eq:packing-povm-completion}
\end{equation}
Indeed,
\(\sum_{r=1}^L\widetilde\Lambda_r=\Pi_{\supp G}\le I\), so the
added complement is positive and the completed operators sum to the
identity.
The Hayashi--Nagaoka inequality \cite{HayashiNagaoka2003} gives, for every
realization of the codebook and every transmitted index \(\ell\),
\[
 I-\widetilde\Lambda_\ell
 \le2(I-\Gamma_\ell)+4\sum_{r\ne\ell}\Gamma_r.
\]
For \(\ell=\ell_0\), the added positive operator can only decrease the
error; for every other index the final and raw elements coincide.  Thus the
same upper bound controls the completed POVM.
Average the first term using \eqref{eq:packing-miss}.  Every competitor is
independent of the transmitted state, whose ensemble average is
\(\bar\sigma^{\otimes n}\), so \eqref{eq:packing-cross} bounds each cross
term.  Hence
\[
 \E P_e\le o(1)+4L\,2^{-n(I(V;B)-\eta_n)}\longrightarrow0
\]
under the stated strict rate inequality.
\end{proof}

\subsubsection{Reliability bound}

Define the effective channel seen by receiver \(j\) as
\begin{equation}
 \sigma_{v_j}^{B_j}
 :=\sum_{v_{\bar\jmath}}p(v_{\bar\jmath}|v_j)
 \rho_{v_1,v_2}^{B_j},\qquad \bar\jmath=3-j.
 \label{eq:effective-channel}
\end{equation}
The corresponding Holevo information is \(I(V_j;B_j)_\sigma\).

For every fine index \(\ell_j=(m_j,s_j,k_j)\), build the detection operator
of \cref{lem:cq-packing} for the codeword \(V_j^n(\ell_j)\), and normalize
the square-root measurement over all \(M_jJ_jK_j\) fine indices.  Receiver
\(j\) records the fine index and discards \((s_j,k_j)\).  Equivalently, its
message POVM is the valid coarse-graining
\begin{equation}
 \Lambda_{m_j}^{(j)}=\sum_{s_j,k_j}\Lambda_{m_j,s_j,k_j}^{(j)}.
 \label{eq:message-coarse-graining}
\end{equation}
The sum occurs after the fine-index square-root measurement has been
normalized.

\begin{theorem}[Reliability of the normalized likelihood encoder]
\label{thm:reliability}
For the code and encoder of \cref{sec:code}, the ensemble-average decoding
errors at both receivers tend to zero if
\begin{align}
 &\widetilde R_1>0,\qquad \widetilde R_2>0,
 \label{eq:positive-correlation-rates}\\
 &\widetilde R_1+\widetilde R_2>I(V_1;V_2),
 \label{eq:reliability-covering}\\
& R_j+R_j'+\widetilde R_j<I(V_j;B_j),\qquad j=1,2.
 \label{eq:reliability-packing}
\end{align}
\end{theorem}

\begin{proof}
Fix receiver \(j\) and abbreviate \(a_i=(m_i,s_i)\) and
\(\ell_i=(a_i,k_i)\).  Its fine-index error for a fixed codebook is
\begin{align}
 P_{e,j}(\cC)=\frac1{M_1J_1M_2J_2}\sum_{a_1,a_2}\sum_{k_1,k_2}
 &Q_\cC(k_1,k_2|a_1,a_2)\nonumber\\[-1mm]
 &\times\Tr[(I-\Lambda_{a_j,k_j}^{(j)})
 \rho_{V_1^n(a_1,k_1),V_2^n(a_2,k_2)}^{B_j^n}].
 \label{eq:actual-fine-error}
\end{align}
The message error is no larger.  Replace \(Q_\cC\) in this expression by the
planted weight \(L_n/(K_1K_2)\), and denote the resulting nonnegative
functional by \(\widetilde P_{e,j}(\cC)\).  Although its weights need not
normalize for fixed \(\cC\), the trace factor lies in \([0,1]\), so
\begin{equation}
 P_{e,j}(\cC)\le\widetilde P_{e,j}(\cC)+\Delta_n(\cC),
 \label{eq:error-change-law}
\end{equation}
where \(\Delta_n\) is the average, over \((a_1,a_2)\), of the L1-distance in \eqref{eq:selector-change-law}.  By
\cref{lem:cloud-mixing}, \(\E\Delta_n=o(1)\).

It remains to analyze the planted term.  Its codebook-ensemble expectation
is an error expectation under the normalized measure
\eqref{eq:global-planted-measure}.  Let
\(I^\star=(\ell_1^\star,\ell_2^\star)\) denote the selected fine-index pair,
put \(V_i^{n,\star}=V_i^n(\ell_i^\star)\), and let
\(\Gamma_{v_j^n}^{(j)}\) denote the detection operator used to form
receiver~\(j\)'s square-root measurement.  Also write
\[
 \sigma_{v_j^n}^{B_j^n}:=\bigotimes_{t=1}^n
 \sigma_{v_{j,t}}^{B_{j,t}},
 \qquad
 \bar\sigma_j^{B_j}:=\sum_{v_j}p_j(v_j)\sigma_{v_j}^{B_j}.
\]
By \cref{lem:planted-factorization}, conditional on \(I^\star\), the
selected pair has law \(p_{V_1V_2}^{\otimes n}\), whereas every competing
word in receiver~\(j\)'s marginal codebook is independent of that pair and
has law \(p_j^{\otimes n}\).  Averaging the unobserved selected auxiliary
therefore gives
\begin{align}
 &\E_{\widetilde P}\!\left[
 \Tr\!\left[(I-\Gamma_{V_j^{n,\star}}^{(j)})
 \rho_{V_1^{n,\star},V_2^{n,\star}}^{B_j^n}\right]
 \middle| I^\star\right]\nonumber\\
 &\qquad=
 \E_{V_j^n\sim p_j^{\otimes n}}
 \Tr\!\left[(I-\Gamma_{V_j^n}^{(j)})
 \sigma_{V_j^n}^{B_j^n}\right]
 =o(1),
 \label{eq:planted-packing-miss}
\end{align}
and, for every \(\ell_j\ne\ell_j^\star\),
\begin{align}
 &\E_{\widetilde P}\!\left[
 \Tr\!\left[\Gamma_{V_j^n(\ell_j)}^{(j)}
 \rho_{V_1^{n,\star},V_2^{n,\star}}^{B_j^n}\right]
 \middle| I^\star\right]\nonumber\\
 &\qquad=
 \E_{\widetilde V_j^n\sim p_j^{\otimes n}}
 \Tr\!\left[\Gamma_{\widetilde V_j^n}^{(j)}
 (\bar\sigma_j^{B_j})^{\otimes n}\right]
 \le 2^{-n[I(V_j;B_j)-\eta_n]}.
 \label{eq:planted-packing-cross}
\end{align}
Both estimates are uniform in the fixed value of \(I^\star\), since the
conditional law in \cref{lem:planted-factorization} is index-independent.
The square-root normalization depends on the full marginal codebook, but the
Hayashi--Nagaoka inequality is applied pointwise.  For each codebook
realization it gives
\begin{align}
 &\Tr\!\left[(I-\Lambda_{\ell_j^\star}^{(j)})
 \rho_{V_1^{n,\star},V_2^{n,\star}}^{B_j^n}\right]\nonumber\\
 &\quad\le
 2\Tr\!\left[(I-\Gamma_{V_j^{n,\star}}^{(j)})
 \rho_{V_1^{n,\star},V_2^{n,\star}}^{B_j^n}\right]
 +4\sum_{\ell_j\ne\ell_j^\star}
 \Tr\!\left[\Gamma_{V_j^n(\ell_j)}^{(j)}
 \rho_{V_1^{n,\star},V_2^{n,\star}}^{B_j^n}\right].
 \label{eq:planted-pointwise-hn}
\end{align}
If \(\ell_j^\star\) is the outcome to which the POVM complement was added,
its error is only smaller.  Taking expectations in
\eqref{eq:planted-pointwise-hn}, using
\eqref{eq:planted-packing-miss}--\eqref{eq:planted-packing-cross}, and
counting the competing fine indices gives
\begin{equation}
 \E\widetilde P_{e,j}
 \le o(1)+4M_jJ_jK_j\,2^{-n[I(V_j;B_j)-\eta_n]}.
 \label{eq:planted-packing-bound}
\end{equation}
Because
\(n^{-1}\log(M_jJ_jK_j)\to R_j+R_j'+\widetilde R_j\), the right side
vanishes under \eqref{eq:reliability-packing}.  Combining
this with \eqref{eq:error-change-law} proves that the same applies for the likelihood-selected law.
The probability that at least one receiver errs is at most their sum by the union bound.
\end{proof}

The two stronger individual covering conditions used for secrecy below
imply \eqref{eq:positive-correlation-rates}--
\eqref{eq:reliability-covering}, as we show in Section \ref{sec-auxialiary-rate-elimination}.  Thus, the same codebook ensemble and
normalized selector serve both reliability and secrecy.

\subsection{Secrecy}
\label{sec:resolvability}


For the secrecy analysis, we compare the likelihood-selector $Q_\cC$ with two oriented posterior
laws, one for each unintended receiver. For receiver 2, we consider the right-oriented selector $Q^\to$, which chooses the $V_2$-index uniformly and then chooses the $V_1$-index conditionally using the likelihood weights. The selector-balancing argument in Lemma~\ref{lem:selector-balancing} shows that $Q^\to$ is close to $Q_\cC$. For receiver 1, the same argument uses the left-oriented selector $Q^\leftarrow$.

These comparisons reduce the analysis to one-sided likelihood resolvability for both receivers.

\subsubsection{A cq soft-covering lemma}

We use the following standard cq soft-covering lemma
\cite{ChengGao2024}. A proof using typical projectors and
a second-moment estimate is given in
Appendix~\ref{app:cq-soft-covering}.

\begin{lemma}[Memoryless cq soft covering]
\label{lem:cq-soft-covering}
Let \(q_X\) be a distribution on a finite alphabet and let
\(x\mapsto\zeta_x^Z\) be a finite-dimensional cq channel.  Put
\[
 \zeta^{XZ}=\sum_xq_X(x)\ketbra{x}\otimes\zeta_x^Z,
 \qquad \bar\zeta^Z=\sum_xq_X(x)\zeta_x^Z.
\]
Generate \(L_n=\lceil2^{nR}\rceil\) independent codewords
\(X^n(\ell)\sim q_X^{\otimes n}\) and define
\[
 \widehat\zeta_\cC^{Z^n}
 :=\frac1{L_n}\sum_{\ell=1}^{L_n}\zeta_{X^n(\ell)}^{Z^n}.
\]
Here \(\zeta_{x^n}^{Z^n}:=\bigotimes_{r=1}^n\zeta_{x_r}^Z\).
If \(R>I(X;Z)_\zeta\), there are constants \(c>0\) and \(n_0\) such that
\begin{equation}
 \E_\cC\dtr{\widehat\zeta_\cC^{Z^n}}
 {\left(\bar\zeta^Z\right)^{\otimes n}}
 \le2^{-cn},\qquad n\ge n_0.
 \label{eq:cq-soft-covering}
\end{equation}
The classical statement follows by representing output distributions as
diagonal states.
\end{lemma}

\begin{proof}
See Appendix \ref{appn:cq-soft-covering}.
\end{proof}

\subsubsection{One-sided likelihood resolvability}

\begin{lemma}[One-sided cq likelihood resolvability]
\label{lem:one-sided}
Let \(q_{AS}\) be a distribution on finite alphabets and let
\((a,s)\mapsto\rho_{a,s}^B\) be a cq channel.  Choose
\(q_{S|A}(\cdot|a)\) arbitrarily when \(q_A(a)=0\), and use the memoryless
extension
\[
 \rho_{a^n,s^n}^{B^n}:=\bigotimes_{t=1}^n\rho_{a_t,s_t}^{B_t}.
\]
Independently generate
\[
 A^n(w,i)\overset{\mathrm{ind}}{\sim} q_A^{\otimes n},\qquad
 w\in[W_n],\ i\in[I_n],
\]
where \(W_n=\lceil2^{nR}\rceil\) and
\(I_n=\lceil2^{n\widetilde R}\rceil\).  For fixed \((w,s^n)\), define
\begin{equation}
 \widehat Q_\cC(i|w,s^n)
 :=\frac{q_{S|A}^{\otimes n}(s^n|A^n(w,i))}
 {\sum_{r=1}^{I_n}q_{S|A}^{\otimes n}(s^n|A^n(w,r))},
 \label{eq:one-sided-selector}
\end{equation}
using the uniform distribution on \([I_n]\) when the denominator is zero.
Define
\begin{align}
 \sigma_{s^n,\cC}^{B^n}
 &:=\frac1{W_n}\sum_{w=1}^{W_n}\sum_{i=1}^{I_n}
 \widehat Q_\cC(i|w,s^n)\rho_{A^n(w,i),s^n}^{B^n},
 \label{eq:one-sided-state}\\
 \bar\rho_{s^n}^{B^n}
 &:=\sum_{a^n}q_{A|S}^{\otimes n}(a^n|s^n)
 \rho_{a^n,s^n}^{B^n}.
 \label{eq:one-sided-target}
\end{align}
For \(q_S(s)=0\), choose \(q_{A|S}(\cdot|s)\) arbitrarily; those letters
have zero weight in \eqref{eq:one-sided-conclusion}.
If
\begin{equation}
 \widetilde R>I(A;S)_q,\qquad
 R+\widetilde R>I(A;SB)_\xi,
 \label{eq:one-sided-rates}
\end{equation}
where
\(\xi^{ASB}=\sum_{a,s}q(a,s)\ketbra{a,s}\otimes\rho_{a,s}^B\) and
\(\xi^{SB}:=\Tr_A\xi^{ASB}\), then
there are \(c>0,n_0\) such that
\begin{equation}
 \E_\cC\sum_{s^n}q_S^{\otimes n}(s^n)
 \dtr{\sigma_{s^n,\cC}^{B^n}}{\bar\rho_{s^n}^{B^n}}
 \le2^{-cn},\qquad n\ge n_0.
 \label{eq:one-sided-conclusion}
\end{equation}
\end{lemma}

\begin{proof}
Fix \(\cC\) and retain every classical register in two cq experiments.  The
actual state is
\begin{align}
 P_\cC^{W_0I_0A^nS^nB^n}
 :=\frac1{W_n}\sum_{w,s^n,i}&q_S^{\otimes n}(s^n)
 \widehat Q_\cC(i|w,s^n)
 \ketbra{w,i,A^n(w,i),s^n}\nonumber\\[-1mm]
 &\otimes\rho_{A^n(w,i),s^n}^{B^n}.
 \label{eq:one-sided-actual-full-state}
\end{align}
Thus one first draws \(S^n\sim q_S^{\otimes n}\) and \(W_0\) uniformly,
and then draws \(I_0\) from the posterior selector.  The ideal state is
\begin{align}
 \Gamma_\cC^{W_0I_0A^nS^nB^n}
 :=\frac1{W_nI_n}\sum_{w,i,s^n}&
 q_{S|A}^{\otimes n}(s^n|A^n(w,i))
 \ketbra{w,i,A^n(w,i),s^n}\nonumber\\[-1mm]
 &\otimes\rho_{A^n(w,i),s^n}^{B^n}.
 \label{eq:one-sided-ideal-full-state}
\end{align}
It draws \((W_0,I_0)\) uniformly, sets \(A^n=A^n(W_0,I_0)\), and passes
that codeword through the memoryless cq channel \(a\mapsto S B\).
Put
\[
 \widehat q_{w,\cC}(s^n)
 =\frac1{I_n}\sum_iq_{S|A}^{\otimes n}(s^n|A^n(w,i)).
\]
Whenever this quantity is positive, Bayes' rule gives
\(\Gamma_\cC(i|w,s^n)=\widehat Q_\cC(i|w,s^n)\).  When it is zero, the
ideal experiment assigns zero mass to \((w,s^n)\), and its regular
conditional cq kernel there may be chosen to equal the uniform fallback in
the actual experiment.  Moreover, \(q_S^{\otimes n}(s^n)=0\) implies
\(\widehat q_{w,\cC}(s^n)=0\) almost surely.  Consequently, the conditional
cq kernels of the two full states agree once \((W_0,S^n)\) is given.
Orthogonality of the \((w,s^n)\) blocks gives
\begin{equation}
 \dtr{P_\cC^{W_0I_0A^nS^nB^n}}
 {\Gamma_\cC^{W_0I_0A^nS^nB^n}}
 =\frac1{W_n}\sum_w\tv{q_S^{\otimes n}}{\widehat q_{w,\cC}}.
 \label{eq:actual-ideal-exact}
\end{equation}
This identity includes the zero-denominator rule.

Apply \cref{lem:cq-soft-covering} to the classical channel
\(a\mapsto q_{S|A}(\cdot|a)\).  The first inequality in
\eqref{eq:one-sided-rates} gives
\begin{equation}
 \E\dtr{P_\cC^{W_0I_0A^nS^nB^n}}
 {\Gamma_\cC^{W_0I_0A^nS^nB^n}}\le2^{-c_1n}
 \label{eq:one-sided-first-cover}
\end{equation}
for some \(c_1>0\).

Next, regard
\begin{equation}
 a\longmapsto\zeta_a^{SB}
 :=\sum_sq(s|a)\ketbra{s}\otimes\rho_{a,s}^B
 \label{eq:augmented-channel}
\end{equation}
as one cq channel.  Under the ideal experiment,
\[
 \Gamma_\cC^{S^nB^n}=\frac1{W_nI_n}\sum_{w,i}
 \zeta_{A^n(w,i)}^{S^nB^n}.
\]
Here \(\zeta_{a^n}^{S^nB^n}:=\bigotimes_{r=1}^n\zeta_{a_r}^{S_rB_r}\).
All \(W_nI_n\) codewords are independent.  The second inequality in
\eqref{eq:one-sided-rates} and \cref{lem:cq-soft-covering} imply
\begin{equation}
 \E\dtr{\Gamma_\cC^{S^nB^n}}{(\xi^{SB})^{\otimes n}}
 \le2^{-c_2n}.
 \label{eq:one-sided-second-cover}
\end{equation}
Tracing out \(W_0I_0A^n\) is a channel, and hence
\[
 \dtr{P_\cC^{S^nB^n}}{\Gamma_\cC^{S^nB^n}}
 \le\dtr{P_\cC^{W_0I_0A^nS^nB^n}}
 {\Gamma_\cC^{W_0I_0A^nS^nB^n}}.
\]
Contractivity, the triangle inequality,
\eqref{eq:one-sided-first-cover}, and \eqref{eq:one-sided-second-cover}
therefore show that
\[
 \E\dtr{P_\cC^{S^nB^n}}{(\xi^{SB})^{\otimes n}}
 \le2^{-c_1n}+2^{-c_2n}.
\]
Both states have the same classical marginal \(q_S^{\otimes n}\), so block
diagonality gives
\[
 \dtr{P_\cC^{S^nB^n}}{(\xi^{SB})^{\otimes n}}
 =\sum_{s^n}q_S^{\otimes n}(s^n)
 \dtr{\sigma_{s^n,\cC}^{B^n}}{\bar\rho_{s^n}^{B^n}}.
\]
Absorbing the sum of exponentials into a smaller exponent proves the claim.
\end{proof}

\subsubsection{Two orientations of the likelihood selector}

\begin{lemma}[Selector balancing]
\label{lem:selector-balancing}
Let \(A_1,\ldots,A_{K_1}\) be i.i.d. \(p_1^{\otimes n}\), and
independently let \(B_1,\ldots,B_{K_2}\) be i.i.d.
\(p_2^{\otimes n}\).  Set \(L_{ij}=L_n(A_i,B_j)\) and define
\[
 Q(i,j)=\frac{L_{ij}}{\sum_{a,b}L_{ab}}.
\]
If the global denominator is zero, set \(Q(i,j)=1/(K_1K_2)\).
Define the right- and left-oriented selectors
\begin{align}
 Q^\to(i,j)&=\frac1{K_2}\frac{L_{ij}}{\sum_aL_{aj}},
 \label{eq:right-selector}\\
 Q^\leftarrow(i,j)&=\frac1{K_1}\frac{L_{ij}}{\sum_bL_{ib}},
 \label{eq:left-selector}
\end{align}
with uniform conditional fallbacks for zero column or row denominators.  If
\(\widetilde R_1>I(V_1;V_2)\), then
\begin{equation}
 \E\tv{Q}{Q^\to}\le2^{-c_\to n}
 \label{eq:right-balance}
\end{equation}
for some \(c_\to>0\).  If
\(\widetilde R_2>I(V_1;V_2)\), then
\begin{equation}
 \E\tv{Q}{Q^\leftarrow}\le2^{-c_\leftarrow n}
 \label{eq:left-balance}
\end{equation}
for some \(c_\leftarrow>0\).
\end{lemma}

\begin{proof}
We prove \eqref{eq:right-balance}.  Put
\[
 Z_j=\sum_iL_{ij},\qquad a_j=Z_j/K_1,
 \qquad \bar a=K_2^{-1}\sum_ja_j.
\]
When \(\bar a>0\), the two selectors have the same conditional distribution
on \(i\) in every positive column.  They differ only in their \(j\)-marginal:
\(Q\) assigns \(a_j/(K_2\bar a)\), whereas \(Q^\to\) assigns \(1/K_2\).
A zero column has zero mass under \(Q\), while \(Q^\to\) assigns it mass
\(1/K_2\); its contribution to total variation is therefore
\(1/(2K_2)\).  Combining positive and zero columns gives
\begin{equation}
 \tv{Q}{Q^\to}=\frac1{2K_2}\sum_j\left|\frac{a_j}{\bar a}-1\right|.
 \label{eq:balance-marginal-distance}
\end{equation}
Let \(d_0=K_2^{-1}\sum_j|a_j-1|\).  Since
\(|\bar a-1|\le d_0\),
\(K_2^{-1}\sum_j|a_j-\bar a|\le2d_0\).  If \(d_0\le1/2\), then
\(\bar a\ge1/2\), and \eqref{eq:balance-marginal-distance} is at most
\(2d_0\).  If \(d_0>1/2\), the same bound follows from total variation
being at most one.  If \(\bar a=0\), both selectors are uniform and hence
equal.  Thus
\begin{equation}
 \tv{Q}{Q^\to}\le\frac2{K_2}\sum_j|a_j-1|.
 \label{eq:balance-degree-bound}
\end{equation}

Condition on the \(A\)-book and define
\[
 \widehat p_{2,\cC_1}(b^n)
 =\frac1{K_1}\sum_i p_{V_2|V_1}^{\otimes n}(b^n|A_i).
\]
The likelihood identity gives
\begin{align}
 \E_{B_j}|a_j-1|
 &=\sum_{b^n}|\widehat p_{2,\cC_1}(b^n)-p_2^{\otimes n}(b^n)|
 \nonumber\\
 &=2\tv{\widehat p_{2,\cC_1}}{p_2^{\otimes n}}.
 \label{eq:degree-soft-covering}
\end{align}
If \(p_2^{\otimes n}(b^n)=0\), then
\(\widehat p_{2,\cC_1}(b^n)=0\) almost surely: a zero marginal forces
\(p_{V_2|V_1}^{\otimes n}(b^n|A_i)=0\) for every generated \(A_i\).
Thus \eqref{eq:degree-soft-covering} involves no division by zero.
Apply \cref{lem:cq-soft-covering} to the classical channel
\(V_1\mapsto V_2\), whose mutual information is \(I(V_1;V_2)\), and
average \eqref{eq:balance-degree-bound}.  This proves
\eqref{eq:right-balance}.  Interchanging the two books proves
\eqref{eq:left-balance}.
\end{proof}

The uniform fallback defines a stochastic encoder for every codebook
realization.  Under the right-balancing condition,
the globally zero event is exponentially unlikely because it forces
\(d_0=1\) in the preceding column proof.  Under the left-balancing condition,
the same conclusion follows from the symmetric row argument.

\subsubsection{Bipartite resolvability}

Define the conditional output states
\begin{align}
 \bar\rho_{v_2^n}^{B_2^n}
 &:=\sum_{v_1^n}p_{V_1|V_2}^{\otimes n}(v_1^n|v_2^n)
 \rho_{v_1^n,v_2^n}^{B_2^n},
 \label{eq:conditional-target-two}\\
 \bar\rho_{v_1^n}^{B_1^n}
 &:=\sum_{v_2^n}p_{V_2|V_1}^{\otimes n}(v_2^n|v_1^n)
 \rho_{v_1^n,v_2^n}^{B_1^n}.
 \label{eq:conditional-target-one}
\end{align}
For a fixed codebook, set the target states
\begin{align}
 \theta_{m_2,\cC}^{B_2^n}
 &:=\frac1{J_2K_2}\sum_{s_2,k_2}
 \bar\rho_{V_2^n(m_2,s_2,k_2)}^{B_2^n},
 \label{eq:target-two}\\
 \theta_{m_1,\cC}^{B_1^n}
 &:=\frac1{J_1K_1}\sum_{s_1,k_1}
 \bar\rho_{V_1^n(m_1,s_1,k_1)}^{B_1^n}.
 \label{eq:target-one}
\end{align}
The first target may depend on \(m_2\) and \(\cC\), but it is independent
of the message \(m_1\) that must be hidden from receiver 2.  The reverse
statement holds for the second target.

\begin{theorem}[Bipartite likelihood-selector cq resolvability]
\label{thm:bipartite-resolvability}
For every fixed message pair \((m_1,m_2)\), the following ensemble-expected
bounds hold.  By codebook symmetry, the constants and exponents are uniform
over the pair.

If
\begin{equation}
 \widetilde R_1>I(V_1;V_2),\qquad
 R_1'+\widetilde R_1>I(V_1;V_2B_2),
 \label{eq:resolvability-rates-two}
\end{equation}
then, for some \(c_2>0\) and all sufficiently large \(n\),
\begin{equation}
 \E_\cC\dtr{\rho_{m_1,m_2,\cC}^{B_2^n}}
 {\theta_{m_2,\cC}^{B_2^n}}
 \le2^{-c_2n}.
 \label{eq:resolvability-conclusion-two}
\end{equation}
If
\begin{equation}
 \widetilde R_2>I(V_1;V_2),\qquad
 R_2'+\widetilde R_2>I(V_2;V_1B_1),
 \label{eq:resolvability-rates-one}
\end{equation}
then, for some \(c_1>0\) and all sufficiently large \(n\),
\begin{equation}
 \E_\cC\dtr{\rho_{m_1,m_2,\cC}^{B_1^n}}
 {\theta_{m_1,\cC}^{B_1^n}}
 \le2^{-c_1n}.
 \label{eq:resolvability-conclusion-one}
\end{equation}
\end{theorem}

\begin{proof}
We prove \eqref{eq:resolvability-conclusion-two}; reverse all subscripts for
the other conclusion.  For fixed \((m_1,m_2,s_1,s_2)\), abbreviate the two
selected subbooks by \(v_{1,k_1}\) and \(v_{2,k_2}\), and define
\begin{equation}
 Q_\cC^\to(k_1,k_2)
 :=\frac1{K_2}\frac{L_n(v_{1,k_1},v_{2,k_2})}
 {\sum_{a=1}^{K_1}L_n(v_{1,a},v_{2,k_2})},
 \label{eq:oriented-selector-proof}
\end{equation}
with a uniform conditional distribution in a zero column.  Let
\(\rho_{m_1,m_2,\cC}^{\to,B_2^n}\) be the state obtained from
\eqref{eq:exact-induced-state} by replacing \(Q_\cC\) with \(Q_\cC^\to\).
The map from an index pair to its output state is cq.  Contractivity,
convexity, and \cref{lem:selector-balancing} give, under the first inequality
of \eqref{eq:resolvability-rates-two},
\begin{equation}
 \E\dtr{\rho_{m_1,m_2,\cC}^{B_2^n}}
 {\rho_{m_1,m_2,\cC}^{\to,B_2^n}}
 \le2^{-\alpha_2n}
 \label{eq:global-to-oriented}
\end{equation}
for some \(\alpha_2>0\).

For a fixed \(v_2^n\), the conditional law of \(k_1\) in
\eqref{eq:oriented-selector-proof} is
\begin{align}
 \frac{L_n(V_1^n(m_1,s_1,k_1),v_2^n)}
 {\sum_aL_n(V_1^n(m_1,s_1,a),v_2^n)}
 =\frac{p_{V_2|V_1}^{\otimes n}
 (v_2^n|V_1^n(m_1,s_1,k_1))}
 {\sum_ap_{V_2|V_1}^{\otimes n}
 (v_2^n|V_1^n(m_1,s_1,a))}.
 \label{eq:likelihood-cancellation}
\end{align}
for every positive column.  In a zero column both the oriented selector and
the one-sided encoder use the same uniform fallback.  Hence this is
\cref{lem:one-sided} with
\begin{equation}
 A=V_1,\quad S=V_2,\quad B=B_2,\quad
 w=s_1,\quad i=k_1,\quad R=R_1',\quad
 \widetilde R=\widetilde R_1.
 \label{eq:one-sided-identification}
\end{equation}
Let \(\sigma_{m_1,v_2^n,\cC_1}^{B_2^n}\) denote the one-sided state
\eqref{eq:one-sided-state} under this identification.  The oriented output
decomposes as
\begin{equation}
 \rho_{m_1,m_2,\cC}^{\to,B_2^n}
 =\frac1{J_2K_2}\sum_{s_2,k_2}
 \sigma_{m_1,V_2^n(m_2,s_2,k_2),\cC_1}^{B_2^n}.
 \label{eq:oriented-decomposition}
\end{equation}
The two information quantities in \eqref{eq:one-sided-rates} become
\[
 I(A;S)=I(V_1;V_2),\qquad I(A;SB)=I(V_1;V_2B_2).
\]
Thus the hypotheses of \cref{lem:one-sided} are
\eqref{eq:resolvability-rates-two}.

By convexity, \eqref{eq:oriented-decomposition}, and
\eqref{eq:target-two},
\begin{align}
 &\dtr{\rho_{m_1,m_2,\cC}^{\to,B_2^n}}
 {\theta_{m_2,\cC}^{B_2^n}}\nonumber\\
 &\quad\le\frac1{J_2K_2}\sum_{s_2,k_2}
 \dtr{\sigma_{m_1,V_2^n(m_2,s_2,k_2),\cC_1}^{B_2^n}}
 {\bar\rho_{V_2^n(m_2,s_2,k_2)}^{B_2^n}}.
 \label{eq:oriented-convexity}
\end{align}
Taking expectations uses the independence of each displayed \(V_2^n\) from
the \(V_1\)-book; linearity handles the sum.  Applying
\cref{lem:one-sided} gives
\begin{equation}
 \E\dtr{\rho_{m_1,m_2,\cC}^{\to,B_2^n}}
 {\theta_{m_2,\cC}^{B_2^n}}
 \le2^{-\beta_2n}
 \label{eq:oriented-to-target}
\end{equation}
for some \(\beta_2>0\).  The triangle inequality with
\eqref{eq:global-to-oriented}
and \eqref{eq:oriented-to-target} proves
\eqref{eq:resolvability-conclusion-two}.

For receiver 1, use \(Q^\leftarrow\) and identify
\(A=V_2,S=V_1,B=B_1,w=s_2,i=k_2,R=R_2'\), and
\(\widetilde R=\widetilde R_2\).
\end{proof}

The chain rule gives
\begin{align}
 I(V_1;V_2B_2)&=I(V_1;V_2)+I(V_1;B_2|V_2),
 \label{eq:chain-two}\\
 I(V_2;V_1B_1)&=I(V_1;V_2)+I(V_2;B_1|V_1).
 \label{eq:chain-one}
\end{align}
Thus the correlation-index rate covers \(I(V_1;V_2)\), while the secrecy
randomization covers the conditional leakage terms.

\subsubsection{Conditional strong secrecy}
\label{sec:completion}

For a fixed codebook let
\[
 \Omega_\cC^{M_1M_2B_2^n}
 =\frac1{M_1M_2}\sum_{m_1,m_2}\ketbra{m_1,m_2}
 \otimes\rho_{m_1,m_2,\cC}^{B_2^n}.
\]

\begin{corollary}[Conditional trace secrecy]
\label{cor:trace-secrecy}
Under \eqref{eq:resolvability-rates-two}, for some \(c>0\),
\begin{equation}
 \E_\cC\dtr{\Omega_\cC^{M_1M_2B_2^n}}
 {\pi^{M_1}\otimes\Omega_\cC^{M_2B_2^n}}
 \le2^{-cn}.
 \label{eq:trace-secrecy-two}
\end{equation}
The symmetric statement holds under \eqref{eq:resolvability-rates-one}.
\end{corollary}

\begin{proof}
Define
\[
 \Xi_\cC^{M_1M_2B_2^n}
 =\pi^{M_1}\otimes\frac1{M_2}\sum_{m_2}\ketbra{m_2}
 \otimes\theta_{m_2,\cC}^{B_2^n}.
\]
Block diagonality gives
\[
 \dtr{\Omega_\cC}{\Xi_\cC}
 =\frac1{M_1M_2}\sum_{m_1,m_2}
 \dtr{\rho_{m_1,m_2,\cC}^{B_2^n}}
 {\theta_{m_2,\cC}^{B_2^n}}.
\]
Thus \cref{thm:bipartite-resolvability} gives
\(\E\dtr{\Omega_\cC}{\Xi_\cC}\le2^{-c_0n}\).  Tracing out \(M_1\)
and using contractivity gives
\[
 \dtr{\Omega_\cC^{M_2B_2^n}}{\Xi_\cC^{M_2B_2^n}}
 \le\dtr{\Omega_\cC}{\Xi_\cC}.
\]
Since \(\Xi_\cC=\pi^{M_1}\otimes\Xi_\cC^{M_2B_2^n}\), the triangle
inequality yields, codebook by codebook,
\[
 \dtr{\Omega_\cC}
 {\pi^{M_1}\otimes\Omega_\cC^{M_2B_2^n}}
 \le2\dtr{\Omega_\cC}{\Xi_\cC}.
\]
Decrease the exponent to absorb the factor two.  The other direction is
identical.
\end{proof}

\subsection{Rate selection and deterministic codebooks}
\label{sec:rate-selection}

We now combine the reliability and secrecy estimates.
Eliminating the auxiliary rates gives the stated rate inequalities
for constant \(U\), and a joint selection argument provides one
codebook satisfying both reliability and secrecy.

\subsubsection{Auxiliary-rate elimination}
\label{sec-auxialiary-rate-elimination}
Put
\begin{align}
 A_1&=I(V_1;B_1),&A_2&=I(V_2;B_2),\nonumber\\
 C&=I(V_1;V_2),&
 D_1&=I(V_1;B_2|V_2),&D_2&=I(V_2;B_1|V_1).
 \label{eq:rate-abbreviations-no-u}
\end{align}
The complete system furnished by
\cref{thm:reliability,thm:bipartite-resolvability} is
\begin{equation}
 R_i\ge0,\qquad R_i'\ge0,\qquad \widetilde R_i\ge0,
 \qquad i=1,2,
 \label{eq:all-rates-nonnegative}
\end{equation}
along with
\begin{align}
 \widetilde R_1+\widetilde R_2&>C,
 \label{eq:all-rates-cloud}\\
 \widetilde R_1&>C,& R_1'+\widetilde R_1&>C+D_1,
 \label{eq:all-rates-one}\\
 \widetilde R_2&>C,& R_2'+\widetilde R_2&>C+D_2,
 \label{eq:all-rates-two}\\
 R_1+R_1'+\widetilde R_1&<A_1,&
 R_2+R_2'+\widetilde R_2&<A_2.
 \label{eq:all-rates-packing}
\end{align}
The first line is redundant under the two individual covering conditions.
Every feasible tuple in this auxiliary-rate system must satisfy
\begin{equation}
 R_1<A_1-C-D_1,\qquad R_2<A_2-C-D_2
 \label{eq:projected-no-u}
\end{equation}
This follows immediately from the resolvability and packing inequalities.
Conversely, if \eqref{eq:projected-no-u} holds, choose \(\eta>0\) such
that \(R_i+C+D_i+2\eta<A_i\), and set
\begin{equation}
 \widetilde R_1=\widetilde R_2=C+\eta,\qquad
 R_1'=D_1+\eta,\qquad R_2'=D_2+\eta.
 \label{eq:auxiliary-rate-choice}
\end{equation}
All inequalities \eqref{eq:all-rates-cloud}--
\eqref{eq:all-rates-packing} then hold, and their projection is
\eqref{eq:projected-no-u}.

\subsubsection{Selection of a deterministic codebook}
\label{sec-selection-of-det-codebook}

For the random codebook, define
\[
 e_n(\cC)=P_{e,1}^{(n)}(\cC)+P_{e,2}^{(n)}(\cC),\qquad
 s_n(\cC)=\Delta_{1,n}(\cC)+\Delta_{2,n}(\cC).
\]
The reliability theorem gives \(\E e_n=a_n\to0\), and the secrecy
corollary gives \(\E s_n\le2^{-n\kappa}\) for some \(\kappa>0\).  Markov's
inequality yields
\begin{align*}
 \Prb\{e_n>\sqrt{a_n}\}&\le\sqrt{a_n},\\
 \Prb\{s_n>2^{-n\kappa/2}\}&\le2^{-n\kappa/2}.
\end{align*}
For sufficiently large \(n\), the complements of the two events have a nonempty intersection, yielding a
deterministic sequence of codebooks and fine-index POVMs with vanishing
average error and exponentially small trace leakage.  The selected codebook
retains the stochastic likelihood selector, the random choices \(S_1,S_2\),
and the channel prefix.

To pass from trace secrecy to \eqref{eq:strong-secrecy-definition}, let
\(\delta_{1,n}=\Delta_{1,n}\) for a selected codebook and put
\[
 \tau^{M_1M_2B_2^n}
 :=\pi^{M_1}\otimes\Omega^{M_2B_2^n}.
\]
The independent uniform messages satisfy
\(H(M_1|M_2)_\Omega=\log M_1\), while
\(H(M_1|M_2B_2^n)_\tau=\log M_1\).  Since
\(\dtr{\Omega}{\tau}=\delta_{1,n}\), Winter's uniform continuity bound for
conditional entropy~\cite{Winter2016} gives
\begin{equation}
 I(M_1;B_2^n|M_2)_\Omega
 \le2\delta_{1,n}\log M_1
 +(1+\delta_{1,n})h_2\!\left(\frac{\delta_{1,n}}
 {1+\delta_{1,n}}\right).
 \label{eq:trace-to-mi}
\end{equation}
Since \(\delta_{1,n}\) is exponentially small and \(\log M_1=O(n)\), the
right side tends to zero.  Interchanging the receivers proves the other
secrecy condition.

\subsection{Restoring the time-sharing variable}
\label{sec:time-sharing}

It remains to lift the constant-\(U\) coding lemmas without imposing their
rate inequalities separately for every value of \(u\).

\begin{lemma}[Rational superletter reduction]
\label{lem:superletter}
Suppose first that \(p_U\) is rational.  Choose \(d\) such that
\(dp_U(u)\) is an integer for every \(u\), and fix a list
\((u_1,\ldots,u_d)\) in which \(u\) occurs exactly \(dp_U(u)\) times.
Define
\[
 V_i^\star=(V_{i,1},\ldots,V_{i,d}),\qquad i=1,2,
\]
with joint law
\begin{equation}
 p(v_1^\star,v_2^\star)
 =\prod_{r=1}^dp(v_{1,r},v_{2,r}|u_r),
 \label{eq:superletter-law}
\end{equation}
and super-channel output
\begin{equation}
 \rho_{v_1^\star,v_2^\star}^{B_1^\star B_2^\star}
 =\bigotimes_{r=1}^d
 \rho_{u_r,v_{1,r},v_{2,r}}^{B_{1,r}B_{2,r}}.
 \label{eq:superletter-channel}
\end{equation}
Then
\begin{align}
 I(V_1^\star;V_2^\star)&=dI(V_1;V_2|U),
 \label{eq:super-correlation}\\
 I(V_i^\star;B_i^\star)&=dI(V_i;B_i|U),\qquad i=1,2,
 \label{eq:super-packing}\\
 I(V_1^\star;V_2^\star B_2^\star)&=dI(V_1;V_2B_2|U),
 \label{eq:super-resolvability-two}\\
 I(V_2^\star;V_1^\star B_1^\star)&=dI(V_2;V_1B_1|U).
 \label{eq:super-resolvability-one}
\end{align}
\end{lemma}

\begin{proof}
The state of one superletter is the tensor product of the \(d\) cq states
corresponding to \(u_1,\ldots,u_d\).  Additivity of relative entropy on
tensor products, applied to its mutual-information representations, gives
\[
 I(V_1^\star;V_2^\star)
 =\sum_{r=1}^dI(V_1;V_2)_{\omega_{u_r}}
 =d\sum_up_U(u)I(V_1;V_2)_{\omega_u},
\]
which is \eqref{eq:super-correlation}.  The same argument with the systems
\(B_i\), \(V_2B_2\), and \(V_1B_1\) proves
\eqref{eq:super-packing}--\eqref{eq:super-resolvability-one}.
\end{proof}

Apply the constant-\(U\) construction to \(N\) i.i.d. uses of this
super-channel.  It is a code for \(n=Nd\) physical channel uses.  A book of
physical rate \(R\) has size \(2^{nR+o(n)}=2^{N(dR)+o(N)}\), hence rate
\(dR\) per superletter.  Thus every rate and every mutual information is
multiplied by the same factor \(d\).  The
superletter likelihood ratio factors across its \(d\) coordinates, so the
preceding cloud-mixing, packing, balancing, and soft-covering arguments apply
to the superchannel.  Deterministic-codebook selection and the
trace-to-mutual-information conversion from Section~\ref{sec-selection-of-det-codebook} then yield one physical code satisfying
both reliability and secrecy.
The selection of \((u_1,\ldots,u_d)\) is public and fixed as part of the code.
For an arbitrary sufficiently large physical blocklength, write
\(n=Nd+r\), where \(0\le r<d\), use the superletter code on the first
\(Nd\) channel uses, and append \(r\) fixed input symbols.  Both receivers
ignore the corresponding outputs, and the resulting rate loss is \(o(1)\).

For arbitrary \(p_U\), choose rational distributions
\(p_U^{(q)}\to p_U\).  Because the alphabets and output dimensions are
finite, all mutual informations in \eqref{eq:main-rate-one}--
\eqref{eq:main-rate-two} are continuous in the underlying distribution.
Every strictly feasible pair therefore remains feasible for all sufficiently
large \(q\).  The rational construction applies, and taking the closure
proves the result for arbitrary \(p_U\).

Finally, with
\begin{align*}
 A_i&=I(V_i;B_i|U),\qquad C=I(V_1;V_2|U),\\
 D_1&=I(V_1;B_2|V_2,U),\qquad
 D_2=I(V_2;B_1|V_1,U),
\end{align*}
the auxiliary rates are nonnegative as in
\eqref{eq:all-rates-nonnegative}, and the nontrivial inequalities are
\eqref{eq:all-rates-cloud}--\eqref{eq:all-rates-packing} with these
conditional quantities.  The choice
\[
 \widetilde R_1=\widetilde R_2=C+\eta,\qquad
 R_1'=D_1+\eta,\qquad R_2'=D_2+\eta
\]
then gives \eqref{eq:main-rate-one}--\eqref{eq:main-rate-two} as
\(\eta\downarrow0\).  This proves the two-active-message part of
\cref{thm:main}.  Its inactive-user boundary is supplied by the following
proposition.

\subsection{Inactive-user specialization}

\begin{proposition}[Inactive-user specialization]
\label{prop:inactive-user}
For every state \(\omega\) in \eqref{eq:single-letter-state}, every rate
\begin{equation}
 R_1<I(V_1;B_1|U)_\omega-I(V_1;V_2B_2|U)_\omega
 \label{eq:inactive-user-one}
\end{equation}
is achievable with \(R_2=0\).  The symmetric statement holds after
interchanging the receivers.
\end{proposition}

\begin{proof}
Marginalize \(V_2\) and its contribution to the prefix, thereby obtaining a
one-user cq wiretap channel with classical auxiliary \(V_1\), legitimate
output \(B_1\), and unintended output \(B_2\).  Data processing gives
\begin{equation}
 I(V_1;B_2|U)\le I(V_1;V_2B_2|U),
 \label{eq:inactive-data-processing}
\end{equation}
so it suffices to achieve every
\[
 R_1<I(V_1;B_1|U)-I(V_1;B_2|U).
\]

First take \(U\) constant.  Generate independent codewords
\(V_1^n(m_1,s_1)\sim p_{V_1}^{\otimes n}\), where the message rate is
\(R_1\) and the randomization rate is \(R_1'\).  Given
\((m_1,s_1)\), use the marginalized stochastic prefix.  Receiver \(1\)
applies the fine-index square-root measurement over \((m_1,s_1)\) and then
discards \(s_1\).  The packing lemma \cref{lem:cq-packing} gives vanishing
ensemble-average error when
\begin{equation}
 R_1+R_1'<I(V_1;B_1).
 \label{eq:inactive-packing}
\end{equation}
For each \(m_1\), averaging over \(s_1\) is an ordinary cq
soft-covering code.  Lemma~\ref{lem:cq-soft-covering}, followed by averaging
over the uniform message, gives an exponentially vanishing expected
trace-secrecy parameter when
\begin{equation}
 R_1'>I(V_1;B_2).
 \label{eq:inactive-covering}
\end{equation}
Apply the two-threshold codebook-selection argument of
Section~\ref{sec-selection-of-det-codebook} to obtain a
deterministic codebook sequence with vanishing decoding error
and exponentially small trace leakage.
Equation~\eqref{eq:trace-to-mi} then gives conditional strong secrecy. Eliminating
\(R_1'\) yields the claimed one-user difference for constant \(U\).

For rational \(p_U\), concatenate the prescribed number of copies of each
conditional one-user channel into the super-letter of
\cref{lem:superletter}; both mutual informations in
\cref{eq:inactive-packing,eq:inactive-covering} add across its coordinates.
Rational approximation and continuity then give arbitrary \(p_U\).
Finally, \eqref{eq:inactive-data-processing} proves
\eqref{eq:inactive-user-one}.  The second boundary follows by exchanging the receiver labels. This establishes the achievability assertion of \cref{thm:main}.
\end{proof}

\subsection{Cardinality bounds}
\label{sec:cardinality}

We finally bound the auxiliary alphabets and show that this does not change the inner
region.  The reduction preserves or increases both averaged rate bounds.

\begin{proposition}[Auxiliary alphabet sizes]
\label{prop:auxiliary-cardinalities}
For the finite-input cq channel of Theorem~\ref{thm:main}, the inner region is unchanged if the union over distributions \eqref{eq:factorization} is restricted to
\begin{equation}
 |\mathcal U|\le 3,\qquad
 |\mathcal V_1|\le |\mathcal X|,\qquad
 |\mathcal V_2|\le |\mathcal X|.
 \label{eq:auxiliary-cardinalities}
\end{equation}
\end{proposition}

\begin{proof}
For a fixed value of \(U\), write the two rate bounds as
\[
 F_1=I(V_1;B_1)-I(V_1;V_2B_2),\qquad
 F_2=I(V_2;B_2)-I(V_2;V_1B_1).
\]
Let \(k\) range over the support of \(V_1\), and hold
\(q_k(v_2,x)=p(v_2,x|V_1=k)\) fixed.  Vary the weights
\(t_k=p(V_1=k)\) in the polytope
\[
 \mathcal T=\left\{t\ge0:
 \sum_k t_k q_k(x)=p_X(x)\ \text{for every }x\right\},
 \qquad q_k(x)=\sum_{v_2}q_k(v_2,x).
\]
Summing the constraints over \(x\) gives \(\sum_k t_k=1\).
Every extreme point of \(\mathcal T\) has at most
\(|\mathcal X|\) positive coordinates: otherwise the corresponding columns \(q_k(\cdot)\) are linearly dependent and admit a nonzero feasible perturbation in both directions.

Define the fixed conditional states
\[
 \sigma_k^{V_2B_j}
 =\sum_{v_2,x}q_k(v_2,x)\ketbra{v_2}\otimes\rho_x^{B_j},
 \qquad j=1,2,
\]
and let \(\sigma_k^{B_1}\) be the \(B_1\) marginal.  The cq entropy
identity \(H(V_1C)=H(t)+\sum_k t_kH(C)_{\sigma_k}\) gives
\begin{align*}
 F_1(t)
 &=H(B_1)+\sum_k t_k
   \bigl[H(V_2B_2)_{\sigma_k}-H(B_1)_{\sigma_k}\bigr]
   -H\!\left(\sum_k t_k\sigma_k^{V_2B_2}\right),\\
 F_2(t)
 &=H(B_2)+\sum_k t_k
   \bigl[H(V_2B_1)_{\sigma_k}-H(B_1)_{\sigma_k}\bigr]
   -H\!\left(\sum_k t_k\sigma_k^{V_2B_2}\right).
\end{align*}
The output marginal entropies are constant on \(\mathcal T\), since
\(p_X\) is fixed.  Entropy concavity therefore makes both \(F_1\) and
\(F_2\) convex in \(t\).  Decompose the original weights into extreme
points, \(t=\sum_a\lambda_a t^{(a)}\).  Then
\[
 F_j(t)\le\sum_a\lambda_a F_j(t^{(a)}),\qquad j=1,2.
\]
Record \(a\) in a refinement of \(U\).  The theorem uses the averaged
conditional bounds, so individual components need not have nonnegative
\(F_1\) and \(F_2\).  Both averaged bounds have increased or stayed the
same, and each component now has \(|\operatorname{supp}V_1|\le|\mathcal X|\).
Within each component, interchange \(V_1,B_1\) with \(V_2,B_2\) and repeat
the argument.  This reduces the support of \(V_2\) to at most
\(|\mathcal X|\) without enlarging that of \(V_1\).

Apply these steps for every original value of \(U\).  Relabel the surviving
symbols separately within each component into common alphabets of size
\(|\mathcal X|\); this is permitted by the \(U\)-dependent prefix in
\eqref{eq:factorization}.  The resulting averaged vector \((F_1,F_2)\)
lies in the convex hull of finitely many component vectors in
\(\mathbb R^2\).  Carath\'eodory's theorem preserves this vector using at
most three components, giving \(|\mathcal U|\le3\).  Every rate pair
allowed by the original distribution remains allowed, including the
inactive-user cases.  The reverse inclusion is immediate, and taking
closures proves the claim.
\end{proof}

\section{Regularized capacity and a quantum-auxiliary outer bound for confidential classical communication}
\label{reg-sec-capacity}

Applying the coding theorem to block channels gives an operational
multiletter characterization of the confidential capacity region.
The converse identifies the independent messages of an arbitrary code
with the block auxiliaries.  We then derive an outer bound expressed
through a single channel use and a quantum auxiliary, which generalizes the single-letter outer bound for classical broadcast channels from \cite{LiuMaricSpasojevicYates2008}. 

\subsection{Regularized capacity}
For
$n\geq1$, let $\mathcal S_n(W)$ be the union, over all finite independent
random variables $V_1,V_2$ and all stochastic maps
$q(x^n|v_1,v_2)$, of the nonnegative pairs $(r_1,r_2)$ satisfying
\begin{align}
 r_1&\leq I(V_1;B_1^n)_\omega-I(V_1;B_2^n|V_2)_\omega,
 \label{reg-eq-operational-one}\\
 r_2&\leq I(V_2;B_2^n)_\omega-I(V_2;B_1^n|V_1)_\omega,
 \label{reg-eq-operational-two}
\end{align}
where
\begin{align}
 \omega^{V_1V_2X^nB_1^nB_2^n}
 =\sum_{v_1,v_2,x^n}&p(v_1)p(v_2)q(x^n|v_1,v_2)
 \ketbra{v_1,v_2,x^n}\otimes\rho_{x^n}^{B_1^nB_2^n}.
 \label{reg-eq-operational-state}
\end{align}
Let $\mathcal R_n(W)$ be the corresponding confidential message region: it is the
union over finite classical $U,V_1,V_2$ and distributions
\begin{equation}
 p(u)p(v_1,v_2|u)q(x^n|u,v_1,v_2)
 \label{reg-eq-marton-law}
\end{equation}
of all nonnegative pairs satisfying
\begin{align}
 r_1&\leq I(V_1;B_1^n|U)-I(V_1;V_2B_2^n|U),
 \label{reg-eq-marton-one}\\
 r_2&\leq I(V_2;B_2^n|U)-I(V_2;V_1B_1^n|U).
 \label{reg-eq-marton-two}
\end{align}
The state in \cref{reg-eq-marton-one,reg-eq-marton-two} is the one
induced by \cref{reg-eq-marton-law} and
$x^n\mapsto\rho_{x^n}^{B_1^nB_2^n}$.

\begin{theorem}[Regularized capacity]
\label{reg-thm-exact-capacity}
For every finite-input memoryless cq broadcast channel $W$,
\begin{equation}
 \mathcal C_{\mathrm{conf}}(W)
 =\overline{\bigcup_{n\geq1}\frac1n\mathcal S_n(W)}
 =\overline{\bigcup_{n\geq1}\frac1n\mathcal R_n(W)},
 \label{reg-eq-exact-capacity}
\end{equation}
where $\mathcal C_{\mathrm{conf}}(W)$ denotes the average-error capacity region
under conditional strong secrecy.
\end{theorem}

\begin{proof}
Fix $n$ and a state entering $\mathcal R_n(W)$.  Regard
$W^{\otimes n}$ as one cq super-channel with input alphabet
$\mathcal X^n$ and output $B_1^nB_2^n$.  Applying the inner bound from \cref{thm:main}
to independent repetitions of this super-channel achieves every point of
$\mathcal R_n(W)$ up to an arbitrarily small reduction in each positive
rate; \cref{prop:inactive-user} covers the case when one rate is zero.  Since one super-channel
use consists of $n$ physical uses,
\begin{equation}
 \overline{\bigcup_{n\geq1}\frac1n\mathcal R_n(W)}
 \subseteq\mathcal C_{\mathrm{conf}}(W).
 \label{reg-eq-achievable-inclusion}
\end{equation}

For the converse, consider an arbitrary $n$-block code with independent
uniform messages.  Let $\widehat M_i$ be receiver $i$'s decoder output and
put
\begin{align*}
 e_{i,n}&=\Pr\{\widehat M_i\neq M_i\},\\
 \ell_{1,n}&=I(M_1;B_2^n|M_2)_\Omega,\qquad
 \ell_{2,n}=I(M_2;B_1^n|M_1)_\Omega,\\
 f_{i,n}&=h_2(e_{i,n})
 +e_{i,n}\log\max\{1,|\cM_i|-1\}.
\end{align*}
Data processing followed by Fano's inequality gives
\begin{equation}
 H(M_i|B_i^n)_\Omega
 \leq H(M_i|\widehat M_i)\leq f_{i,n}.
 \label{reg-eq-fano}
\end{equation}
Adding and subtracting the corresponding leakage therefore gives
\begin{align}
 \log|\cM_1|
 &\leq I(M_1;B_1^n)_\Omega-I(M_1;B_2^n|M_2)_\Omega
       +f_{1,n}+\ell_{1,n},
 \label{reg-eq-converse-one}\\
 \log|\cM_2|
 &\leq I(M_2;B_2^n)_\Omega-I(M_2;B_1^n|M_1)_\Omega
       +f_{2,n}+\ell_{2,n}.
 \label{reg-eq-converse-two}
\end{align}
Choose $V_1=M_1$, $V_2=M_2$, and let $q(x^n|v_1,v_2)$ be the
actual stochastic encoder.  The two differences in
\cref{reg-eq-converse-one,reg-eq-converse-two} are thus evaluated on one
common state of the form \cref{reg-eq-operational-state}.

The rates of every reliable sequence are bounded.  Indeed,
\cref{reg-eq-fano} and $I(M_i;B_i^n)\leq n\log\dim B_i$ imply
\begin{equation}
 (1-e_{i,n})\log|\cM_i|
 \leq n\log\dim B_i+h_2(e_{i,n}).
 \label{reg-eq-rate-bounded}
\end{equation}
Consequently $f_{i,n}/n\to0$.  Conditional strong secrecy gives
$\ell_{i,n}\to0$.  If both limiting rates are positive, then, for all
sufficiently large $n$,
\[
 r_{i,n}:=\log|\cM_i|-f_{i,n}-\ell_{i,n}\ge0,\qquad i=1,2.
\]
Equations~\eqref{reg-eq-converse-one} and
\eqref{reg-eq-converse-two} show that
$(r_{1,n},r_{2,n})\in\mathcal S_n(W)$.  Dividing by $n$ and taking the
closure yields
\begin{equation}
 \mathcal C_{\mathrm{conf}}(W)
 \subseteq\overline{\bigcup_{n\geq1}\frac1n\mathcal S_n(W)}.
 \label{reg-eq-converse-inclusion}
\end{equation}
If one limiting rate is zero, the corresponding message may be replaced by
a singleton.  Averaging the other receiver's conditional error and leakage
over the unused message selects a fixed value for which both still vanish.
The associated auxiliary is then constant, and its endpoint in
\cref{reg-eq-operational-one,reg-eq-operational-two} equals zero.  This
also covers the boundary axes in \cref{reg-eq-converse-inclusion}.
Finally, $\mathcal S_n(W)\subseteq\mathcal R_n(W)$ by taking $U$
constant and using $I(V_1;V_2)=0$.  Combining this with
\cref{reg-eq-achievable-inclusion,reg-eq-converse-inclusion} proves
\cref{reg-eq-exact-capacity}.
\end{proof}

We next derive a one-use outer inequality with an unrestricted quantum
auxiliary.  

\subsection{An outer bound with a quantum auxiliary}
Let
$\mathfrak Q(W)$ be the class of states
\begin{align}
 \sigma^{UV_1V_2XB_1B_2}
 =\sum_{v_1,v_2,x}&p(v_1)p(v_2)p(x|v_1,v_2)
 \ketbra{v_1,v_2,x}\nonumber\\[-1mm]
 &\otimes\sigma_{v_1,v_2,x}^{U}\otimes\rho_x^{B_1B_2},
 \label{reg-eq-quantum-aux-state}
\end{align}
where $U$ is any finite-dimensional quantum system.  Thus the current
channel output is conditionally independent of $UV_1V_2$ given $X$.
The systems $U$ and $X$ need not be conditionally independent given
$(V_1,V_2)$.

\begin{lemma}[Quantum sum identity]
\label{reg-lem-csiszar}
For every state on $ADB_1^nB_2^n$, define
$U_i=(B_1^{i-1},B_{2,i+1}^n)$.  Then
\begin{align}
 I(A;B_1^n|D)-I(A;B_2^n|D)
 =\sum_{i=1}^n\bigl[&I(A;B_{1,i}|D,U_i)\nonumber\\[-1mm]
                    &-I(A;B_{2,i}|D,U_i)\bigr].
 \label{reg-eq-csiszar}
\end{align}
\end{lemma}

\begin{proof}
The entropy chain rule gives
\begin{align*}
 \sum_{i=1}^n\bigl[&H(B_{1,i}|D,B_1^{i-1},B_{2,i+1}^n)
 -H(B_{2,i}|D,B_1^{i-1},B_{2,i+1}^n)\bigr]\\
 &=H(B_1^n|D)-H(B_2^n|D).
\end{align*}
The same identity holds with $AD$ in place of $D$.  Subtracting the two
identities and using the definition of conditional mutual information
proves \cref{reg-eq-csiszar}.
\end{proof}

\begin{theorem}[Common-quantum-auxiliary outer bound]
\label{reg-thm-quantum-aux-outer}
The region $\mathcal C_{\mathrm{conf}}(W)$ is contained in the closure of the union,
over $\sigma\in\mathfrak Q(W)$, of the nonnegative pairs satisfying
\begin{align}
 R_1\leq\min\bigl\{&I(V_1;B_1|U)-I(V_1;B_2|U),\nonumber\\[-1mm]
                   &I(V_1;B_1|V_2U)-I(V_1;B_2|V_2U)\bigr\},
 \label{reg-eq-outer-one}\\
 R_2\leq\min\bigl\{&I(V_2;B_2|U)-I(V_2;B_1|U),\nonumber\\[-1mm]
                   &I(V_2;B_2|V_1U)-I(V_2;B_1|V_1U)\bigr\}.
 \label{reg-eq-outer-two}
\end{align}
All four quantities are evaluated on the same state $\sigma$.
\end{theorem}

\begin{proof}
Using the notation from the converse proof of \cref{reg-thm-exact-capacity}, recall that we have \eqref{reg-eq-converse-one} and \eqref{reg-eq-converse-two}:
\begin{align}
 \log|\cM_1|
 &\leq I(M_1;B_1^n)_\Omega-I(M_1;B_2^n|M_2)_\Omega
       +f_{1,n}+\ell_{1,n},\\
 \log|\cM_2|
 &\leq I(M_2;B_2^n)_\Omega-I(M_2;B_1^n|M_1)_\Omega
       +f_{2,n}+\ell_{2,n}.
\end{align}

Giving receiver 1 the side information $M_2$ cannot worsen their decoder, hence
\begin{equation}
 \log|\cM_1|
 \leq I(M_1;B_1^n|M_2)-I(M_1;B_2^n|M_2)
       +f_{1,n}+\ell_{1,n},
 \label{reg-eq-conditioned-one}
\end{equation}
and the symmetric inequality conditioned on $M_1$ also holds.

Use Lemma~\ref{reg-lem-csiszar} with $A=M_1$ and, respectively, $D$ empty and
$D=M_2$. Since the messages are independent,
\[
 I(M_1;B_2^n)\le I(M_1;B_2^n|M_2),\qquad
 I(M_2;B_1^n)\le I(M_2;B_1^n|M_1).
\]
Replacing the subtracted leakage terms by these unconditional
quantities gives the corresponding unconditioned
mutual-information differences. For message two, take the negative of
\cref{reg-eq-csiszar}, with $A=M_2$ and, respectively, $D$ empty and
$D=M_1$.  All four applications use the same
\begin{equation}
 U_i=(B_1^{i-1},B_{2,i+1}^n).
 \label{reg-eq-common-history}
\end{equation}
Let $T$ be uniform on $\{1,\ldots,n\}$ and independent of the code, and set
\begin{equation}
 U=(T,U_T),\quad V_1=M_1,\quad V_2=M_2,
 \quad X=X_T,\quad B_j=B_{j,T}.
 \label{reg-eq-time-shared-outer-state}
\end{equation}
Formally, \(U\) is the classical-\(T\)-controlled direct sum of the
different history Hilbert spaces; this is a single finite-dimensional
quantum register for every fixed blocklength.
The four normalized sums supplied by Lemma~\ref{reg-lem-csiszar} are
the four differences in \cref{reg-eq-outer-one,reg-eq-outer-two}.

It remains to verify the state class.  Conditioned on
$(V_1,V_2,X)=(v_1,v_2,x)$, memorylessness makes the history in
\cref{reg-eq-common-history} and the current output a product
$\sigma_{v_1,v_2,x}^{U}\otimes\rho_x^{B_1B_2}$.  Thus the time-shared
state has the form \cref{reg-eq-quantum-aux-state}.  Divide the four
finite-block inequalities by $n$, use
$f_{i,n}/n\to0$ and $\ell_{i,n}/n\to0$.  When both limiting rates are
positive, every \((R_1-\eta,R_2-\eta)\), for fixed \(\eta>0\), lies in
one of the displayed finite-block rectangles for all sufficiently large
blocklengths.  For an axis point, replace the unused message by a singleton:
averaging over that message selects a fixed value for which the surviving
receiver's conditional error and leakage still vanish.  The unused
auxiliary is then constant, so both of its outer-bound differences vanish.
Taking \(\eta\downarrow0\) and the closure proves the claim
on the full nonnegative quadrant.
\end{proof}

\begin{remark}
The register $U$ includes quantum output histories and its dimension may grow
with $n$.  Accordingly, \cref{reg-thm-quantum-aux-outer} is stated with an
unrestricted finite-dimensional quantum auxiliary.
\end{remark}

\section{Applications to cq broadcast channels}
\label{sec:applications-cq}
\label{det-sec-deterministic-cq}

We apply the coding theorem to two classes of channels for which the
confidential-message capacity region can be characterized. We first
consider degraded cq broadcast channels, then deterministic classical
broadcast channels and the classical Blackwell example. Throughout this
section, the encoder chooses symbols from a fixed classical alphabet and
each message must remain secret from the unintended receiver, who knows
its own message.

\subsection{Degraded cq broadcast channels}
\label{sec:degr-cq}

Let \(W_i:x\mapsto\rho_x^{B_i}\), \(i=1,2\), denote the marginal
channels of \(W:x\mapsto\rho_x^{B_1B_2}\). We call \(W\)
\emph{degraded with stronger receiver \(B_1\)} if there is a channel
\(\mathcal D:B_1\to B_2\) such that
\begin{equation}
 \rho_x^{B_2}=\mathcal D(\rho_x^{B_1})
 \qquad\text{for every }x\in\mathcal X.
 \label{eq:degr-cq-condition}
\end{equation}
This condition concerns the two marginal channels. Receiver $B_1$ can
simulate receiver $B_2$'s observations and hence its decoder. As a result,
receiver $B_2$ cannot receive a positive-rate message that remains secret
from receiver $B_1$; only the first confidential message can be active.

\begin{theorem}[Confidential capacity of a degraded cq broadcast channel]
\label{thm:degr-cq-capacity}
If \eqref{eq:degr-cq-condition} holds, then the conditional
strong-secrecy capacity region is
\begin{equation}
 \mathcal C_{\mathrm{conf}}(W)
 =
 \left\{
 (R_1,R_2):
 R_2=0,\quad
 0\le R_1\le
 \max_{p_X}\bigl[I(X;B_1)-I(X;B_2)\bigr]
 \right\}.
 \label{eq:degr-cq-capacity}
\end{equation}
The mutual informations in \eqref{eq:degr-cq-capacity} are evaluated on
\(\sum_x p_X(x)\ketbra{x}\otimes\rho_x^{B_1B_2}\).
\end{theorem}

\begin{proof}
For achievability, take \(U\) and \(V_2\) to be constant and set
\(V_1=X\) in the inactive-user specialization
\cref{prop:inactive-user}. Its bound becomes
\[
 R_1<I(X;B_1)-I(X;B_2),
\]
while no message is sent to \(B_2\). Stochastic prefixing cannot enlarge
the displayed maximum. Indeed, for every stochastic prefix \(V-X\),
the chain rule gives
\begin{align}
 I(X;B_1)-I(X;B_2)
 ={}& I(V;B_1)-I(V;B_2)\nonumber\\
 &+I(X;B_1|V)-I(X;B_2|V),
 \label{eq:degr-prefix-elimination}
\end{align}
and the last difference is nonnegative by conditional data processing.

For the converse, let a sequence of codes have decoding
errors tending to zero and satisfy the conditional secrecy criteria, and
write \(R_{i,n}=n^{-1}\log|\cM_i|\). Reliability and
\(I(M_i;B_i^n)\le n\log\dim B_i\) imply
\(\log|\cM_i|=O(n)\), so all Fano remainders below are \(o(n)\). Since
\(B_1\) can apply \(\mathcal D^{\otimes n}\) and then use receiver~2's
decoder, it can produce an estimate of \(M_2\) with the same error
probability as receiver~2. Fano's inequality and data processing therefore
give
\begin{align}
 nR_{2,n}
 &=H(M_2|M_1)\nonumber\\
 &\le I(M_2;B_1^n|M_1)+o(n)=o(n),
 \label{eq:degr-r2-converse}
\end{align}
where the final equality is secrecy of \(M_2\) from receiver~1. Hence
\(R_2=0\).

For the other rate, reliability at receiver~1 and secrecy from receiver~2
imply
\begin{align}
 nR_{1,n}
 &\le I(M_1;B_1^n|M_2)-I(M_1;B_2^n|M_2)+o(n).
 \label{eq:degr-r1-start}
\end{align}
Fix a value of \(M_2\), and let \(X^n\) be the random channel input produced
by the stochastic encoder. The conditional \(B_2^n\)-output is obtained
from the corresponding \(B_1^n\)-output by \(\mathcal D^{\otimes n}\).
The chain rule and conditional data processing therefore yield
\begin{align}
 I(M_1;B_1^n)-I(M_1;B_2^n)
 \le I(X^n;B_1^n)-I(X^n;B_2^n).
 \label{eq:degr-message-to-input}
\end{align}
It remains to single-letterize the expression on the right.

Let \(\mathcal V:B_1\to\widetilde B_2F\) be a Stinespring isometry of
\(\mathcal D\), where the marginal states on \(\widetilde B_2\) equal
the corresponding states on \(B_2\). For an arbitrary distribution on
\(X^n\), let \(\overline\omega\) be the state obtained after applying
\(\mathcal V^{\otimes n}\) to the average \(B_1^n\)-output. Since a
conditional output given \(x^n\) is a product state,
\begin{align}
 &I(X^n;B_1^n)-I(X^n;\widetilde B_2^n)\nonumber\\
 &=S(F^n|\widetilde B_2^n)_{\overline\omega}
   -\sum_{x^n}p(x^n)
      S(F^n|\widetilde B_2^n)_{\omega_{x^n}}\nonumber\\
 &\le \sum_{t=1}^n S(F_t|\widetilde B_{2,t})_{\overline\omega}
   -\sum_{t=1}^n\sum_{x_t}p_t(x_t)
      S(F_t|\widetilde B_{2,t})_{\omega_{x_t}}\nonumber\\
 &=\sum_{t=1}^n
   \bigl[I(X_t;B_{1,t})-I(X_t;B_{2,t})\bigr]\nonumber\\
 &\le n\max_{p_X}\bigl[I(X;B_1)-I(X;B_2)\bigr].
 \label{eq:degr-cqq-single-letterization}
\end{align}
Here we used
\(S(F^n|\widetilde B_2^n)\le
  \sum_t S(F_t|\widetilde B_{2,t})\), which follows by repeated strong
subadditivity, whereas the conditional state corresponding to a fixed
\(x^n\) is a tensor product and hence gives equality. Averaging
\eqref{eq:degr-message-to-input} over \(M_2\), inserting
\eqref{eq:degr-cqq-single-letterization} into
\eqref{eq:degr-r1-start}, and taking \(n\to\infty\) proves the converse.
This is the standard single-letterization for a degraded cqq wiretap
channel; see also \cite{WinterDegraded2016}.
\end{proof}

\begin{remark}[Relation to degradable quantum-input channels]
\label{rem:degr-cq-quantum-input}
The theorem concerns a fixed classical input alphabet and secrecy from
receiver $B_2$. The quantum-input model in
\cref{fq-sec-operational-model} allows arbitrary quantum preparations
and requires secrecy from the joint system $C=B_2E$, where $E$ is the
Stinespring environment. Its degraded case assumes that the full complementary channel can be simulated from receiver $B_1$:
$\mathcal N_1^c=\mathcal D\circ\mathcal N_1$ on every input state.
Under this condition, \cref{fq-cor-degraded} identifies the
confidential classical capacity with the quantum capacity
$Q(\mathcal N_1)$ of the degradable marginal channel. The messages in
that corollary are still classical.

The connection between the two formulas follows from the entropy
identities for an input ensemble $\{p(x),\theta_x\}$. Define
$\bar\theta=\sum_xp(x)\theta_x$ and
\[
 I_c(\theta,\mathcal N_1)
 :=H(\mathcal N_1(\theta))-H(\mathcal N_1^c(\theta)).
\]
For the induced cq state on $XB_1C$,
\begin{equation}
 I(X;B_1)-I(X;C)
 =I_c(\bar\theta,\mathcal N_1)
  -\sum_xp(x)I_c(\theta_x,\mathcal N_1).
 \label{eq:degr-ensemble-coherent-information}
\end{equation}
If every $\theta_x$ is pure, its Stinespring output on $B_1C$ is pure,
so each term in the sum vanishes. A pure-state decomposition of an
input maximizing coherent information therefore attains
$Q(\mathcal N_1)$ in the quantum-input model. For a fixed cq alphabet,
the available preparations need not include such an ensemble and
$B_2$ need not be a complementary output. Its capacity is given by
\eqref{eq:degr-cq-capacity}.
\end{remark}

\subsection{Deterministic classical broadcast channels}
\label{det-subsec-classical}

Fix maps $f_i:\mathcal X\to\mathcal Y_i$, $i=1,2$, and consider the
classical channel:
\begin{equation}
 \mathcal W:x\longmapsto(f_1(x),f_2(x)).
 \label{det-eq-channel}
\end{equation}
To apply our cq coding theorem, represent its outputs by the orthogonal
states
\[
 \ketbra{f_1(x)}^{B_1}\otimes\ketbra{f_2(x)}^{B_2}.
\]
This representation is operationally equivalent to the classical channel.
The classical confidential-message region was characterized in
\cite{CaiLam2000}; see also \cite{GoldfeldKramerPermuter2017}. We recover
it under conditional strong secrecy from the main coding theorem.

For an input distribution \(p_X\), let
\begin{equation}
 Y_1=f_1(X),\qquad Y_2=f_2(X),
 \label{det-eq-output-rvs}
\end{equation}
and define the downward-closed rectangle
\begin{equation}
 \mathcal R_{\mathrm{det}}(p_X)
 :=\left\{(R_1,R_2):
 \begin{array}{l}
 0\le R_1\le H_{p_X}(Y_1|Y_2),\\
 0\le R_2\le H_{p_X}(Y_2|Y_1)
 \end{array}\right\}.
 \label{det-eq-region-p}
\end{equation}

\begin{theorem}[Confidential classical capacity of a deterministic broadcast channel]
\label{det-thm-capacity}
The confidential classical capacity region of the channel
\eqref{det-eq-channel} is
\begin{equation}
 \mathcal C_{\mathrm{conf}}(\mathcal W)
 =\bigcup_{p_X}\mathcal R_{\mathrm{det}}(p_X).
 \label{det-eq-capacity}
\end{equation}
The union in \eqref{det-eq-capacity} is closed and convex, so an additional
time-sharing variable is unnecessary.
\end{theorem}

\begin{proof}
\emph{Achievability.}
Fix \(p_X\), take \(U\) to be constant in \cref{thm:main}, and use the joint
law
\begin{equation}
 p(v_1,v_2,x)
 =p_X(x)\,\mathbf 1\{v_1=f_1(x),\ v_2=f_2(x)\}.
 \label{det-eq-auxiliary-law}
\end{equation}
Equivalently, \((V_1,V_2)=(Y_1,Y_2)\), with the generally correlated joint
law induced by \(p_X\), and \(p(x|v_1,v_2)\) is the corresponding conditional
distribution. Under \eqref{det-eq-channel}, \(B_i\) is a classical copy of
\(V_i\). Consequently,
\begin{align}
 I(V_1;B_1)-I(V_1;V_2)-I(V_1;B_2|V_2)
 &=H(Y_1)-I(Y_1;Y_2)\nonumber\\
 &=H(Y_1|Y_2),
 \label{det-eq-ach-one}
\end{align}
because \(I(V_1;B_2|V_2)=0\). The symmetric calculation gives
\begin{equation}
 I(V_2;B_2)-I(V_1;V_2)-I(V_2;B_1|V_1)
 =H(Y_2|Y_1).
 \label{det-eq-ach-two}
\end{equation}
If both conditional entropies are positive, every interior point of
\(\mathcal R_{\mathrm{det}}(p_X)\) is achievable by the
likelihood-selected code of \cref{thm:main}. If either conditional entropy
vanishes, the corresponding message is inactive and
\cref{prop:inactive-user} gives the remaining axis. Taking the closure
therefore gives the entire rectangle in every case.

\emph{Converse.}
Consider an arbitrary blocklength-\(n\) stochastic code. Let \(M_1,M_2\) be
independent and uniform, let \(X^n\) be the channel input produced by the
encoder, and let \(Y_i^n=f_i^{\otimes n}(X^n)\). Denote the two average
decoding errors by \(\epsilon_{i,n}\), and put
\begin{equation}
 \delta_{1,n}:=I(M_1;Y_2^n|M_2),\qquad
 \delta_{2,n}:=I(M_2;Y_1^n|M_1).
 \label{det-eq-leakages}
\end{equation}
Reliability and the finite output alphabets first imply
\begin{equation}
 (1-\epsilon_{i,n})\log|\mathcal M_i|
 \le n\log|\mathcal Y_i|+h_2(\epsilon_{i,n}),
 \label{det-eq-rate-bound}
\end{equation}
so \(\log|\mathcal M_i|=O(n)\). Fano's inequality then supplies
quantities \(\eta_{i,n}=o(n)\) such that
\begin{equation}
 H(M_i|Y_i^n,M_{\bar i})\le H(M_i|Y_i^n)\le\eta_{i,n},
 \qquad \bar 1=2,\ \bar 2=1.
 \label{det-eq-fano}
\end{equation}
Giving $Y_2^n$ to receiver $1$ as additional information and using the
secrecy constraint gives
\begin{align}
 \log|\mathcal M_1|
 &=H(M_1|M_2)\notag\\
 &\le I(M_1;Y_1^nY_2^n|M_2)+\eta_{1,n}\notag\\
 &=I(M_1;Y_2^n|M_2)
   +I(M_1;Y_1^n|Y_2^nM_2)+\eta_{1,n}\notag\\
 &\le H(Y_1^n|Y_2^n)+\delta_{1,n}+\eta_{1,n}\notag\\
 &\le\sum_{t=1}^n H(Y_{1,t}|Y_{2,t})
      +\delta_{1,n}+\eta_{1,n}.
 \label{det-eq-converse-one}
\end{align}
The last step follows from the chain rule and conditioning. Applying the
same argument to receiver 2 under this code-induced input law gives
\begin{equation}
 \log|\mathcal M_2|
 \le \sum_{t=1}^n H(Y_{2,t}|Y_{1,t})
      +\delta_{2,n}+\eta_{2,n}.
 \label{det-eq-converse-two}
\end{equation}

Let \(p_t(x):=\Prb\{X_t=x\}\) be the \(t\)-th input marginal induced by
this same code, and define their average
\begin{equation}
 \bar p_X(x):=\frac1n\sum_{t=1}^n p_t(x).
 \label{det-eq-average-input}
\end{equation}
Classical conditional entropy is concave in the joint distribution. Since
the joint law of \((Y_1,Y_2)\) depends linearly on \(p_X\), this gives the two
simultaneous estimates
\begin{align}
 \frac1n\sum_{t=1}^n H_{p_t}(Y_1|Y_2)
 &\le H_{\bar p_X}(Y_1|Y_2),
 \label{det-eq-concavity-one}\\
 \frac1n\sum_{t=1}^n H_{p_t}(Y_2|Y_1)
 &\le H_{\bar p_X}(Y_2|Y_1).
 \label{det-eq-concavity-two}
\end{align}
Both inequalities involve the same \(\bar p_X\).
Dividing \eqref{det-eq-converse-one}--\eqref{det-eq-converse-two} by \(n\),
using \(\epsilon_{i,n}\to0\) and \(\delta_{i,n}\to0\), and taking a
convergent subsequence of \(\bar p_X\) proves the converse.

The right-hand side of \eqref{det-eq-capacity} is convex.
Indeed, if \(r\in\mathcal R_{\mathrm{det}}(p)\),
\(s\in\mathcal R_{\mathrm{det}}(q)\), and \(0\le\lambda\le1\), concavity
implies that \(\lambda r+(1-\lambda)s\) lies in
\(\mathcal R_{\mathrm{det}}(\lambda p+(1-\lambda)q)\). It is closed because
the input simplex is compact and both conditional entropies are continuous.
\end{proof}

\begin{remark}[Injective deterministic realization]
\label{det-rem-injective}
If \(x\mapsto(f_1(x),f_2(x))\) is injective, then \(X\) is a function of
\((Y_1,Y_2)\), and therefore
\begin{equation}
 H(Y_1|Y_2)=H(X|Y_2),\qquad
 H(Y_2|Y_1)=H(X|Y_1).
 \label{det-eq-injective-entropies}
\end{equation}
Without injectivity, the expressions in \cref{det-thm-capacity} should remain
in terms of \(Y_1,Y_2\); different input symbols producing the same output
pair are operationally indistinguishable.
\end{remark}

\subsection{The classical Blackwell broadcast channel}
\label{det-subsec-blackwell}

The classical Blackwell channel \cite{BlackwellCh,vanDerMeulen1977,Gelfand1977}, denoted by \(\mathcal W_{\mathrm{BW}}\),
has the input alphabet \(\{0,1,2\}\) and the deterministic output map
\begin{equation}
 0\longmapsto(0,0),\qquad
 1\longmapsto(0,1),\qquad
 2\longmapsto(1,0).
 \label{det-eq-blackwell-map}
\end{equation}

Write
\begin{equation}
 p_0:=\Prb\{X=0\},\qquad p_1:=\Prb\{X=1\},\qquad
 p_2:=1-p_0-p_1.
 \label{det-eq-blackwell-p}
\end{equation}
For \(Y_2=1\), the value of \(Y_1\) is fixed. Conditional on \(Y_2=0\),
which has probability \(1-p_1=p_0+p_2\), the probabilities of \(Y_1=0,1\)
are \(p_0/(1-p_1)\) and \(p_2/(1-p_1)\), respectively. Hence
\begin{align}
 H(Y_1|Y_2)
 &=(1-p_1)h_2\!\left(\frac{p_0}{1-p_1}\right),
 \label{det-eq-blackwell-rate-one}\\
 H(Y_2|Y_1)
 &=(p_0+p_1)h_2\!\left(\frac{p_0}{p_0+p_1}\right),
 \label{det-eq-blackwell-rate-two}
\end{align}
where a term with a zero prefactor is understood as zero. Thus
\cref{det-thm-capacity} gives the full confidential-message capacity region
\begin{equation}
 \mathcal C_{\mathrm{conf}}(\mathcal W_{\mathrm{BW}})
 =\bigcup_{p_X}
 [0,H_{p_X}(Y_1|Y_2)]\times[0,H_{p_X}(Y_2|Y_1)],
 \label{det-eq-blackwell-region}
\end{equation}
with the two entropies given by
\eqref{det-eq-blackwell-rate-one}--\eqref{det-eq-blackwell-rate-two}.
This union is already closed and convex by \cref{det-thm-capacity}.

The two axis endpoints have simple zero-error codes. Using inputs $0,2$
with equal probabilities sends one perfectly confidential bit to receiver
$1$, since receiver $2$ always observes $0$. Using inputs $0,1$ does the
same for receiver $2$. Thus $(1,0)$ and $(0,1)$ are achievable. Each
individual rate is at most one bit per channel use because the respective
receiver has a binary output. Consequently, these are the individual
confidential-capacity maxima.

\begin{proposition}[Maximum confidential sum rate of the classical Blackwell channel]
\label{det-prop-blackwell-sum}
Let \(\varphi=(1+\sqrt5)/2\) be the golden ratio. The maximum sum rate in
the confidential-capacity region of the classical Blackwell channel is
\begin{equation}
 \max_{p_X}\bigl[H(Y_1|Y_2)+H(Y_2|Y_1)\bigr]
 =2\log\varphi.
 \label{det-eq-blackwell-sum}
\end{equation}
It is attained at
\begin{equation}
 p_0=\frac1{\sqrt5},\qquad
 p_1=p_2=\frac{1-1/\sqrt5}{2},
 \label{det-eq-blackwell-optimizer}
\end{equation}
where
\begin{equation}
 H(Y_1|Y_2)=H(Y_2|Y_1)=\log\varphi.
 \label{det-eq-blackwell-symmetric-rate}
\end{equation}
Numerically, the two rates are approximately \(0.694242\) bits per use,
and their sum is approximately \(1.388484\) bits per use.
\end{proposition}

\begin{proof}
Set
\begin{equation}
 S(p_0,p_1,p_2):=H(Y_1|Y_2)+H(Y_2|Y_1).
 \label{det-eq-blackwell-objective}
\end{equation}
The function \(S\) is concave in the input distribution and invariant under
interchanging \(p_1\) and \(p_2\). Symmetrizing any input distribution
therefore cannot decrease \(S\), so a maximizer may be chosen with
\begin{equation}
 p_1=p_2=\frac{1-p_0}{2}.
 \label{det-eq-blackwell-symmetry}
\end{equation}
Writing \(t=p_0\), the two conditional entropies coincide and their sum is
\begin{equation}
 F(t)=(1+t)h_2\!\left(\frac{2t}{1+t}\right),
 \qquad 0\le t\le1.
 \label{det-eq-blackwell-one-variable}
\end{equation}
For \(0<t<1\), direct differentiation gives
\begin{equation}
 F'(t)=\log\!\left(\frac{1-t^2}{4t^2}\right).
 \label{det-eq-blackwell-derivative}
\end{equation}
Thus \(F\) increases up to \(t=1/\sqrt5\) and decreases thereafter. At
this value,
\begin{equation}
 \frac{1+t}{2}=\frac{\varphi}{\sqrt5},
 \qquad
 \frac{2t}{1+t}=\frac1\varphi,
 \label{det-eq-blackwell-golden-identities}
\end{equation}
and \(1-1/\varphi=1/\varphi^2\). Hence
\begin{align}
 \frac{1+t}{2}h_2\!\left(\frac1\varphi\right)
 &=\frac{\varphi}{\sqrt5}
   \left(\frac1\varphi+\frac2{\varphi^2}\right)\log\varphi\nonumber\\
 &=\log\varphi,
 \label{det-eq-blackwell-golden-evaluation}
\end{align}
where the final equality uses \(1+2/\varphi=\sqrt5\). This proves
\eqref{det-eq-blackwell-sum}--\eqref{det-eq-blackwell-symmetric-rate}.
\end{proof}

Figure~\ref{fig:blackwell-cq-confidential} depicts the capacity region
\eqref{det-eq-blackwell-region}. The boundary is sampled by maximizing
$wH(Y_1|Y_2)+(1-w)H(Y_2|Y_1)$ over $p_X$ for different
$0\le w\le1$. Thus the shaded region represents the full capacity formula,
with a numerical rendering of its curved boundary.

\begin{figure}[H]
\centering
\pgfplotstableread{
Rone Rtwo
0.00000000 1.00000000
0.00001950 0.99999905
0.00133473 0.99990311
0.01019904 0.99902332
0.03299736 0.99609139
0.07009533 0.99013835
0.11771578 0.98087009
0.17128661 0.96849020
0.22707328 0.95342826
0.28252238 0.93615526
0.33608943 0.91709127
0.38695588 0.89657095
0.43478081 0.87483837
0.47951845 0.85205394
0.52129458 0.82830517
0.56032658 0.80361716
0.59687317 0.77796158
0.63120363 0.75126378
0.66357919 0.72340808
0.69424191 0.69424191
0.72340808 0.66357919
0.75126378 0.63120363
0.77796158 0.59687317
0.80361716 0.56032658
0.82830517 0.52129458
0.85205394 0.47951845
0.87483837 0.43478081
0.89657095 0.38695588
0.91709127 0.33608943
0.93615526 0.28252238
0.95342826 0.22707328
0.96849020 0.17128661
0.98087009 0.11771578
0.99013835 0.07009533
0.99609139 0.03299736
0.99902332 0.01019904
0.99990311 0.00133473
0.99999905 0.00001950
1.00000000 0.00000000
}\BlackwellBoundary

\begin{tikzpicture}
\begin{axis}[
  width=0.72\linewidth,
  height=0.61\linewidth,
  axis lines=left,
  xmin=0, xmax=1.04,
  ymin=0, ymax=1.04,
  xlabel={$R_1$ (bits/use)},
  ylabel={$R_2$ (bits/use)},
  xtick={0,0.25,0.5,0.75,1},
  ytick={0,0.25,0.5,0.75,1},
  tick label style={font=\small},
  label style={font=\small},
  clip=false
]
\addplot[
  name path=BlackwellCurve,
  draw=none
] table[x=Rone,y=Rtwo] {\BlackwellBoundary};

\addplot[
  name path=BlackwellBase,
  draw=none
] coordinates {(0,0) (1,0)};

\addplot[
  blue!18
] fill between[
  of=BlackwellCurve and BlackwellBase
];

\addplot[
  very thick,
  blue!60!black,
  line join=round
] table[x=Rone,y=Rtwo] {\BlackwellBoundary};

\addplot[
  densely dashed,
  gray!65
] coordinates {
  (0.69424191,0)
  (0.69424191,0.69424191)
  (0,0.69424191)
};

\addplot[
  only marks,
  mark=*,
  mark size=2.4pt,
  black
] coordinates {(0,1) (1,0) (0.69424191,0.69424191)};

\node[
  anchor=south west,
  font=\small,
  inner sep=2pt
] at (axis cs:0.704,0.704)
  {$(\log_2\varphi,\log_2\varphi)$};

\node[
  font=\small,
  text=blue!55!black
] at (axis cs:0.43,0.34)
  {$\mathcal C_{\mathrm{conf}}(\mathcal W_{\mathrm{BW}})$};
\end{axis}
\end{tikzpicture}
\caption{Confidential classical capacity  region of the classical Blackwell
channel. The endpoints $(1,0)$ and $(0,1)$ attain one perfectly
confidential bit per use for the respective receiver. The marked
symmetric point attains the maximum sum rate $2\log_2\varphi$.}
\label{fig:blackwell-cq-confidential}
\end{figure}
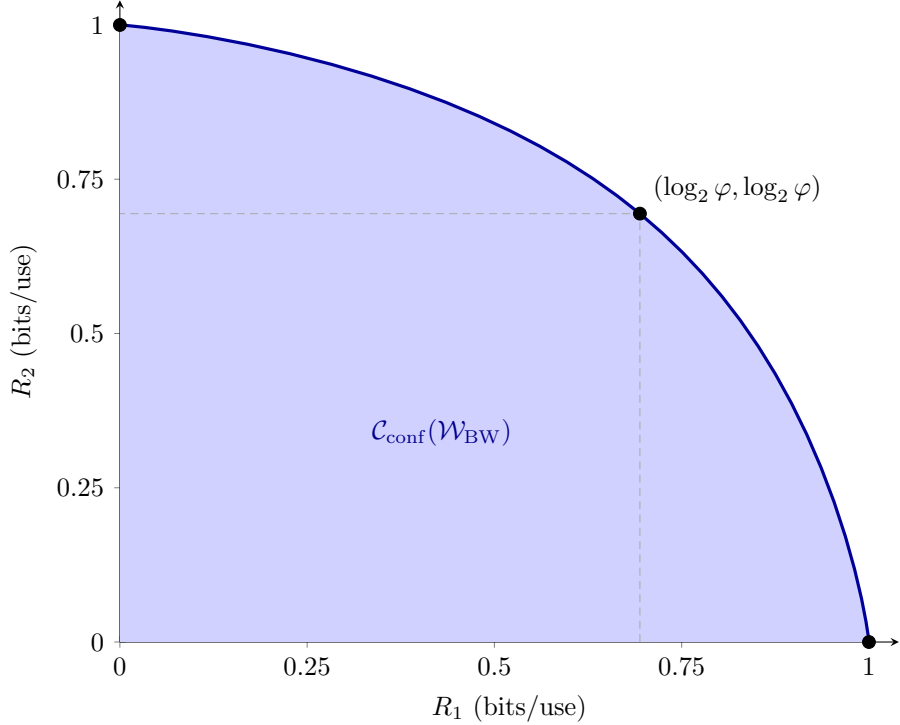

The coherent extension of the Blackwell map, which accepts arbitrary
qutrit states, will be considered separately in \cref{bw-q-sec}, where we compare its mutually confidential classical and quantum capacity regions.

\section{Confidential classical communication with quantum inputs}
\label{fq-sec-quantum-input}

The preceding results considered channels with a fixed classical input
alphabet. We now allow the encoder to prepare arbitrary quantum states,
including states entangled across channel uses, while the transmitted
messages remain classical. We also strengthen the secrecy requirement:
each message must remain secret from the coalition of the unintended
receiver and the Stinespring environment. This models an adversary with
joint access to the receiver output and the channel environment; secrecy
from the two systems separately would not imply secrecy from their joint
state.

Let
\[
 W_{\mathcal N}^{A\longrightarrow B_1B_2E}
\]
be a Stinespring isometry for a finite-dimensional quantum broadcast channel
\(\mathcal N^{A\to B_1B_2}\). The complementary systems are \(B_2E\) for
receiver 1 and \(B_1E\) for receiver 2. This joint requirement is invariant
under the choice of Stinespring dilation.

\begin{figure}[t]
\centering
\begin{tikzpicture}[
  font=\small,
  message/.style={
    draw, rounded corners=1.5pt, minimum width=10mm,
    minimum height=8mm, inner sep=2pt, fill=black!2
  },
  block/.style={
    draw, rounded corners=2pt, align=center, minimum height=13mm,
    inner sep=4pt, fill=blue!4
  },
  system/.style={
    draw, rounded corners=1.5pt, minimum width=12mm,
    minimum height=8mm, inner sep=2pt, fill=black!2
  },
  decoder/.style={
    draw, rounded corners=2pt, align=center, minimum width=18mm,
    minimum height=10mm, inner sep=3pt, fill=blue!4
  },
  wire/.style={
    -{Latex[length=2.1mm,width=1.4mm]}, semithick
  },
  secrecy/.style={
    draw=black!65, dashed, rounded corners=2pt, align=center,
    text width=6.15cm, minimum height=15mm, inner sep=4pt,
    fill=black!1
  }
]
\node[message] (m1) {$M_1$};
\node[message, below=7mm of m1] (m2) {$M_2$};
\coordinate (mcenter) at ($(m1)!0.5!(m2)$);

\node[block, right=11mm of mcenter, text width=2.45cm] (encoder)
  {Stochastic block\\encoder};
\node[block, right=17mm of encoder, text width=3.25cm] (channel)
  {Broadcast Stinespring\\
   isometry $W_{\mathcal N}^{\otimes n}$\\
   $A^n\longrightarrow B_1^nB_2^nE^n$};

\node[system] (b1)
  at ([xshift=12mm,yshift=12mm]channel.east) {$B_1^n$};
\node[system] (b2)
  at ([xshift=12mm]channel.east) {$B_2^n$};
\node[system] (env)
  at ([xshift=12mm,yshift=-12mm]channel.east) {$E^n$};

\node[decoder, right=7mm of b1] (dec1) {$\mathcal D_1^{(n)}$};
\node[decoder, right=7mm of b2] (dec2) {$\mathcal D_2^{(n)}$};
\node[message, right=6mm of dec1] (hat1) {$\widehat M_1$};
\node[message, right=6mm of dec2] (hat2) {$\widehat M_2$};

\draw[wire] (m1.east) -- (encoder.west);
\draw[wire] (m2.east) -- (encoder.west);
\draw[wire] (encoder.east) --
  node[above,font=\footnotesize] {$\sigma_{m_1,m_2}^{A^n}$}
  (channel.west);
\draw[wire] (channel.east) -- (b1.west);
\draw[wire] (channel.east) -- (b2.west);
\draw[wire] (channel.east) -- (env.west);
\draw[wire] (b1.east) -- (dec1.west);
\draw[wire] (b2.east) -- (dec2.west);
\draw[wire] (dec1.east) -- (hat1.west);
\draw[wire] (dec2.east) -- (hat2.west);

\coordinate (figurecenter) at ($(m1.west)!0.5!(hat2.east)$);
\node[secrecy, anchor=north east] (sec1)
  at ([xshift=-3mm,yshift=-29mm]figurecenter)
  {\textbf{Secrecy of $M_1$}\\[-0.5mm]
   coalition $B_2^nE^n$ with side information $M_2$\\[-0.5mm]
   $I(M_1;B_2^nE^n\mid M_2)\longrightarrow0$};
\node[secrecy, anchor=north west] (sec2)
  at ([xshift=3mm,yshift=-29mm]figurecenter)
  {\textbf{Secrecy of $M_2$}\\[-0.5mm]
   coalition $B_1^nE^n$ with side information $M_1$\\[-0.5mm]
   $I(M_2;B_1^nE^n\mid M_1)\longrightarrow0$};
\end{tikzpicture}
\caption{Quantum-input confidential-message model. The two secrecy
conditions protect $M_1$ from the coalition $B_2^nE^n$ given $M_2$, and
$M_2$ from $B_1^nE^n$ given $M_1$. The displayed strong-secrecy conditions are attained by our
constructions; the operational capacity definition requires
only vanishing conditional trace leakage.}
\label{fq-fig-operational-model}
\end{figure}
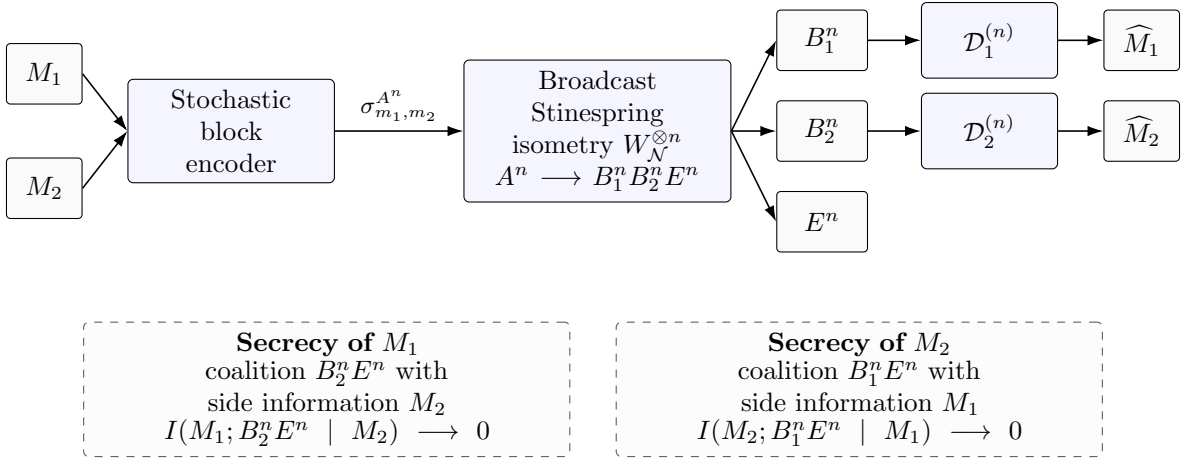

\subsection{Operational model}
\label{fq-sec-operational-model}

\begin{definition}[Quantum-input confidential-message code]
\label{fq-def-code}
An \(n\)-block code consists of two finite message sets
\(\cM_1,\cM_2\), a family of density operators
\[
 \bigl\{\sigma_{m_1,m_2}^{A^n}:
 (m_1,m_2)\in\cM_1\times\cM_2\bigr\},
\]
and decoding POVMs
\(\{\Lambda_{m_i}^{(i)}:m_i\in\cM_i\}\) on \(B_i^n\),
for \(i=1,2\). Private encoder randomness is already included by allowing
the states \(\sigma_{m_1,m_2}^{A^n}\) to be mixed. For independent uniform
messages, the induced state is
\begin{align}
 \Omega_n^{M_1M_2B_1^nB_2^nE^n}
 :=\frac{1}{|\cM_1||\cM_2|}
 \sum_{m_1,m_2}\ketbra{m_1,m_2}
 \otimes W_{\mathcal N}^{\otimes n}
 \sigma_{m_1,m_2}^{A^n}
 W_{\mathcal N}^{\dagger\otimes n}.
 \label{fq-eq-code-state}
\end{align}
Writing \(\Omega_{m_1,m_2}^{B_i^n}\) for the corresponding conditional
marginal, the average decoding errors are
\begin{equation}
 P_{e,i}^{(n)}
 :=1-\frac{1}{|\cM_1||\cM_2|}
 \sum_{m_1,m_2}
 \Tr\!\left[\Lambda_{m_i}^{(i)}
 \Omega_{m_1,m_2}^{B_i^n}\right],
 \qquad i=1,2.
 \label{fq-eq-average-error}
\end{equation}
The conditional trace-secrecy parameters are
\begin{align}
 \Delta_{1,n}
 &:=\frac12\left\|
 \Omega_n^{M_1M_2B_2^nE^n}
 -\pi^{M_1}\otimes\Omega_n^{M_2B_2^nE^n}
 \right\|_1,
 \label{fq-eq-trace-secrecy-one}\\
 \Delta_{2,n}
 &:=\frac12\left\|
 \Omega_n^{M_1M_2B_1^nE^n}
 -\pi^{M_2}\otimes\Omega_n^{M_1B_1^nE^n}
 \right\|_1.
 \label{fq-eq-trace-secrecy-two}
\end{align}
\end{definition}

A rate pair is achievable if there is a sequence of such codes for which
\(n^{-1}\log|\cM_i|\) tends to at least \(R_i\), both average errors vanish,
and \(\Delta_{1,n},\Delta_{2,n}\to0\). The closure of the achievable rate
pairs is denoted by \(\mathcal C_{\mathrm{conf}}(\mathcal N)\). The codes
constructed below satisfy the stronger conditions
\begin{equation}
 I(M_1;B_2^nE^n|M_2)_{\Omega_n}\longrightarrow0,
 \qquad
 I(M_2;B_1^nE^n|M_1)_{\Omega_n}\longrightarrow0.
 \label{fq-eq-mi-secrecy}
\end{equation}
For the converse, the trace criterion is sufficient. Conditional-entropy
continuity gives
\begin{align}
 I(M_1;B_2^nE^n|M_2)_{\Omega_n}
 \leq{}&2\Delta_{1,n}\log|\cM_1|
 +(1+\Delta_{1,n})
 h_2\!\left(\frac{\Delta_{1,n}}{1+\Delta_{1,n}}\right),
 \label{fq-eq-trace-to-mi}
\end{align}
and similarly after interchanging the receivers. Along a reliable code
sequence, Fano's inequality and
\(I(M_i;B_i^n)\le n\log\dim B_i\) also give
\begin{equation}
 (1-P_{e,i}^{(n)})\log|\cM_i|
 \le n\log\dim B_i+h_2(P_{e,i}^{(n)}),
 \label{fq-eq-rate-bound}
\end{equation}
so \(\log|\cM_i|=O(n)\). Consequently, the leakage in
\eqref{fq-eq-trace-to-mi} is \(o(n)\) whenever
\(\Delta_{i,n}\to0\).

The two colluding-adversary conditions also imply joint privacy from the
environment alone. Indeed, contractivity under partial trace and the
triangle inequality give
\begin{align}
 \frac12\left\|
 \Omega_n^{M_1M_2E^n}
 -\pi^{M_1}\otimes\pi^{M_2}\otimes\Omega_n^{E^n}
 \right\|_1
 \leq \Delta_{1,n}+\Delta_{2,n}.
 \label{fq-eq-joint-environment-secrecy}
\end{align}

\subsection{The cq theorem with enlarged secrecy outputs}
\label{fq-sec-enlarged-cq}

A prescribed finite family of quantum input preparations induces a cq channel. To apply the coding theorem with receiver--environment secrecy, we first enlarge the output registers in the resolvability argument.

\begin{corollary}[Enlarged secrecy outputs]
\label{fq-cor-enlarged-cq}
Let \( x\longmapsto \rho_x^{B_1B_2E} \) be a memoryless cq channel. Fix finite \(U,V_1,V_2,X\) and the law \eqref{eq:factorization}
\[
 p(u,v_1,v_2,x)
 =p(u)p(v_1,v_2|u)p(x|u,v_1,v_2),
\]
and evaluate all information quantities on
\begin{align}
 \omega^{UV_1V_2XB_1B_2E}
 :=\sum_{u,v_1,v_2,x} p(u)p(v_1,v_2|u)p(x|u,v_1,v_2)
 \, \ketbra{u,v_1,v_2,x} \otimes\rho_x^{B_1B_2E}.
 \label{fq-eq-tripartite-cq-state}
\end{align}
Every nonnegative pair satisfying
\begin{align}
 R_1
 &<I(V_1;B_1|U)_\omega-I(V_1;V_2|U)_\omega
   -I(V_1;B_2E|V_2,U)_\omega,
 \label{fq-eq-cq-rate-one}\\
 R_2
 &<I(V_2;B_2|U)_\omega-I(V_1;V_2|U)_\omega
   -I(V_2;B_1E|V_1,U)_\omega
 \label{fq-eq-cq-rate-two}
\end{align}
is achievable with vanishing average error and conditional strong secrecy
in the sense of \eqref{fq-eq-mi-secrecy}. As in
\cref{prop:inactive-user}, if one rate is zero, only the other displayed inequality is required. The closure of the resulting union, over all finite auxiliaries and distributions \eqref{eq:factorization}, is an inner bound on the confidential classical capacity region with an environment. The same inner region is obtained by restricting the union to
\[
 |\mathcal U|\le 3,\qquad
 |\mathcal V_1|\le|\mathcal X|,\qquad
 |\mathcal V_2|\le|\mathcal X|.
\]
\end{corollary}

\begin{proof}
Use the normalized likelihood selector and fine-index decoders from
the cq theorem. The planted-law packing analysis depends only on the
marginals on \(B_1\) and \(B_2\), and is therefore unchanged. In the first
oriented resolvability argument, replace the soft-covering output \(B_2\) by
\(B_2E\); in the second, replace \(B_1\) by \(B_1E\). Ordinary cq soft
covering applies to these enlarged finite-dimensional outputs, while the
selector-balancing estimates depend only on the classical law of
\((U,V_1,V_2)\). The resulting auxiliary-rate conditions are
\begin{align*}
 \widetilde R_1&>I(V_1;V_2|U),
 &R'_1+\widetilde R_1&>I(V_1;V_2B_2E|U),\\
 \widetilde R_2&>I(V_1;V_2|U),
 &R'_2+\widetilde R_2&>I(V_2;V_1B_1E|U),
\end{align*}
together with the unchanged packing inequalities
\[
 R_i+R'_i+\widetilde R_i<I(V_i;B_i|U),\qquad i=1,2.
\]
The same auxiliary-rate choice and elimination used in the cq theorem gives
\eqref{fq-eq-cq-rate-one}--\eqref{fq-eq-cq-rate-two}. Finally, select one
deterministic codebook for which both decoding errors and both trace-secrecy
parameters vanish. The exponential decay of the trace-secrecy parameters and
\eqref{fq-eq-trace-to-mi} give \eqref{fq-eq-mi-secrecy}.
For an inactive receiver, apply \cref{prop:inactive-user} with unintended
output \(B_2E\) or \(B_1E\), respectively. The cardinality bounds follow from Proposition~\ref{prop:auxiliary-cardinalities}.
\end{proof}


\subsection{A multiletter capacity formula}
\label{fq-sec-regularized-capacity}

Applying the cq theorem to block quantum preparations gives an operational
multiletter characterization. Each preparation is treated as a letter of a
superchannel, and the converse uses the independent messages themselves as
block auxiliaries. The preparations may be mixed and entangled across the
physical channel uses. A stochastic prefix over a finite ensemble can be
incorporated into the density operator
\[
 \sum_x p(x|u,v_1,v_2)\sigma_x^A
 =:\sigma_{u,v_1,v_2}^A,
\]
which we use below.

For each integer \(q\geq1\), let \(\mathfrak S_q^{\mathrm{ind}}(\mathcal N)\)
be the collection of states
\begin{align}
 \omega^{UV_1V_2B_1^qB_2^qE^q}
 :=\sum_{u,v_1,v_2} p(u)p(v_1|u)p(v_2|u)\,\ketbra{u,v_1,v_2} \otimes W_{\mathcal N}^{\otimes q}
 \sigma_{u,v_1,v_2}^{A^q}
 W_{\mathcal N}^{\dagger\otimes q},
 \label{fq-eq-independent-block-state}
\end{align}
where \(U,V_1,V_2\) are arbitrary finite classical variables and the
\(\sigma_{u,v_1,v_2}^{A^q}\) are arbitrary density operators. In
particular, the latter may be mixed and entangled across the \(q\) channel
inputs. For \(\omega\in\mathfrak S_q^{\mathrm{ind}}(\mathcal N)\), put
\begin{align}
 J_1^{(q)}(\omega)
 &:=I(V_1;B_1^q|U)_\omega
   -I(V_1;B_2^qE^q|V_2,U)_\omega,
 \label{fq-eq-block-private-one}\\
 J_2^{(q)}(\omega)
 &:=I(V_2;B_2^q|U)_\omega
   -I(V_2;B_1^qE^q|V_1,U)_\omega.
 \label{fq-eq-block-private-two}
\end{align}
Define the downward-closed region
\begin{align}
 \mathcal R_q^{\mathrm{ind}}(\mathcal N)
 :=\bigcup_{\omega\in\mathfrak S_q^{\mathrm{ind}}(\mathcal N)}
 \left\{(r_1,r_2)\in\mathbb R_+^2:
 r_i\leq J_i^{(q)}(\omega),\ i=1,2\right\}.
 \label{fq-eq-block-region}
\end{align}
The union is convex because \(U\) time-shares among the component regions.

\begin{theorem}[Regularized capacity region]
\label{fq-thm-regularized-capacity}
For any quantum broadcast channel $\mathcal{N}$, in the colluding receiver--environment secrecy model of
\cref{fq-def-code}, we have
\begin{align}
 \mathcal C_{\mathrm{conf}}(\mathcal N)
 =\overline{\bigcup_{q\geq1}
 \frac1q\mathcal R_q^{\mathrm{ind}}(\mathcal N)}.
 \label{fq-eq-exact-capacity}
\end{align}
All auxiliary systems in this expression are classical.
\end{theorem}

\begin{proof}
We first prove achievability. Fix \(q\), a state \(\omega\) of the form
\eqref{fq-eq-independent-block-state}, and back off arbitrarily from each
positive endpoint of its rectangle; a zero endpoint is covered by the
inactive-user clause of \cref{fq-cor-enlarged-cq}. Regard \(q\) physical
channel uses as one superchannel use and,
conditional on \((u,v_1,v_2)\), prepare
\(\sigma_{u,v_1,v_2}^{A^q}\). Equivalently, take the classical
superchannel letter to be \(X=(U,V_1,V_2)\) and use the deterministic prefix
that prepares this state. This gives a finite-input cq superchannel with
output states
\[
 \rho_{u,v_1,v_2}^{B_1^qB_2^qE^q}
 :=W_{\mathcal N}^{\otimes q}
 \sigma_{u,v_1,v_2}^{A^q}
 W_{\mathcal N}^{\dagger\otimes q}.
\]
Apply \cref{fq-cor-enlarged-cq}. Because
\(I(V_1;V_2|U)_\omega=0\), its two rate bounds reduce to
\eqref{fq-eq-block-private-one}--\eqref{fq-eq-block-private-two}. A code
using \(N\) superchannel uses employs \(Nq\) physical channel uses, so its
rates per physical use are divided by \(q\). Taking the union and closure
proves the direct inclusion in \eqref{fq-eq-exact-capacity}.

For the converse, consider an arbitrary reliable and secure \(n\)-block
code and its state \(\Omega_n\) from \eqref{fq-eq-code-state}. Fano's
inequality and data processing through receiver \(1\)'s measurement give
\begin{equation}
 \log|\cM_1|
 \leq I(M_1;B_1^n)_{\Omega_n}+o(n).
 \label{fq-eq-fano-one}
\end{equation}
By \eqref{fq-eq-trace-to-mi},
\(I(M_1;B_2^nE^n|M_2)_{\Omega_n}=o(n)\). Since the messages are
independent,
\begin{align}
 I(M_1;M_2B_2^nE^n)_{\Omega_n}
 &=I(M_1;B_2^nE^n|M_2)_{\Omega_n}.
 \label{fq-eq-independence-one}
\end{align}
Subtracting this \(o(n)\) quantity from \eqref{fq-eq-fano-one} yields
\begin{align}
 \log|\cM_1|
 \leq I(M_1;B_1^n)_{\Omega_n}
 -I(M_1;B_2^nE^n|M_2)_{\Omega_n}+o(n).
 \label{fq-eq-converse-one}
\end{align}
The symmetric argument gives
\begin{align}
 \log|\cM_2|
 \leq I(M_2;B_2^n)_{\Omega_n}
 -I(M_2;B_1^nE^n|M_1)_{\Omega_n}+o(n).
 \label{fq-eq-converse-two}
\end{align}
Now take \(q=n\), let \(U\) be constant, set \(V_1=M_1\) and
\(V_2=M_2\), and choose
\(\sigma_{v_1,v_2}^{A^n}=\sigma_{m_1,m_2}^{A^n}\). The messages are
independent, so the resulting state belongs to
\(\mathfrak S_n^{\mathrm{ind}}(\mathcal N)\). If both target rates are
positive, \eqref{fq-eq-converse-one}--\eqref{fq-eq-converse-two} give
positive rectangle endpoints for sufficiently large $n$. Reducing the
target rates by $o(1)$ then gives a point in
$n^{-1}\mathcal R_n^{\mathrm{ind}}(\mathcal N)$.

At an axis point, say $R_2=0<R_1$, replace $V_2$ by a constant and use the
averaged preparations
\[
 \bar\sigma_{v_1}^{A^n}
 :=\sum_{v_2}p(v_2)\sigma_{v_1,v_2}^{A^n}.
\]
This makes $J_2^{(n)}=0$ and leaves the $V_1B_1^n$ marginal unchanged.
Moreover, independence and data processing give
\[
 I(V_1;B_2^nE^n|V_2)
 =I(V_1;V_2B_2^nE^n)
 \ge I(V_1;B_2^nE^n),
\]
so the replacement cannot decrease $J_1^{(n)}$. The positive rate can
again be reduced by $o(1)$ to obtain a point in the regularized region.
The other axis is symmetric, and constant auxiliaries give the origin.
Taking limits and the closure proves the reverse inclusion.
\end{proof}

\begin{remark}[Relation to the correlated Marton region]
\label{fq-rem-correlated-region}
For a fixed \(q\), permitting a general law
\(p(u)p(v_1,v_2|u)\) and using the two bounds
\begin{align*}
 r_1&\leq I(V_1;B_1^q|U)-I(V_1;V_2B_2^qE^q|U),\\
 r_2&\leq I(V_2;B_2^q|U)-I(V_2;V_1B_1^qE^q|U)
\end{align*}
gives the natural correlated Marton inner region. Its regularized union is
also equal to \eqref{fq-eq-exact-capacity}: it contains the independent
region, while every code it produces is subject to the converse just proved.
Thus the correlated expression is a one-letter inner bound, while the
conditionally independent expression yields the same regularized union.
\end{remark}

\subsection{Deterministic Quantum Broadcast Channels}
\label{sec:deterministic_quantum-broadcast}

In this section, we consider deterministic quantum broadcast channels, by which we mean isometric broadcast channels $\U: \H^A \rightarrow \H^{B} \otimes \H^{C}$ of the form $\U\ket{x}^A = \ket{y(x)}^{B}\ket{z(x)}^{C}$ for some functions $y$ and $z$. These channels provide coherent isometric realizations of deterministic classical broadcast channels, for which the capacity regions simplify (see \cite[Example 8.2]{el2011network}).

If $x\mapsto(y(x),z(x))$ is injective, its linear extension $\U\ket{x}^A = \ket{y(x)}^{B}\ket{z(x)}^{C}$ 
is an isometry, defining a quantum channel whose inputs may be superpositions or mixed states. Basis-state preparations reproduce the
classical channel. While a general classical deterministic broadcast channel need not have that property, we only consider injective maps $x\mapsto(y(x),z(x))$ here. We show that, for such channels, the capacity region admits a single-letter characterization, giving a rate region that is the quantum analogue of the rate region for deterministic classical channels obtained in \cite{CaiLam2000}.

Since the additional
Stinespring environment $E$ is trivial for an isometric broadcast channel, the secrecy
conditions of \cref{fq-sec-operational-model} protect each message from
the unintended receiver given that receiver's own message.

\begin{theorem}
\label{thm:capacity-deterministic_quantum}
Let
\(
  \, \mathcal U:\mathcal H^A\longrightarrow
  \mathcal H^B\otimes\mathcal H^C
\)
be an isometry satisfying
\(
  \,\mathcal U\ket{x}^A
  =\ket{y(x)}^B\ket{z(x)}^C,
  \; x\in \mathcal X,
\)
for fixed orthonormal bases of \(\H^A\), \(\H^B\), and \(\H^C\).
For \(p_X\in\mathcal P(\mathcal X)\), define
\[
  \omega_p^{XBC}
  :=\sum_x p_X(x)[x]^X
       \otimes[y(x)]^B\otimes[z(x)]^C.
\]
Then the confidential-message capacity region for the deterministic quantum broadcast channel $\U$ is
\[ \mathcal C_{\mathrm{conf}}(\mathcal U) =
  \operatorname{conv}\bigcup_{p_X \in \mathcal P(\mathcal X)} \left\{
    (R_1,R_2)\in\mathbb R_+^2:
    \begin{array}{l}
      R_1\le H(X|C)_{\omega_p},\\
      R_2\le H(X|B)_{\omega_p}
    \end{array}
  \right\}.
\]
\end{theorem}
The proof is divided into achievability and converse parts.

\subsubsection{Achievability}
The achievability proof below uses only encoders that are classical in the distinguished input basis. We subsequently prove a converse for arbitrary encoder states, including coherent inputs and states entangled across channel uses.
\begin{proposition}[Achievability]
\label{prop:direct_coding_deterministic}
For every \(p_X\in\mathcal P(\mathcal X)\), every rate pair satisfying
\[
0\le R_1\le H(X|C)_{\omega_p},
\qquad
0\le R_2\le H(X|B)_{\omega_p}
\]
is achievable. The conditional entropies above are evaluated with respect to the state of the following form:
    \begin{align*}
        \omega_p^{XBC} = \sum_{x} p_X(x) [x]^X \otimes [y(x)]^{B} \otimes [z(x)]^{C}
    \end{align*} Consequently, the inner region
\[
\mathcal R_{\mathrm{det}}(\mathcal U)
:=
\operatorname{conv}
\bigcup_{p_X\in\mathcal P(\mathcal X)}
\left\{
(R_1,R_2)\in\mathbb R_+^2:
R_1\le H(X|C)_{\omega_p},\
R_2\le H(X|B)_{\omega_p}
\right\}
\]
satisfies
\[
\mathcal R_{\mathrm{det}}(\mathcal U)
\subseteq\mathcal C_{\mathrm{conf}}(\mathcal U).
\]
\end{proposition}

\begin{proof}
Fix \(p_X \in \mathcal P(\mathcal X)\), and let \( Y=y(X),\; Z=z(X). \)
In Theorem~\ref{thm:main}, take the time-sharing auxiliary to be constant and choose \( V_1=Y,\; V_2=Z. \)
The channel input is the basis state \(\ket{x}\) associated with
\((V_1,V_2)=(y(x),z(x))\). Equivalently, we choose
\[
p(v_1,v_2,x)
=
p_X(x)\,
\mathbf 1\{v_1=y(x),\,v_2=z(x)\}.
\]
Injectivity of \(x\mapsto(y(x),z(x))\) ensures that \(x\) is a
well-defined function of every pair \((v_1,v_2)\) in the support. Since the channel outputs are orthogonal
classical states, the first rate bound reduces to
\[
\begin{aligned}
I(V_1;B)-I(V_1;V_2)-I(V_1;C|V_2)
 &=H(Y)-I(Y;Z)\\
 &=H(Y|Z).
\end{aligned}
\]
Similarly, the second bound reduces to \(H(Z|Y)\).

Because \(\mathcal U\) is an isometry, \(x\mapsto(y(x),z(x))\) is
injective, and hence \(H(X|Y,Z)=0\). Consequently,
\[
H(Y|Z)=H(X|Z)=H(X|C)_{\omega_p},
\]
and, similarly,
\[
H(Z|Y)=H(X|Y)=H(X|B)_{\omega_p}.
\]
Thus every pair satisfying the claimed bounds is achievable.
Time sharing gives their convex hull in Proposition~\ref{prop:direct_coding_deterministic}.
\end{proof}

\subsubsection{The Single-Letter Converse}
Although the preceding construction uses only classical input states, the following converse applies to arbitrary quantum encoders and shows that input coherence cannot enlarge the confidential-message capacity region.

\begin{proposition}[Converse]
\label{prop:converse-deterministic}
For every deterministic quantum broadcast channel \(\mathcal U\),
\[
\mathcal C_{\mathrm{conf}}(\mathcal U) \subseteq \mathcal R_{\mathrm{det}}(\mathcal U).
\]
\end{proposition}

\begin{proof}
Consider a code of the form 
\begin{equation}
 \Omega^{M_1M_2B^nC^n}
 =
 \frac{1}{|\mathcal M_1||\mathcal M_2|}
 \sum_{m_1,m_2}\proj{m_1,m_2}\otimes
 \U^{\otimes n}
 \theta_{m_1,m_2}^{A^n}
 (\U^\dagger)^{\otimes n}.
 \label{det-c-code-state}
\end{equation}
where \(\theta^{A^n}_{m_1,m_2} \in \D((\H^A)^{\otimes n}\)) are arbitrary. Let $\epsilon_{i,n}$ be receiver $i$'s average decoding error, and put
\[
 \delta_{1,n}=I(M_1;C^n|M_2)_\Omega,\qquad
 \delta_{2,n}=I(M_2;B^n|M_1)_\Omega.
\]
Measure $C^n$ in the computational basis and denote its classical
outcome by $Z^n$. This operation leaves the marginal
$M_1M_2B^n$ unchanged. Data processing for the secrecy term gives
\begin{equation}
 I(M_1;Z^n|M_2)\le\delta_{1,n}.
 \label{det-c-measured-leakage}
\end{equation}
Giving $Z^n$ to receiver $1$ can only help decoding. Fano's inequality,
message independence, and the chain rule therefore imply
\begin{align}
 (1-\epsilon_{1,n})\log|\mathcal M_1|-h_2(\epsilon_{1,n})
 &\le I(M_1;B^n|M_2)\notag\\
 &\le I(M_1;B^nZ^n|M_2)\notag\\
 &=I(M_1;Z^n|M_2)
   +I(M_1;B^n|Z^nM_2)\notag\\
 &\le\delta_{1,n}+H(B^n|Z^nM_2)\notag\\
 &\le\delta_{1,n}+H(B^n|Z^n)\notag\\
 &\le\delta_{1,n}+H(Y^n|Z^n).
 \label{det-c-converse-one}
\end{align}
All conditioning registers in the fourth line are classical. Hence
$H(B^n|M_1M_2Z^n)\ge0$, which gives the Holevo-information bound used
there. The next line follows from strong subadditivity. In the last
line, $Y^n$ denotes computational-basis measurement of $B^n$:
dephasing increases the entropy of each state conditioned on $Z^n$.

Interchanging the two receivers gives
\begin{equation}
 (1-\epsilon_{2,n})\log|\mathcal M_2| -h_2(\epsilon_{2,n})
 \le\delta_{2,n}+H(Z^n|Y^n).
 \label{det-c-converse-two}
\end{equation}
The classical entropies in both bounds refer to the same joint law
\begin{equation}
 q(y^n,z^n)
 =
 \langle y^n,z^n|
 \Omega^{B^nC^n}
 |y^n,z^n\rangle,
 \label{det-c-diagonal-law}
\end{equation} 
and the input distribution is given by \[p(x^n) = \bra{x^n} \overline{\theta}^{A^n} \ket{x^n}, \qquad \overline{\theta}^{A^n} =  \frac{1}{|\mathcal{M}_1| |\mathcal{M}_2|}\sum_{m_1,m_2}  \theta^{A^n}_{m_1,m_2} .\]
Since
\[
\overline\theta^{A^n}
=
\sum_{x^n,\widetilde x^n}
\alpha_{x^n,\widetilde x^n}
\ket{x^n}\!\bra{\widetilde x^n},
\]
the diagonal output law satisfies
\[
q(y^n,z^n)
=
\sum_{\substack{x^n:\,
y^n(x^n)=y^n\\
z^n(x^n)=z^n}}
p(x^n).
\]
Indeed, injectivity of \(x\mapsto(y(x),z(x))\) implies that every valid pair \((y^n,z^n)\) has at most one preimage, so no off-diagonal coefficient \(\alpha_{x^n,\widetilde x^n}\), \(x^n\ne\widetilde x^n\), contributes to the jointly measured output law.

For each $i$, define \(p_i(x) = \bra{x} \overline{\theta}^{A_i} \ket{x}\), where \(\overline{\theta}^{A_i} = \tc_{A^{[n]\backslash i}} \overline{\theta}^{A^n}  \), and further define the probability distribution \[ \overline{p}_n(x) = \displaystyle \frac{1}{n} \sum_{i=1}^n p_i(x). \] Both bounds use the same family $\{ p_i(x) \}$ of input probability distributions.
Let $q_i (b,c) = \displaystyle \sum_{\substack{x: \,y(x) =b, \\ z(x) =c}} p_i (x)$ be its marginal at coordinate $i$, and let
$\bar q=n^{-1}\sum_{i=1}^nq_i$. 

Again, since \( \U\) is an isometry, \( x \mapsto (y(x),z(x))\) is injective. This implies that \( H(X|Y,Z) = 0\). Then, the classical chain rule and concavity of conditional entropy give
\begin{align}
 \frac1nH(Y^n|Z^n) & \le \frac1n\sum_{t=1}^n H(Y|Z)_{q_t}
 \le H(Y|Z)_{\overline{q}} = H(X|C)_{\omega_{\overline{p}_n}} \notag \\
 \frac1nH(Z^n|Y^n)
 &\le H(Z|Y)_{\overline{q}} = H(X|B)_{\omega_{\overline{p}_n}}.
 \label{det-c-single-letter}
\end{align}
Put
\[
R_{j,n}:=\frac1n\log|\mathcal M_j|.
\]
Dividing \eqref{det-c-converse-one} and \eqref{det-c-converse-two} by \(n\) gives
\begin{align*}
(1-\epsilon_{1,n})R_{1,n}
-\frac1n h_2(\epsilon_{1,n})
&\le
\frac{\delta_{1,n}}n
+H(X|C)_{\omega_{\bar p_n}},\\
(1-\epsilon_{2,n})R_{2,n}
-\frac1n h_2(\epsilon_{2,n})
&\le
\frac{\delta_{2,n}}n
+H(X|B)_{\omega_{\bar p_n}}.
\end{align*}
Equations~\eqref{fq-eq-trace-to-mi} and
\eqref{fq-eq-rate-bound} imply
\[
 \frac{\delta_{i,n}}{n}\longrightarrow0,\qquad i=1,2,
\]
since the trace-secrecy parameters vanish and
\(\log|\mathcal M_i|=O(n)\).
By compactness of \(\mathcal P(\mathcal X)\), there exist a
subsequence \(n_k\) and a distribution \(p_\ast\) such that
\(\bar p_{n_k}\to p_\ast\).
Taking limits along this same subsequence, using
\(\epsilon_{i,n}\to0\),
\(\liminf_{n\to\infty}R_{i,n}\ge R_i\), and continuity of
conditional entropy, gives
\[
 R_1\le H(X|C)_{\omega_{p_\ast}},
 \qquad
 R_2\le H(X|B)_{\omega_{p_\ast}}.
\]
Therefore, \[ \mathcal C_{\mathrm{conf}}(\mathcal U) \subseteq\mathcal R_{\mathrm{det}}(\mathcal U).
\]

\end{proof}
\vspace{-5pt}

\begin{proof}[Proof of Theorem~\ref{thm:capacity-deterministic_quantum}]
Proposition~\ref{prop:direct_coding_deterministic} gives the direct inclusion, while Proposition~\ref{prop:converse-deterministic} gives the reverse inclusion.
\end{proof}

\subsection{A degraded Stinespring broadcast channel}
\label{fq-sec-degraded}

We return to the degraded setting of \cref{sec:degr-cq}, now with arbitrary
quantum preparations and the enlarged secrecy outputs. The relevant
condition is degradability of receiver $1$'s marginal: its degrading map
must reproduce the entire complementary output $B_2E$. Let
\begin{align*}
 \mathcal N_1(\rho)
 &:=\Tr_{B_2E}\!\left[W_{\mathcal N}\rho
 W_{\mathcal N}^{\dagger}\right],\\
 \mathcal N_1^c(\rho)
 &:=\Tr_{B_1}\!\left[W_{\mathcal N}\rho
 W_{\mathcal N}^{\dagger}\right],
\end{align*}
so that \(\mathcal N_1^c\) has output \(B_2E\). Suppose that receiver
\(1\)'s marginal is degradable, i.e., there is a channel
\(\mathcal D^{B_1\to B_2E}\) such that
\begin{equation}
 \mathcal N_1^c=\mathcal D\circ\mathcal N_1.
 \label{fq-eq-degrading-map}
\end{equation}

\begin{corollary}[Degraded case]
\label{fq-cor-degraded}
Under \eqref{fq-eq-degrading-map},
\begin{align}
 \mathcal C_{\mathrm{conf}}(\mathcal N)
 =\left\{(R_1,0):
 0\leq R_1\leq Q(\mathcal N_1)\right\},
 \label{fq-eq-degraded-region}
\end{align}
where
\begin{align}
 Q(\mathcal N_1)
 =\max_{\rho^A}
 \left\{H(B_1)_{W_{\mathcal N}\rho W_{\mathcal N}^{\dagger}}
 -H(B_2E)_{W_{\mathcal N}\rho W_{\mathcal N}^{\dagger}}\right\}
 \label{fq-eq-degraded-coherent-information}
\end{align}
is the quantum capacity of the degradable channel \(\mathcal N_1\),
expressed by its maximum single-letter coherent information.
The rates $R_1,R_2$ here are rates of confidential classical messages.
\end{corollary}

\begin{proof}
First, tracing out \(E\) after the degrading map in
\eqref{fq-eq-degrading-map} gives a channel from \(B_1\) to \(B_2\).
Consequently, data processing implies
\[
 I(V_2;B_2^q|U)\leq I(V_2;B_1^q|U).
\]
Moreover, conditional independence of \(V_1,V_2\) gives
\begin{align*}
 I(V_2;B_1^qE^q|V_1,U)
 &=I(V_2;V_1B_1^qE^q|U)\\
 &\geq I(V_2;B_1^q|U).
\end{align*}
It follows that \(J_2^{(q)}(\omega)\leq0\) for every \(q\) and every
admissible state. Hence \(R_2=0\).

Put \(C:=B_2E\). For the first coordinate, conditional independence
similarly gives
\begin{align*}
 J_1^{(q)}(\omega)
 &=I(V_1;B_1^q|U)_\omega-I(V_1;V_2C^q|U)_\omega\\
 &\leq I(V_1;B_1^q|U)_\omega-I(V_1;C^q|U)_\omega\\
 &\leq \max_{\rho^{A^q}} I_c(\rho,\mathcal N_1^{\otimes q}).
\end{align*}
To see the second inequality explicitly, for fixed \(u\) define
\[
 \bar\sigma_{u,v_1}
 :=\sum_{v_2}p(v_2|u)\sigma_{u,v_1,v_2},
 \qquad
 \bar\sigma_u:=\sum_{v_1}p(v_1|u)\bar\sigma_{u,v_1}.
\]
Then
\begin{align*}
 I(V_1;B_1^q)_{\omega_u}-I(V_1;C^q)_{\omega_u}
 =I_c(\bar\sigma_u,\mathcal N_1^{\otimes q})
 -\sum_{v_1}p(v_1|u)
 I_c(\bar\sigma_{u,v_1},\mathcal N_1^{\otimes q}).
\end{align*}
Coherent information is nonnegative on every input of a degradable channel: it is concave and vanishes on pure inputs. Thus, the last sum is nonnegative, proving the bound after averaging over \(u\). The maximum coherent information is additive for degradable channels \cite{DevetakShor2005}, and therefore, the last expression equals \(qQ(\mathcal N_1)\).

For achievability, let \(\rho^A\) maximize
\eqref{fq-eq-degraded-coherent-information} and choose a pure-state ensemble
\(\rho=\sum_{v_1}p(v_1)\psi_{v_1}\). Take \(U,V_2\) constant and prepare
\(\psi_{v_1}\). Since
\(W_{\mathcal N}\psi_{v_1}W_{\mathcal N}^{\dagger}\) is pure on
\(B_1C\), its two marginals have equal entropy. Hence
\begin{align*}
 I(V_1;B_1)-I(V_1;C)
 =H(B_1)_\rho-H(C)_\rho
 =I_c(\rho,\mathcal N_1),
\end{align*}
and \cref{fq-cor-enlarged-cq} achieves every strict rate below
\(Q(\mathcal N_1)\). Taking the closure proves
\eqref{fq-eq-degraded-region}.
\end{proof}

If the broadcast channel itself is isometric,
\(W_{\mathcal N}^{A\to B_1B_2}\), then \(E\) is one-dimensional and
\eqref{fq-eq-degrading-map} reduces to the usual condition that \(B_2\) can
be simulated from \(B_1\). The symmetric statement holds when
\(B_1E\) can instead be simulated from \(B_2\): then \(R_1=0\) and
\(R_2\) ranges from zero to the quantum capacity of receiver \(2\)'s
degradable marginal.

The next section applies the quantum-input model to the Blackwell and
Platypus isometries and compares confidential classical rates with rates
for quantum communication.

\section{Confidential classical versus quantum communication}
\label{sec:classical-quantum-comparison}
\label{sec:degr-platypus}

We now compare confidential classical communication with quantum
communication over the same isometric broadcast channels. We first determine both capacity regions for the Blackwell isometry and show that the confidential classical capacity region strictly exceeds the
quantum capacity region for this channel, see also Figure~\ref{fig:blackwell-cap-regions}. We further show a separation between the achievable rate regions  for the
Platypus isometry.

\subsection{The coherent Blackwell channel}
\label{bw-q-sec}

The coherent Blackwell channel is defined by the isometry
\begin{equation}
 U_{\mathrm{BW}}:\quad
 \lvert0\rangle\longmapsto\lvert00\rangle,\qquad
 \lvert1\rangle\longmapsto\lvert01\rangle,\qquad
 \lvert2\rangle\longmapsto\lvert10\rangle.
 \label{bw-q-isometry}
\end{equation}
Here the input is a qutrit and the two output qubits belong to receivers
$1$ and $2$. An input state $\theta$ is mapped to
$U_{\mathrm{BW}}\theta U_{\mathrm{BW}}^\dagger$. The isometry acts
linearly on superpositions. For example,
\begin{equation}
 U_{\mathrm{BW}}
 \frac{\lvert0\rangle+\lvert2\rangle}{\sqrt2}
 =
 \frac{\lvert0\rangle+\lvert1\rangle}{\sqrt2}^{B_1}
 \lvert0\rangle^{B_2}.
 \label{det-eq-blackwell-coherence}
\end{equation}
The channel accepts arbitrary qutrit states, while the classical Blackwell channel $\mathcal W_{\mathrm{BW}}$ from \cref{det-eq-blackwell-map}
accepts only classical symbols. For every blocklength $n$, its output is
supported on $\mathcal S^{\otimes n}$, where
\begin{equation}
 \mathcal S=\operatorname{span}
 \{\lvert00\rangle,\lvert01\rangle,\lvert10\rangle\}.
 \label{bw-output-support}
\end{equation}
We first determine its capacity for confidential classical messages and
then its capacity for quantum communication.

\subsubsection{Confidential classical communication}
\label{bw-c-sec}

The encoder may prepare an arbitrary state
$\theta_{m_1,m_2}^{A^n}$ for each pair of independent uniform messages.
These states may be mixed and entangled across channel uses.
The induced state is
\begin{equation}
 \Omega^{M_1M_2B_1^nB_2^n}
 =
 \frac{1}{|\mathcal M_1||\mathcal M_2|}
 \sum_{m_1,m_2}\ketbra{m_1,m_2}\otimes
 U_{\mathrm{BW}}^{\otimes n}
 \theta_{m_1,m_2}^{A^n}
 (U_{\mathrm{BW}}^\dagger)^{\otimes n}.
 \label{bw-c-code-state}
\end{equation}
Each receiver uses a local POVM to recover its message. We require
vanishing average decoding errors and
\[
 I(M_1;B_2^n|M_2)_\Omega\longrightarrow0,\qquad
 I(M_2;B_1^n|M_1)_\Omega\longrightarrow0.
\]
These are the conditional strong-secrecy requirements of the quantum-input
model in \cref{fq-sec-operational-model}, with a trivial environment (which suffices because the channel is isometric).
The converse below also applies to that model's trace-secrecy criterion:
\eqref{fq-eq-trace-to-mi} and \eqref{fq-eq-rate-bound} imply that both
mutual-information leakages are $o(n)$, which is sufficient for the
argument.

\begin{corollary}
[Confidential classical capacity of the Blackwell isometry]
\label{bw-c-capacity}
The confidential classical capacity region of the coherent Blackwell
channel is
\begin{equation}
 \mathcal C_{\mathrm{conf}}(U_{\mathrm{BW}})
 =
 \mathcal C_{\mathrm{conf}}(\mathcal W_{\mathrm{BW}}),
 \label{bw-c-region}
\end{equation}
where the right-hand side is the classical region in
\eqref{det-eq-blackwell-region}.
\end{corollary}

\begin{proof}
The corollary follows immediately from Theorem~\ref{thm:capacity-deterministic_quantum}.
\end{proof}


\subsubsection{Quantum communication}
\label{bw-quantum-sec}

We use the unassisted entanglement-transmission model
\cite{DupuisHaydenLi2010}. For block-length $n$, the sender holds the
$A_i$ halves of two independent maximally entangled states
$\Phi_{d_i}^{R_iA_i}$, where $\dim R_i=\dim A_i=d_i$.
An arbitrary channel encodes $A_1A_2$ into the $n$ qutrit inputs.
Receiver $i$ applies a local decoding channel from $B_i^n$ to an
$d_i$-dimensional system $\widehat A_i$.
There is no preshared entanglement, feedback, or auxiliary communication.
A rate pair \((Q_1,Q_2)\) is achievable when the trace distance of the decoded joint
state from
$\Phi_{d_1}^{R_1\widehat A_1}\otimes
\Phi_{d_2}^{R_2\widehat A_2}$ tends to zero and
$\liminf_{n\to\infty}n^{-1}\log d_i\ge Q_i$.

Let
$I_c(R\rangle B)=-H(R|B)=H(B)-H(RB)$ denote the coherent information.
The following entropy bound applies to the output of any block encoder.

\begin{lemma}[Multiletter coherent-information bound]
\label{bw-q-lemma}
For every $n\ge1$, finite-dimensional $R_1,R_2$, and state
$\rho^{R_1R_2B_1^nB_2^n}$ whose output is supported on
$\mathcal S^{\otimes n}$,
\begin{equation}
 I_c(R_1\rangle B_1^n)_\rho+
 I_c(R_2\rangle B_2^n)_\rho\le n.
 \label{bw-q-multiletter}
\end{equation}
\end{lemma}

\begin{proof}
Conditional entropy is concave, so the left side of
\eqref{bw-q-multiletter} is convex in $\rho$. A pure-state decomposition
within $R_1R_2\otimes\mathcal S^{\otimes n}$ therefore reduces the proof to
a pure state $\psi$. Put
\[
 L=R_1B_1^n,\qquad D=R_2B_2^n,\qquad
 \alpha=\psi^L,\qquad \sigma=\psi^D,\qquad
 E=H(\alpha)=H(\sigma).
\]
Let $\Delta_i$ dephase $B_i^n$ in the computational basis, acting as the
identity on any other registers. Expand
\[
 \lvert\psi\rangle
 =\sum_{x\in\{0,1\}^n}\lvert x\rangle^{B_1^n}
   \lvert\psi_x\rangle^{R_1D},\qquad
 \sigma_x=\Tr_{R_1}\ketbra{\psi_x}.
\]
Thus $\sigma_x$ are subnormalized positive operators,
$\sum_x\sigma_x=\sigma$, and
\[
 \omega^{XD}=\sum_x\ketbra{x}\otimes\sigma_x
\]
is a normalized cq state. Each nonzero $\psi_x$ is pure on $R_1D$,
so its two reduced operators have the same nonzero eigenvalues. The
block-diagonal entropy formula gives
\begin{equation}
 H(\omega)=H(\Delta_1\alpha).
 \label{bw-q-branch-entropy}
\end{equation}

For $x\in\{0,1\}^n$, define a phase operator on $D$ by
\[
 Z_x=I^{R_2}\otimes
 \sum_{y\in\{0,1\}^n}(-1)^{x\cdot y}\ketbra{y}^{B_2^n}.
\]
The support assumption implies that $\sigma_x$ is supported on strings
$y$ with $x_ty_t=0$ at every coordinate $t$. Hence
$Z_x\sigma_xZ_x=\sigma_x$. Consider the channel
\[
 \mathcal T(\eta^{XD})
 =\sum_x Z_x\langle x|\eta|x\rangle^D Z_x
\]
and the uniform state $\tau^X=2^{-n}I^X$. They satisfy
\begin{equation}
 \mathcal T(\omega)=\sigma,\qquad
 \mathcal T(\tau^X\otimes\sigma)
 =2^{-n}\sum_x Z_x\sigma Z_x=\Delta_2\sigma.
 \label{bw-q-phase-average}
\end{equation}
For quantum relative entropy
$D(a\Vert b)=\Tr a(\log a-\log b)$, data processing now gives
\begin{align}
 H(\Delta_2\sigma)-E
 &=D(\sigma\Vert\Delta_2\sigma)\notag\\
 &\le D(\omega\Vert\tau^X\otimes\sigma)
 =n-H(\omega)+E.
 \label{bw-q-dpi}
\end{align}
These relative entropies are finite because
$\operatorname{supp}\sigma_x\subseteq\operatorname{supp}\sigma$ and
$\operatorname{supp}\sigma\subseteq\operatorname{supp}(\Delta_2\sigma)$.
The first equality follows by taking the trace against the
block-diagonal operator $\log(\Delta_2\sigma)$ on its support.
Combining \eqref{bw-q-branch-entropy} and \eqref{bw-q-dpi} yields
\begin{equation}
 H(\Delta_1\alpha)+H(\Delta_2\sigma)-2E\le n.
 \label{bw-q-dephased-bound}
\end{equation}

If $p_x=\Tr\sigma_x$ and
$\alpha_x=\Tr_D\ketbra{\psi_x}$, then
$H(\Delta_1\alpha)=H(p)+\sum_{x:p_x>0}p_xH(\alpha_x/p_x)
\ge H(p)$.
Moreover, $H(p)=H(\Delta_1\psi^{B_1^n})\ge H(B_1^n)_\psi$,
since dephasing cannot decrease entropy. Applying the same
block-diagonal entropy formula to $\Delta_2\sigma$ gives
$H(\Delta_2\sigma)\ge H(B_2^n)_\psi$. Thus
\[
 I_c(R_1\rangle B_1^n)_\psi+
 I_c(R_2\rangle B_2^n)_\psi
 =H(B_1^n)_\psi+H(B_2^n)_\psi-2E\le n.
\]
The convexity reduction proves the claim for mixed states as well.
\end{proof}

\begin{proposition}[Quantum capacity of the coherent Blackwell channel]
\label{bw-q-capacity}
The unassisted quantum-capacity region of $U_{\mathrm{BW}}$ is
\begin{equation}
 \mathcal Q(U_{\mathrm{BW}})
 =\{(Q_1,Q_2)\in\mathbb R_{\ge0}^2:Q_1+Q_2\le1\}.
 \label{bw-q-region}
\end{equation}
\end{proposition}

\begin{proof}
Let $\varepsilon_n$ be the joint trace-distance error of an $n$-block
code, and let $\rho$ be its state before decoding. Each decoded marginal
has error at most $\varepsilon_n$. Conditional entropy continuity
\cite{Winter2016} and data processing on the receiver give, for $i=1,2$,
\[
 I_c(R_i\rangle B_i^n)_\rho
 \ge I_c(R_i\rangle\widehat A_i)
 \ge (1-2\varepsilon_n)\log d_i-g(\varepsilon_n),
 \qquad
 g(t)=(1+t)h_2\!\left(\frac{t}{1+t}\right).
\]
The output of every encoder is supported on $\mathcal S^{\otimes n}$.
By \cref{bw-q-lemma},
\begin{equation}
 (1-2\varepsilon_n)\log(d_1d_2)
 \le n+2g(\varepsilon_n).
 \label{bw-q-finite-error}
\end{equation}
Dividing by $n$ and letting $\varepsilon_n\to0$ proves the converse,
including for mixed encodings and inputs entangled across channel uses.

The input subspace
$\operatorname{span}\{\lvert0\rangle,\lvert2\rangle\}$ transmits a
noiseless qubit to receiver $1$ while receiver $2$ receives
$\lvert0\rangle$. The subspace
$\operatorname{span}\{\lvert0\rangle,\lvert1\rangle\}$ does the same for
receiver $2$. Time sharing between these two subspace codes achieves
every nonnegative rate pair with sum at most one.
\end{proof}

Together, \cref{bw-c-capacity,bw-q-capacity} give two different capacity
regions for the same isometric channel. The confidential classical achievable rate region
contains the triangle generated by $(0,0)$, $(1,0)$, and $(0,1)$, and it also
contains the symmetric point $(\log_2\varphi,\log_2\varphi)$, whose sum
$2\log_2\varphi>1$ exceeds the quantum sum-rate capacity. When rates are
compared numerically in bits and qubits per use, respectively,
\begin{equation}
 \mathcal Q(U_{\mathrm{BW}})
 \subsetneq\mathcal C_{\mathrm{conf}}(U_{\mathrm{BW}}).
 \label{bw-capacity-separation}
\end{equation}

The full rate regions are shown in Figure~\ref{fig:blackwell-cap-regions}.

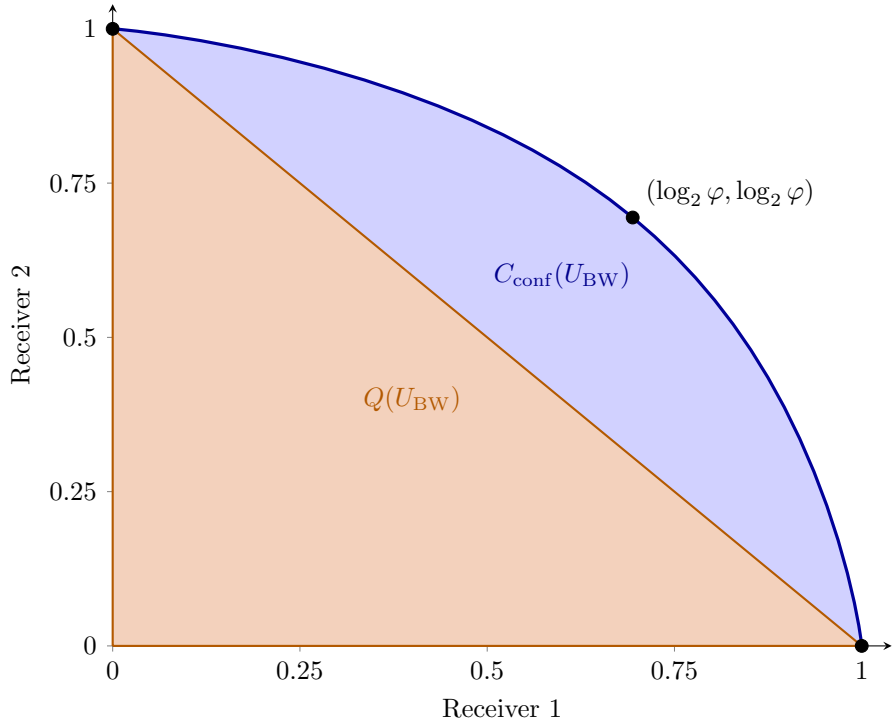
\begin{figure}[H]
\centering
\pgfplotstableread{
Rone Rtwo
0.00000000 1.00000000
0.00001950 0.99999905
0.00133473 0.99990311
0.01019904 0.99902332
0.03299736 0.99609139
0.07009533 0.99013835
0.11771578 0.98087009
0.17128661 0.96849020
0.22707328 0.95342826
0.28252238 0.93615526
0.33608943 0.91709127
0.38695588 0.89657095
0.43478081 0.87483837
0.47951845 0.85205394
0.52129458 0.82830517
0.56032658 0.80361716
0.59687317 0.77796158
0.63120363 0.75126378
0.66357919 0.72340808
0.69424191 0.69424191
0.72340808 0.66357919
0.75126378 0.63120363
0.77796158 0.59687317
0.80361716 0.56032658
0.82830517 0.52129458
0.85205394 0.47951845
0.87483837 0.43478081
0.89657095 0.38695588
0.91709127 0.33608943
0.93615526 0.28252238
0.95342826 0.22707328
0.96849020 0.17128661
0.98087009 0.11771578
0.99013835 0.07009533
0.99609139 0.03299736
0.99902332 0.01019904
0.99990311 0.00133473
0.99999905 0.00001950
1.00000000 0.00000000
}\BlackwellBoundary

\begin{tikzpicture}
\begin{axis}[
  width=0.72\linewidth,
  height=0.61\linewidth,
  axis lines=left,
  xmin=0, xmax=1.04,
  ymin=0, ymax=1.04,
  xlabel={Receiver 1},
  ylabel={Receiver 2},
  xtick={0,0.25,0.5,0.75,1},
  ytick={0,0.25,0.5,0.75,1},
  tick label style={font=\small},
  label style={font=\small},
  clip=false
]
\addplot[
  name path=BlackwellCurve,
  draw=none
] table[x=Rone,y=Rtwo] {\BlackwellBoundary};

\addplot[
  name path=BlackwellBase,
  draw=none
] coordinates {(0,0) (1,0)};

\addplot[
  blue!18
] fill between[
  of=BlackwellCurve and BlackwellBase
];

\addplot[
  fill=orange!35,
  fill opacity=0.75,
  draw=orange!70!black,
  thick
] coordinates {
  (0,0)
  (1,0)
  (0,1)
  (0,0)
};

\addplot[
  very thick,
  blue!60!black,
  line join=round
] table[x=Rone,y=Rtwo] {\BlackwellBoundary};


\addplot[
  only marks,
  mark=*,
  mark size=2.4pt,
  black
] coordinates {(0,1) (1,0) (0.69424191,0.69424191)};

\node[
  anchor=south west,
  font=\small,
  inner sep=2pt
] at (axis cs:0.704,0.704)
  {$(\log_2\varphi,\log_2\varphi)$};
\node[
  font=\small,
  text=orange!70!black
] at (axis cs:0.4,0.4)
  {$Q(U_{\mathrm{BW}})$};
\node[
  font=\small,
  text=blue!55!black
] at (axis cs:0.6,0.6)
  {$C_{\mathrm{conf}}(U_{\mathrm{BW}})$};
\end{axis}
\end{tikzpicture}
\caption{Quantifying rates numerically in bits per use for the confidential classical capacity and in qubits per use for the quantum capacity, we plot the two capacity regions for the Blackwell isometry.}
\label{fig:blackwell-cap-regions}
\end{figure}

\subsection{The Platypus isometry}
\label{sec:plat-example}

For \(0\le s\le\tfrac12\), consider the isometry \cite{LeditzkyEtAl2023}
\begin{align}
 F_s|0\rangle
 &=\sqrt{s}\,|00\rangle+\sqrt{1-s}\,|11\rangle,\nonumber\\
 F_s|1\rangle&=|20\rangle,\qquad
 F_s|2\rangle=|21\rangle,
 \label{eq:plat-isometry}
\end{align}
with a qutrit input, a qutrit output $B_1$, and a qubit output $B_2$.
Write \(\mathcal N_s\) and \(\mathcal N_s^c\) for its marginals to
\(B_1\) and \(B_2\), respectively. The restriction to
\(s\le\tfrac12\) loses no generality: \(F_s\) and \(F_{1-s}\) are related
by swaps of the input basis vectors \(|1\rangle,|2\rangle\) and of the two
output computational bases. Because \(F_s\) is itself an isometry, the additional environment \(E\) in the quantum-input model is one-dimensional.

At $s=0$, this isometry is equivalent to the Blackwell isometry on its
output support. To see this, interchange the input labels $0$ and $2$,
relabel the occupied $B_1$ basis vectors $|2\rangle,|1\rangle$ as
$|0\rangle,|1\rangle$, and swap the $B_2$ basis vectors. Thus the
Blackwell results also determine both capacity regions at $s=0$.

\begin{proposition}[Confidential-region bounds]
\label{prop:plat-private-bounds}
For every \(0\le s\le\tfrac12\),
\begin{equation}
 \operatorname{conv}\{(0,0),(1,0),(0,1)\}
 \subseteq \mathcal C_{\mathrm{conf}}(F_s)
 \subseteq [0,1]^2.
 \label{eq:plat-triangle-square}
\end{equation}
The inner inclusion is attained by zero-error, perfectly secret one-use codes
and time sharing.
\end{proposition}

\begin{proof}
For the point \((1,0)\), use the two equiprobable input states
\begin{equation}
 \sigma_0=|0\rangle\!\langle0|,
 \qquad
 \sigma_1=s|1\rangle\!\langle1|
              +(1-s)|2\rangle\!\langle2|.
 \label{eq:plat-private-code-one}
\end{equation}
Receiver~1 obtains
\begin{align}
 \mathcal N_s(\sigma_0)
 &=s|0\rangle\!\langle0|+(1-s)|1\rangle\!\langle1|,\nonumber\\
 \mathcal N_s(\sigma_1)&=|2\rangle\!\langle2|,
 \label{eq:plat-private-code-one-bob}
\end{align}
which have orthogonal supports. Receiver~2 obtains the same state for the
two messages:
\begin{equation}
 \mathcal N_s^c(\sigma_0)=\mathcal N_s^c(\sigma_1)
 =s|0\rangle\!\langle0|+(1-s)|1\rangle\!\langle1|.
 \label{eq:plat-private-code-one-eve}
\end{equation}
Thus one confidential bit is transmitted to receiver~1 with zero error
and perfect secrecy. For the point \((0,1)\), use the equiprobable inputs
\(|1\rangle\!\langle1|\) and \(|2\rangle\!\langle2|\). Receiver~2
obtains the orthogonal states \(|0\rangle\!\langle0|\) and
\(|1\rangle\!\langle1|\), while receiver~1 obtains
\(|2\rangle\!\langle2|\) in both cases. Time sharing gives the left
inclusion in \eqref{eq:plat-triangle-square}.

Write $C$, $P$, and $Q$ for point-to-point classical, private classical,
and quantum capacity, respectively. The Platypus marginals satisfy
\begin{equation}
 P(\mathcal N_s)=C(\mathcal N_s)=1,
 \qquad
 P(\mathcal N_s^c)=Q(\mathcal N_s^c)=C(\mathcal N_s^c)=1
 \label{eq:plat-marginal-capacities}
\end{equation}
\cite{LeditzkyEtAl2023}. Given any broadcast code, average its stochastic
encoder uniformly over the other message. For receiver \(i\), this is a
point-to-point classical code for the corresponding marginal channel with
the same message size and the same average decoding error. Hence
\(R_1\le C(\mathcal N_s)=1\) and
\(R_2\le C(\mathcal N_s^c)=1\), proving the right inclusion.
\end{proof}

The coding theorem also provides inner bounds from finite ensembles of
qutrit states. Each family \(x\mapsto\sigma_x\) induces the cq broadcast
channel
\begin{equation}
 x\longmapsto F_s\sigma_xF_s^\dagger,
 \label{eq:plat-induced-cq-channel}
\end{equation}
to which \cref{thm:main} applies directly. Optimizing over these families
and their auxiliary distributions provides a way to investigate whether rates beyond time sharing are achievable.

For quantum communication, we use the unassisted entanglement-transmission
model of \cref{bw-quantum-sec}. Any simultaneous code with rates
$(Q_1,Q_2)$ gives a point-to-point code for the first marginal: initialize
the second source input in its maximally mixed state, discard its
reference, and retain the first receiver's decoder. The error for the
first task cannot increase. Consequently, $Q_1\le Q(\mathcal N_s)$.

\begin{proposition}[A Platypus rate-pair separation]
\label{prop:plat-private-quantum-separation}
For \(0<s\le\tfrac12\), set
\begin{equation}
 u_s:=\log\!\left(1+\sqrt{1-s}\right)<1.
 \label{eq:plat-quantum-upper}
\end{equation}
For every \(\lambda\in(u_s,1)\), the confidential classical rate pair
\begin{equation}
 (R_1,R_2)=(\lambda,1-\lambda)
 \label{eq:plat-private-separated-pair}
\end{equation}
is achievable, whereas the same pair is not achievable for simultaneous
quantum communication.
\end{proposition}

\begin{proof}
The confidential classical pair follows by time sharing between the two
one-use codes in the proof of \cref{prop:plat-private-bounds}.
On the other hand, every simultaneous quantum code gives a point-to-point
quantum code for the first marginal, and hence \(Q_1\le Q(\mathcal N_s)\).
The upper bound
\begin{equation}
 Q(\mathcal N_s)\le
 \log\!\left(1+\sqrt{1-s}\right)=u_s
 \label{eq:plat-q-bound-used}
\end{equation}
was proved in \cite{LeditzkyEtAl2023}. Since \(\lambda>u_s\), the quantum
rate pair in \eqref{eq:plat-private-separated-pair} is impossible.
\end{proof}

For every $0<s\le\tfrac12$, the proposition gives confidential classical achievable
rate pairs with both coordinates positive that lie outside the quantum
region. Determining the full confidential classical and quantum-capacity regions for these parameters would sharpen this comparison.

\section{Discussion}

The normalized likelihood selector gives a common encoding law for
reliability and secrecy.  Its planted representation exposes the
independence needed for fine-index packing.  Its two posterior
orientations instead make one selected correlation index uniform, allowing a
one-sided resolvability argument for the other codebook.  Comparing both
orientations with the global selector gives the bipartite cq
resolvability bound.  These ingredients establish the quantum-output
counterpart of the classical confidential-message inner bound under
conditional strong secrecy.

The auxiliary-alphabet reduction makes this cq inner bound a
finite-dimensional optimization.  Its proof uses entropy concavity to
reduce the supports while preserving or increasing both the averaged rate
bounds.  The multiletter capacity formulas follow by applying the coding
theorem to blocks, with arbitrary quantum preparations allowed in the
receiver--environment secrecy model.  Finding conditions under which
these formulas become single-letter remains a natural question.
A finite-block bipartite cq resolvability bound for the likelihood
selector would also extend the present analysis beyond its asymptotic
form.

The Blackwell comparison shows why the input model matters.  The
classical channel specifies outputs for classical symbols, whereas its
coherent isometric extension also transmits superpositions.  For this
extension, a multiletter entropy bound determines the quantum-capacity
region, while basis encoding already achieves confidential classical
rate pairs outside that region.  The Platypus example gives another
separation at specified rate pairs. 

\section*{Acknowledgements}

This work was supported by the Natural Sciences and Engineering
Research Council of Canada (NSERC)
(ALLRP-586858-2023, RGPIN-2025-02094, ALLRP-578455-2022).
Nous remercions le Conseil de recherches en sciences naturelles
et en g\'enie du Canada (CRSNG) de son soutien
(ALLRP-586858-2023, RGPIN-2025-02094, ALLRP-578455-2022). F.S.'s work is supported by the European Commission through a
Marie Sk{\l}odowska-Curie Global Fellowship. He gratefully acknowledges
the hospitality of the Institute for Quantum Computing (IQC),
University of Waterloo, and Freie Universit\"at Berlin.
AI was used during the manuscript's editing and the development and debugging of the code used for our numerical findings. All AI-assisted material was reviewed and verified by the authors, who take full responsibility for the results and conclusions.

\begingroup
\small
\bibliographystyle{unsrt}
\bibliography{likelihood_marton_confidential}
\endgroup

\appendix
\section{Proof of the memoryless cq soft-covering lemma}
\label{app:cq-soft-covering}
\label{appn:cq-soft-covering}
Delete letters outside \(\supp q_X\), and restrict the output Hilbert space
to \(\supp\bar\zeta\).  This is without loss of generality.  Indeed,
if \(\lvert\psi\rangle\in\ker\bar\zeta\), then
\(0=\sum_xq_X(x)\langle\psi|\zeta_x|\psi\rangle\).  Positivity implies
\(\zeta_x|\psi\rangle=0\) for every \(x\) with \(q_X(x)>0\), and hence
\(\supp\zeta_x\subseteq\supp\bar\zeta\).  Strong-typical sequences below
therefore contain no zero-probability letters, and all nonzero eigenvalues
relevant to the estimates lie in fixed finite supports.

Write
\[
 H(Z)=H(\bar\zeta),\qquad H(Z|X)=\sum_xq_X(x)H(\zeta_x),
 \qquad I(X;Z)=H(Z)-H(Z|X).
\]
Choose \(\delta>0\) with \(R>I(X;Z)+4\delta\).  Standard typical-subspace
estimates \cite{Wilde2017} give a spectral typical projector \(\Pi\) of
\(\bar\zeta^{\otimes n}\) obeying, for some \(a>0\),
\begin{align}
 \Tr[\Pi\bar\zeta^{\otimes n}]&\ge1-2^{-an},
 \label{eq:soft-average-success}\\
 \rank\Pi&\le2^{n[H(Z)+\delta]}.
 \label{eq:soft-average-rank}
\end{align}
Choose a strong-typicality tolerance \(\tau>0\), smaller than
\(\frac12\min_{x:q_X(x)>0}q_X(x)\), such that every \(\tau\)-typical type
\(\widehat q\) satisfies
\[
 \left|\sum_x \left[ \widehat q(x)-q_X(x)\right] H(\zeta_x)\right|\le\delta.
\]
Choose the spectral typicality widths for the finitely many states
\(\zeta_x\) so that their total eigenvalue slack is at most \(\delta\).
For a \(\tau\)-typical \(x^n\), permute tensor factors so that equal letters
are contiguous, diagonalize each \(\zeta_x\), and take the tensor product of
the corresponding spectral typical projectors over the resulting letter
blocks.  Use the identity on a zero-count block and undo the permutation.
Denote the resulting projector by \(\Pi_{x^n}\).
Finite-alphabet large-deviation bounds for the letter counts and conditional
eigenvalue sequences, after decreasing \(a\) if necessary, give estimates
uniform over
\(x^n\in\cT_\tau^{(n)}\), such that
\begin{align}
 q_X^{\otimes n}(\cT_\tau^{(n)})&\ge1-2^{-an},
 \label{eq:soft-classical-success}\\
 \Tr[\Pi_{x^n}\zeta_{x^n}]&\ge1-2^{-an},
 \label{eq:soft-conditional-success}\\
 \Pi_{x^n}\zeta_{x^n}\Pi_{x^n}
 &\le2^{-n[H(Z|X)-2\delta]}\Pi_{x^n}.
 \label{eq:soft-conditional-bound}
\end{align}
The last two statements hold for \(x^n\in\cT_\tau^{(n)}\).  They follow by
multiplying eigenvalues within each conditional block.  No simultaneous
diagonalization of the different \(\zeta_x\)'s is used.

For \(x^n\in\cT_\tau^{(n)}\), set
\[
 \widetilde\zeta_{x^n}=\Pi_{x^n}\zeta_{x^n}\Pi_{x^n},
 \qquad
 \zeta'_{x^n}=\Pi\widetilde\zeta_{x^n}\Pi,
\]
and set both operators to zero for atypical \(x^n\).  The gentle measurement
lemma gives the explicit estimate
\begin{equation}
 \begin{split}
 \E_{X^n}\lVert\zeta_{X^n}-\widetilde\zeta_{X^n}\rVert_1
 &\le q_X^{\otimes n}((\cT_\tau^{(n)})^c)\\
 &\quad+2\E\!\left[\one\{X^n\in\cT_\tau^{(n)}\}
 \sqrt{1-\Tr[\Pi_{X^n}\zeta_{X^n}]}\right]
 \le2^{-a_1n}
 \end{split}
 \label{eq:soft-first-truncation}
\end{equation}
for some \(a_1>0\), where the exponent may be decreased from line to line, and $(\cT_\tau^{(n)})^c$ denotes the complement of $\cT_\tau^{(n)}$, which is the set of sequences that are not $\tau$-typical.
Also,
\begin{align}
 \E\Tr[(I-\Pi)\widetilde\zeta_{X^n}]
 &\le\Tr[(I-\Pi)\bar\zeta^{\otimes n}]
 +\E\lVert\zeta_{X^n}-\widetilde\zeta_{X^n}\rVert_1
 \le2^{-a_2n}.
 \label{eq:soft-outer-failure}
\end{align}
Gentle measurement for subnormalized states and Jensen's inequality now give
\begin{align}
 \E\lVert\widetilde\zeta_{X^n}
       -\Pi\widetilde\zeta_{X^n}\Pi\rVert_1
 &\le2\E\sqrt{\Tr[(I-\Pi)\widetilde\zeta_{X^n}]}\nonumber\\
 &\le2\sqrt{\E\Tr[(I-\Pi)\widetilde\zeta_{X^n}]}
 \le2^{-a_3n}.
 \label{eq:soft-second-truncation}
\end{align}
Combining the two truncations yields a sequence
\(\eta_n\le2^{-a_4n}\) such that
\begin{equation}
 \E\lVert\zeta_{X^n}-\zeta'_{X^n}\rVert_1\le\eta_n.
 \label{eq:soft-total-truncation}
\end{equation}

Put \(\lambda_n:=2^{-n[H(Z|X)-2\delta]}\).  For typical \(x^n\),
\[
 \Pi\widetilde\zeta_{x^n}\Pi
 \le\lambda_n\Pi\Pi_{x^n}\Pi\le\lambda_n\Pi.
\]
For atypical \(x^n\), \(\zeta'_{x^n}=0\), so the same final bound is
trivial.  Hence every \(\zeta'_{x^n}\) is supported on \(\Pi\), and
\begin{equation}
 0\le\zeta'_{x^n}\le\lambda_n\Pi.
 \label{eq:soft-truncated-bound}
\end{equation}
Let
\[
 \widehat\zeta'_\cC=\frac1{L_n}\sum_\ell\zeta'_{X^n(\ell)},
 \qquad \bar\zeta'_n=\E\zeta'_{X^n}.
\]
For a Hermitian operator supported on a \(D\)-dimensional space,
\(\lVert T\rVert_1\le\sqrt D\lVert T\rVert_2\).  Independence and
centering of the codeword terms give
\begin{align}
 \E\lVert\widehat\zeta'_\cC-\bar\zeta'_n\rVert_1
 &\le\sqrt{\frac{\rank\Pi}{L_n}
 \left(\E\Tr[(\zeta'_{X^n})^2]-\Tr[(\bar\zeta'_n)^2]\right)}
 \nonumber\\
 &\le\sqrt{\frac{\rank\Pi}{L_n}
 \E\Tr[(\zeta'_{X^n})^2]}.
 \label{eq:soft-second-moment}
\end{align}
Because \(\Tr\zeta'_{x^n}\le1\) and
\(\zeta'_{x^n}\le\lambda_n\Pi\),
\(\Tr[(\zeta'_{x^n})^2]\le\lambda_n\).  Combining this with
\eqref{eq:soft-average-rank} gives
\begin{equation}
 \E\lVert\widehat\zeta'_\cC-\bar\zeta'_n\rVert_1
 \le2^{-\frac n2[R-I(X;Z)-3\delta]}.
 \label{eq:soft-second-moment-exponent}
\end{equation}
Finally, convexity, \eqref{eq:soft-total-truncation}, and the triangle
inequality yield
\[
 \E\lVert\widehat\zeta_\cC-\bar\zeta^{\otimes n}\rVert_1
 \le2\eta_n+2^{-\frac n2[R-I(X;Z)-3\delta]}.
\]
Choose
\[
 0<c<\min\!\left\{a_4,
 \frac{R-I(X;Z)-3\delta}{2}\right\}.
\]
After increasing \(n_0\) to absorb the fixed prefactors, division by two
gives \eqref{eq:cq-soft-covering}.
\end{document}